\documentclass{lmcs}
\pdfoutput=1

\usepackage{lastpage}
\lmcsdoi{16}{1}{11}
\lmcsheading{}{\pageref{LastPage}}{}{}%
{Aug.~17,~2018}{Feb.~13,~2020}{}


\usepackage{stmaryrd}
\usepackage{amssymb}
\usepackage{proof}
\usepackage[all]{xy}

\newcommand{\eqvl}{\leftrightarrow}

\newcommand{\reduces}{\ensuremath{\to^*}}
\newcommand{\infred}{\ensuremath{\to^\infty}}

\newcommand{\valuation}[3]{\ensuremath{\llbracket#1\rrbracket_{#2}^{#3}}}
\newcommand{\proves}{\ensuremath{\vdash}}
\newcommand{\To}{\ensuremath{\Rightarrow}}

\newcommand{\Nbb}{\ensuremath{\mathbb{N}}}

\newcommand{\Pow}{\ensuremath{\mathcal{P}}}

\newcommand{\ra}{\ensuremath{\rangle}}

\newcommand{\dom}{\ensuremath{\mathrm{dom}}}
\newcommand{\codom}{\ensuremath{\mathrm{codom}}}

\newcommand{\Nat}{\ensuremath{\mathrm{Nat}}}
\newcommand{\List}{\ensuremath{\mathrm{List}}}
\newcommand{\Strm}{\ensuremath{\mathrm{Strm}}}
\newcommand{\Tree}{\ensuremath{\mathrm{Tree}}}
\newcommand{\FTree}{\ensuremath{\mathrm{FTree}}}
\newcommand{\BTree}{\ensuremath{\mathrm{BTree}}}
\newcommand{\node}{\ensuremath{\mathtt{node}}}
\newcommand{\fnode}{\ensuremath{\mathtt{fnode}}}
\newcommand{\bnode}{\ensuremath{\mathtt{bnode}}}
\newcommand{\tail}{\ensuremath{\mathtt{tl}}}
\newcommand{\head}{\ensuremath{\mathtt{hd}}}
\newcommand{\mput}{\ensuremath{\mathtt{put}}}
\newcommand{\get}{\ensuremath{\mathtt{get}}}
\newcommand{\odd}{\ensuremath{\mathtt{odd}}}
\newcommand{\out}{\ensuremath{\mathtt{out}}}
\newcommand{\run}{\ensuremath{\mathtt{run}}}
\newcommand{\runi}{\ensuremath{\mathtt{runi}}}
\newcommand{\SP}{\ensuremath{\mathrm{SP}}}
\newcommand{\SPi}{\ensuremath{\mathrm{SPi}}}

\newcommand{\Odd}{\ensuremath{\mathrm{Odd}}}
\newcommand{\Even}{\ensuremath{\mathrm{Even}}}

\newcommand{\Tb}{\ensuremath{{\mathbb{T}}}}

\newcommand{\Tc}{\ensuremath{{\mathcal T}}}

\newcommand{\Sc}{\ensuremath{{\mathcal S}}}

\newcommand{\Cc}{\ensuremath{{\mathcal C}}}
\newcommand{\Dc}{\ensuremath{{\mathcal D}}}

\newcommand{\Uc}{\ensuremath{{\mathcal U}}}
\newcommand{\Vc}{\ensuremath{{\mathcal V}}}

\newcommand{\Ys}{\ensuremath{{\mathsf Y}}}

\newcommand{\nil}{\ensuremath{{\mathrm{nil}}}}
\newcommand{\cons}{\ensuremath{{\mathrm{cons}}}}
\newcommand{\case}{\ensuremath{{\mathrm{case}}}}
\newcommand{\fix}{\ensuremath{{\mathrm{fix}\,}}}
\newcommand{\cofix}{\ensuremath{{\mathrm{cofix}\,}}}

\newcommand{\app}{\ensuremath{{\mathrm{app}}}}
\newcommand{\lam}{\ensuremath{{\mathrm{lam}}}}

\newcommand{\rval}[1]{\ensuremath{\llbracket#1\rrbracket}}

\newcommand{\FV}{\mathrm{TV}}
\newcommand{\SV}{\mathrm{SV}}
\newcommand{\FSV}{\mathrm{FSV}}

\newcommand{\tgt}{\mathrm{tgt}}
\newcommand{\chgtgt}{\mathrm{chgtgt}}

\newcommand{\Ind}{\mathrm{Ind}}
\newcommand{\CoInd}{\mathrm{CoInd}}
\newcommand{\IndCoInd}{\mathrm{(Co)Ind}}
\newcommand{\erase}[1]{\ensuremath{|#1|}}

\newcommand{\Constr}{\mathrm{Constr}}
\newcommand{\Def}{\mathrm{Def}}
\newcommand{\ArgTypes}{\mathrm{ArgTypes}}

\title{An operational interpretation of coinductive types}

\author[{\L}.~Czajka]{{\L}ukasz Czajka}
\address{TU Dortmund University, Dortmund, Germany}
\email{lukaszcz@mimuw.edu.pl}
\thanks{Supported by the European Union's Horizon 2020 research and
  innovation programme under the Marie Sk{\l}odowska-Curie grant
  agreement number~704111.}

\begin{document}

\begin{abstract}
  \noindent We introduce an operational rewriting-based semantics for
  strictly positive nested higher-order (co)inductive types. The
  semantics takes into account the ``limits'' of infinite reduction
  sequences. This may be seen as a refinement and generalization of
  the notion of productivity in term rewriting to a setting with
  higher-order functions and with data specified by nested higher-order
  inductive and coinductive definitions. Intuitively, we interpret
  lazy data structures in a higher-order functional language by
  potentially infinite terms corresponding to their complete
  unfoldings.

  We prove an approximation theorem which essentially states that if a
  term reduces to an arbitrarily large finite approximation of an
  infinite object in the interpretation of a coinductive type, then it
  infinitarily (i.e.~in the ``limit'') reduces to an infinite object
  in the interpretation of this type. We introduce a sufficient
  syntactic correctness criterion, in the form of a type system, for
  finite terms decorated with type information. Using the
  approximation theorem, we show that each well-typed term has a
  well-defined interpretation in our semantics.
\end{abstract}

\maketitle

\section{Introduction}\label{sec_intro}

It is natural to consider an interpretation of coinductive types where
the elements of a coinductive type~$\nu$ are possibly infinite
terms. Each finite term of type~$\nu$ containing fixpoint operators
then ``unfolds'' to a possibly infinite term without fixpoint
operators in the interpretation of~$\nu$. For instance, one would
interpret the type of binary streams as the set of infinite terms of
the form $b_1 :: b_2 :: \ldots$ where $b_1 \in \{0,1\}$ and~$::$ is an
infix notation for the stream constructor. Then any fixpoint
definition of a term of this type should ``unfold'' to such an
infinite term. This kind of interpretation corresponds closely to a
naive understanding of infinite objects and coinductive types.

This paper is devoted to a study of such an interpretation in the
context of infinitary rewriting. Infinitary rewriting extends term
rewriting by infinite terms and transfinite reductions. This enables
the consideration of ``limits'' of terms under infinite reduction
sequences.

We consider a combination of simple function types with strictly
positive nested higher-order inductive and coinductive types. An
example of a higher-order coinductive type is the type of trees with
potentially infinite branches and two kinds of nodes: nodes with a
list of finitely many children and nodes with infinitely many children
specified by a function on natural numbers. In our notation this type
may be represented as the coinductive definition $\Tree_2 = \CoInd\{
c_1 : \List(\Tree_2) \to \Tree_2,\, c_2 : (\Nat \to \Tree_2) \to
\Tree_2 \}$ which intuitively specifies that each element of~$\Tree_2$
is a possibly infinite term which has one of the forms:
\begin{itemize}
\item $c_1 (t_1 :: t_2 :: \ldots :: t_n :: \nil)$ where each $t_i$ is
  an element of~$\Tree_2$ and $::$ is the finite list constructor, or
\item $c_2 f$ where $f$ is a term which represents a function from
  $\Nat$ to $\Tree_2$.
\end{itemize}

\noindent
We interpret each type~$\tau$ as a subset~$\valuation{\tau}{}{}$ of
the set~$\Tb^\infty$ of finite and infinite terms. This interpretation
may be seen as a refinement and generalization of the notion of
productivity in term rewriting to a setting with higher-order
functions and more complex (co)inductive data structures. From a
programming language perspective, we essentially interpret lazy data
structures in a higher-order functional language by potentially
infinite terms corresponding to their complete unfoldings (i.e.~their
``limits'' under infinite reductions).

For example, the interpretation~$\valuation{\Strm}{}{}$ of the
coinductive type~$\Strm$ of streams of natural numbers with a single
constructor $\cons : \Nat \to \Strm \to \Strm$ consists of all
infinite terms of the form $\cons \, n_0 (\cons\, n_1 (\ldots))$ where
$n_k \in \valuation{\Nat}{}{}$ for $k \in \Nbb$. The
interpretation~$\valuation{\Strm\to\Strm}{}{}$ of an arrow type $\Strm
\to \Strm$ is the set of all terms~$t$ such that for every $u \in
\valuation{\Strm}{}{}$ there is~$u' \in \valuation{\Strm}{}{}$ with $t
u \infred u'$, where~$\infred$ denotes the infinitary reduction
relation (so~$u'$ is the ``limit'' of a reduction starting with~$t
u$). This means that~$t$ is productive -- it computes (in the limit) a
stream when given a stream as an argument, producing any initial
finite segment of the result using only an initial finite segment of
the argument. Note that the argument~$u$ is just any infinite stream
of natural numbers -- it need not even be computable. This corresponds
with the view that arguments to a function may come from an outside
``environment'' about which nothing is assumed, e.g., the argument may
be a stream of requests for an interactive program.

One could informally argue that including infinite objects explicitly
is not necessary, because it suffices to consider finite
``approximations''~$u_n$ of ``size''~$n$ of an infinite argument
object~$u$ (which itself is possibly not computable), and if~$t u_n$
reduces to progressively larger approximations of an infinite object
for progressively larger~$n$, then this ``defines'' the application
of~$t$ to~$u$, because to compute any finite part of the result it
suffices to take a sufficiently large approximation as an argument. We
actually make this intuition precise in the framework of infinitary
rewriting. We show that if for every approximation~$u_n$ of size~$n$
of an infinite object~$u$ the application~$t u_n$ reduces to an
approximation of an infinite object of the right type, with the result
approximations getting larger as~$n$ gets larger, then there is a
reduction starting from~$t u$ which ``in the limit'' produces an
infinite object of the right type. For nested higher-order
(co)inductive types this result turns out to be non-trivial.

The result mentioned above actually follows from the approximation
theorem which is the central technical result of this paper. It may be
stated as follows: if $t \infred t_n \in \valuation{\nu}{}{n}$ for
each $n \in \Nbb$ then there is~$t'$ with
$t \infred t' \in \valuation{\nu}{}{}$, where~$\nu$ is a coinductive
type and~$\valuation{\nu}{}{n}$ is the set of approximations of
size~$n$ of the (typically infinite) objects of type~$\nu$ (i.e.~of
the terms in~$\valuation{\nu}{}{}$).

In the second part of the paper we consider \emph{finite} terms
decorated with type annotations. We present a type system which gives
a sufficient syntactic correctness criterion for such terms. The
system enables reasoning about sizes of (co)inductive types, similarly
as in systems with sized types. Using the approximation theorem we
show soundness: if a finite decorated term~$t$ may be assigned
type~$\tau$ in our type system, then there is~$t' \in
\valuation{\tau}{}{}$ such that $\erase{t} \infred t'$,
where~$\erase{t}$ denotes the term~$t$ with type decorations
erased. This means that every typable term~$t$ has a well-defined
interpretation in the corresponding type, which may be obtained as a
limit of a reduction sequence starting from~$\erase{t}$.

Our definition of the rewriting semantics is natural and relatively
straightforward. It is not difficult to prove it sound for a
restricted form of non-nested first-order (co)inductive
types. However, once we allow parameterized nested higher-order
inductive and coinductive types significant complications occur
because of the alternation of least and greatest fixpoints in the
definitions. Our main technical contribution is the proof of the
approximation theorem. This proof involves some heavy infinitary
rewriting machinery, but just to apply the theorem no deep familiarity
with infinitary rewriting is needed.

The main purpose of this paper is to define an infinitary rewriting
semantics, to precisely state and prove the approximation theorem, and
to show that the approximation theorem may be used to derive soundness
of the rewriting semantics for systems based on sized types. The type
system itself presented in the second part of the paper is not a
significant improvement over the state-of-the-art in type systems
based on sized types. It is mostly intended as an illustration of a
system for which our rewriting semantics is particularly perspicuous.

\subsection{Related work}

The notion of productivity dates back to the work of
Dijkstra~\cite{Dijkstra1980}, and the later work of
Sijtsma~\cite{Sijtsma1989}. Our rewriting semantics may be considered
a generalization of Isihara's definition of productivity in
algorithmic systems~\cite{Isihara2008}, of Zantema's and
Raffelsieper's definition of productivity in infinite data
structures~\cite{ZantemaRaffelsieper2010}, and of the definition of
stream
productivity~\cite{Endrullis2010,EndrullisGrabmayerHendriks2008,EndrullisHendriks2011}. In
comparison to our setting, the infinite data structures considered
before in term rewriting literature are very simple. None of the
papers mentioned allow higher-order functions or higher-order
(co)inductive types. The relative difficulty of our main results stems
from the fact that the data structures we consider may be much more
complex.

Infinitary rewriting was introduced
in~\cite{KennawayKlopSleepVries1995,KennawayKlopSleepVries1995b,KennawayKlopSleepVries1997}. See~\cite{KennawayVries2003}
for more references and a general introduction.

In the context of type theory, infinite objects were studied by
Martin-L{\"o}f~\cite{MartinLof1988} and
Coquand~\cite{Coquand1993}. Gimenez~\cite{Gimenez1994} introduced the
guardedness condition to incorporate coinductive types and corecursion
into dependent type theory, which is the approach currently used
in~Coq. Sized types are a long-studied approach for ensuring
termination and productivity in type
theories~\cite{HughesParetoSabry1996,BartheFadeGimenezPintoUustalu2004,Abel2006,AbelPientka2016}. In
comparison to previous work on sized types, the type system introduced
in the second part of this paper is not a significant advance, but as
mentioned before this is not the point of the present work. In order
to justify the correctness of systems with sized types, usually strong
normalization on typable terms is shown for a restriction of the
reduction relation. We provide an infinitary rewriting semantics. Our
approach may probably be extended to provide an infinitary rewriting
semantics for at least some of the systems from the type theory
literature. This semantics is interesting in its own right.

In~\cite{SeveriVries2012a} infinitary weak normalization is proven for
a broad class of Pure Type Systems extended with corecursion on
streams (CoPTSs), which includes Krishnaswami and Benton's typed
$\lambda$-calculus of reactive
programs~\cite{KrishnaswamiBenton2011}. This is related to our work in
that it provides some infinitary rewriting interpretation for a class
of type systems. The formalism of~CoPTSs is not based on sized types,
but on a modal \emph{next} operator, and it only supports the
coinductive type of streams.

Our work is also related to the work on computability at higher
types~\cite{Longley2000}, but we have not yet investigated the precise
relationships.

Coinduction has been studied from a more general coalgebraic
perspective~\cite{JacobsRutten2011}. In this paper we use a few simple
proofs by coinduction and one definition by corecursion. Formally,
they could be justified as in
e.g.~\cite{KozenSilva2017,Sangiorgi2012,JacobsRutten2011,Czajka2018}. Our
use of coinduction in this paper is not very involved, and there are
no implicit corecursive function definitions like
in~\cite{Czajka2018}.

\section{Infinitary rewriting}\label{sec_rewriting}

In this section we define infinitary terms and reductions. We assume
familiary with the lambda calculus~\cite{Barendregt1984} and basic
notions such as $\alpha$-conversion, substitution, etc. Prior
familiarity with infinitary rewriting or infinitary lambda
calculus~\cite{KennawayVries2003,KennawayKlopSleepVries1997} is not
necessary but is helpful.

We assume a countable set~$\Vc$ of \emph{variables}, and a countable
set~$\Cc$ of \emph{constructors}. The set~$\Tb^\infty$ of all finite
and infinite \emph{terms}~$t$ is given by
\[
  \begin{array}{rcl}
    t &::=& x \mid c \mid \lambda x . t \mid t t \mid \case(t; \{ c_k
            \vec{x} \To t_k \mid k = 1,\ldots,n \})
  \end{array}
\]
where~$x \in \Vc$ and~$c,c_k \in \Cc$. We use the notation $\vec{t}$
(resp.~$\vec{x}$) to denote a sequence of terms (resp.~variables) of
an unspecified length.

More precisely, the set $\Tb^\infty$ is defined as an appropriate
metric completion (analogously to~\cite{KennawayVries2003}), but the
above specification is clear and the details of the definition are not
significant for our purposes. We consider terms modulo
$\alpha$-conversion. Below (Definition~\ref{def_icrs}) we will present
the terms together with the rewrite rules as an
iCRS~\cite{KetemaSimonsen2011}, which may be considered a formal
definition of our rewrite system.

There are the following reductions:
\[
\begin{array}{rcl}
  (\lambda x . t) t' &\to_\beta& t[t'/x] \\
  \case(c_k \vec{u}; \{c_l \vec{x} \To t_l\}) &\to_\iota& t_k[\vec{u}/\vec{x}]
\end{array}
\]
In the $\iota$-rule we require that the appropriate
sequences~$\vec{u}$ and~$\vec{x}$ have the same lengths, all variables
in each~$\vec{x}$ are pairwise distinct, and the constructors~$c_l$
are all distinct. For instance, $\case(c t_1 t_2; \{ c x y \To x,\, d
x y \To y\}) \to_\iota t_1$ (assuming $c \ne d$), but $\case(c t_1; \{
c x y \To x,\, d x y \To y\})$, $\case(c' t_1 t_2; \{ c x y \To x,\, d
x y \To y\})$ and $\case(c t_1 t_2; \{ c x y \To x,\, c x y \To y\})$
do not have $\iota$-reducts (assuming $c' \notin \{c,d\}$). We usually
write $t \reduces t'$ to denote a finitary reduction $t
\reduces_{\beta\iota} t'$.

\begin{defi}\label{def_infred}
  Following~\cite{EndrullisPolonsky2011,EndrullisHansenHendriksPolonskySilva2015,
    EndrullisHansenHendriksPolonskySilva2018}, we define infinitary
  reduction $t \infred t'$ coinductively.
  \[
  \begin{array}{c}
    \infer={t \infred x}{t \reduces x}\quad\infer={t \infred
      c}{t \reduces c} \\ \\
    \infer={t \infred \lambda x . r'}{t \reduces \lambda x . r & r \infred r'}\quad
    \infer={t \infred r_1'r_2'}{t \reduces r_1r_2 & r_k \infred r_k'}
    \\ \\
    \infer={t \infred \case(r'; \{c_k\vec{x} \To r_k'\})}{t \reduces \case(r;
      \{c_k\vec{x} \To r_k\}) & r \infred r' & r_k \infred r_k'}
  \end{array}
  \]
\end{defi}

Intuitively, $t \infred t'$ holds if it may be obtained as the
conclusion of a potentially infinite derivation tree built using the
above rules. The idea with the definition of the infinitary
reduction~$\infred$ is that the depth at which a redex is contracted
should tend to infinity. This is achieved by defining~$\infred$ in
such a way that always after finitely many reduction steps the
subsequent contractions may be performed only at a greater depth. In
other words, if $t \infred t'$ then to produce any finite prefix
of~$t'$ only a finitary reduction from~$t$ is necessary, i.e., any
finite prefix of~$t'$ becomes fixed after finitely many reduction
steps and afterwards all reductions occur only at higher depths. The
idea for the definition of~$\infred$ comes
from~\cite{EndrullisPolonsky2011,EndrullisHansenHendriksPolonskySilva2015,EndrullisHansenHendriksPolonskySilva2018}.

Our coinductively defined notion of infinitary reduction corresponds
to the established notion of strongly convergent reduction in
infinitary rewriting~\cite{KennawayVries2003} (see
Lemma~\ref{lem_strongly_convergent_equivalent}). This notion has good
formal properties and an intuitive computational interpretation. Note
that this is different from weak (Cauchy) convergence where one
requires convergence with respect to the metric topology on terms, but
the depth of the reduction activity is not required to increase. A
reduction sequence may weakly converge to a limit, even though every
step is performed at the root. The term can then be thought of as
still changing, even though in the limit it is being reduced to
itself. See~\cite[Section~12.3]{KennawayVries2003} for a more detailed
discussion.

The proofs of the next three lemmas follow the pattern
from~\cite[Lemma~4.3-4.5]{EndrullisPolonsky2011}.

\begin{lem}\label{lem_infred_subst}
  If $t_1 \infred t_1'$ and $t_2 \infred t_2'$ then $t_1[t_2/x]
  \infred t_1'[t_2'/x]$.
\end{lem}

\begin{proof}
  Coinduction with case analysis on $t_1 \infred t_1'$, using that $t
  \reduces t'$ implies $t[t_2/x] \reduces t'[t_2/x]$.
\end{proof}

\begin{lem}\label{lem_infred_append}
  If $t \infred t' \to_{\beta\iota} t''$ then $t \infred t''$.
\end{lem}

\begin{proof}
  Induction on $t' \to_{\beta\iota} t''$, using
  Lemma~\ref{lem_infred_subst}.
\end{proof}

\begin{lem}\label{lem_infred_concat}
  If $t \infred t' \infred t''$ then $t \infred t''$.
\end{lem}

\begin{proof}
  By coinduction, analyzing $t' \infred t''$ and using
  Lemma~\ref{lem_infred_append}.
\end{proof}

The rest of this section contains some technical definitions and
results which are needed for the proof of the approximation theorem. A
reader not interested in the infinitary rewriting details of this
proof may skip the remainder of this section.

\begin{defi}
  We define the relation~$\to^{2\infty}$ analogously to~$\infred$, but
  replacing~$\reduces$ with~$\infred$ and~$\infred$
  with~$\to^{2\infty}$ in Definition~\ref{def_infred}.
\end{defi}

We may consider~$\infred$ (resp.~$\to^{2\infty}$) as defining a
strongly convergent ordinal-indexed reduction
sequence~\cite{KennawayVries2003} of length at most~$\omega$
(resp.~$\omega^2$), obtained by concatenating the finite
reductions~$\reduces$ occurring in the coinductive derivation. The
next lemma may be seen as a kind of compression lemma.

\begin{lem}\label{lem_infred_two}
  If $t \to^{2\infty} t'$ then $t \infred t'$.
\end{lem}

\begin{proof}
  By coinduction, using Lemma~\ref{lem_infred_concat}. See for
  example~\cite[Lemma~6.3]{Czajka2018} for details.
\end{proof}

The system of $\beta\iota$-reductions on infinitary terms~$\Tb^\infty$
may be presented as a fully-extended infinitary Combinatory Reduction
System (iCRS)~\cite{KetemaSimonsen2011}. One checks that this~iCRS is
orthogonal. A reader not familiar with the~iCRS formalism may skip the
following definition.

\begin{defi}\label{def_icrs}
  The signature of the iCRS contains:
  \begin{itemize}
  \item a distinct nullary symbol~$c$ for each constructor,
  \item a binary symbol~$\app$ denoting application,
  \item a unary symbol~$\lam$ denoting lambda abstraction, and
  \item for each $n \in \Nbb$ and each sequence of distinct
    constructors $c_1,\ldots,c_n$ and each sequence of natural numbers
    $k_1,\ldots,k_n$, a symbol
    $\case_{c_1,\ldots,c_n}^{k_1,\ldots,k_n}$ of arity $n+1$.
  \end{itemize}
  The iCRS has the following rewrite rules:
  \begin{itemize}
  \item $\app(\lam([x]Z(x)), X) \to Z(X)$,
  \item for each symbol~$\case_{c_1,\ldots,c_n}^{k_1,\ldots,k_n}$ and
    each $i=1,\ldots,n$:
    \[
    \begin{array}{l}
      \case_{c_1,\ldots,c_n}^{k_1,\ldots,k_n}(\app(\ldots (\app(\app(c_i, X_1), X_2)\ldots), X_{k_i}), \\
      \quad\quad\quad\quad\quad [x_1,\ldots,x_{k_1}]Z_1(x_1,\ldots,x_{k_1}), \ldots, [x_1,\ldots,x_{k_n}]Z_n(x_1,\ldots,x_{k_n}))\\
      \to\\
      Z_i(X_1,\ldots,X_{k_i})
    \end{array}
    \]
  \end{itemize}
  We assume $x_1,\ldots,x_{k_i}$ to be pairwise distinct, for
  $i=1,\ldots,n$.

  One sees that this iCRS corresponds to our informal presentation of
  terms and reductions, and that it is fully-extended and orthogonal.
\end{defi}

Our coinductive definition of the infinitary reduction
relation~$\infred$ corresponds to, in the sense of existence, to the
well-established notion of strongly convergent reduction
sequences~\cite{KetemaSimonsen2011,KennawayVries2003}. This is made precise in the next
lemma.

\begin{lem}\label{lem_strongly_convergent_equivalent}
  $t \infred t'$ iff there exists a strongly convergent reduction
  sequence from~$t$ to~$t'$.
\end{lem}

\begin{proof}
  This follows by a proof completely analogous
  to~\cite[Theorem~6.4]{Czajka2018}, \cite[Theorem~48]{Czajka2015a}
  or~\cite[Theorem~3]{EndrullisPolonsky2011}. The technique originates
  from~\cite{EndrullisPolonsky2011}. Lemma~\ref{lem_infred_two} is
  needed in the proof.
\end{proof}

\begin{defi}
  A term~$t$ is \emph{root-active} if for every $t'$ with $t \infred
  t'$ there is a $\beta\iota$-redex~$t''$ such that $t' \infred
  t''$. The set of root-active, or \emph{meaningless}, terms is
  denoted by~$\Uc$. By~$\sim_\Uc$ we denote equality of terms modulo
  equivalence of meaningless subterms.
\end{defi}

Meaningless terms are a technical notion needed in the proofs, because
for infinitary rewriting confluence holds only
modulo~$\sim_\Uc$. Intuitively, meaningless terms have no
``meaningful'' interpretation and may all be identified. An example of
a meaningless term is $\Omega = (\lambda x . x x) (\lambda x . x
x)$. Various other sets of meaningless terms have been considered in
the infinitary lambda
calculus~\cite{KennawayVries2003,Vries2016,SeveriVries2011,SeveriVries2011b,SeveriVries2005,KennawayOostromVries1999}. The
set of root-active terms is a subset of each of them.

Because our~iCRS is fully-extended and orthogonal, the following are
consequences of some results in~\cite{KetemaSimonsen2009} and the
previous lemma. Note that because all rules are collapsing, in our
setting root-active terms are the same as the hypercollapsing terms
from~\cite{KetemaSimonsen2009}.

\begin{lem}\label{lem_sim_concat}
  If $t \sim_\Uc t' \sim_\Uc t''$ then $t \sim_\Uc t''$.
\end{lem}

\begin{proof}
  Follows from~\cite[Proposition~4.12]{KetemaSimonsen2009}.
\end{proof}

\begin{lem}\label{lem_infred_sim}
  If $t \infred w$ and $t \sim_\Uc t'$ then there is~$w'$ with $t'
  \infred w'$ and $w \sim_\Uc w'$.
\end{lem}

\begin{proof}
  Follows from~\cite[Lemma~4.14]{KetemaSimonsen2009}.
\end{proof}

\begin{thm}
  The relation of infinitary reduction~$\infred$ is confluent
  modulo~$\Uc$, i.e., if $t \sim_\Uc t'$ and $t \infred u$ and
  $t' \infred u'$ then there exist $w,w'$ such that $w \sim_\Uc w'$
  and $u \infred w$ and $u' \infred w'$.
\end{thm}

\begin{proof}
  Follows from~\cite[Theorem~4.17]{KetemaSimonsen2009}.
\end{proof}

\section{Types}\label{sec_types}

In this section we define the types for which we will provide an
interpretation in our rewriting semantics. Some types will be
decorated with sizes of (co)inductive types, indicating the type of
approximations of a (co)inductive type of a given size.

\begin{defi}\label{def_size_expr}
  \emph{Size expressions} are given by the following grammar:
  \[
  \begin{array}{rcl}
    s &::=& \infty \mid 0 \mid i \mid s + 1 \mid \min(s, s)
            \mid \max(s, s)
  \end{array}
  \]
  where $i$ is a size variable. We denote the set of size variables
  by~$\Vc_S$.
\end{defi}
We use obvious abbreviations for size expressions, e.g., $i + 3$ for
$((i + 1) + 1) + 1$, or $\min(s_1,s_2,s_3)$ for
$\min(\min(s_1,s_2),s_3)$, or $\max(s)$ for~$s$, etc. Substitution
$s[s'/i]$ of~$s'$ for the size variable~$i$ in the size expression~$s$
is defined in the obvious way.

\begin{defi}\label{def_types}
  We assume an infinite set~$\Dc$ of (co)inductive definition names
  $d,d',d_1,\ldots$. \emph{Types} $\tau,\alpha,\beta$ are defined by:
  \[
  \begin{array}{rcl}
    \tau &::=& A \mid d^s(\tau_1,\ldots,\tau_n) \mid \tau_1 \to \tau_2
    \mid \forall i . \tau
  \end{array}
  \]
  where $A \in \Vc_T$ is a type variable, $s$ is a size expression,
  $i$ is a size variable, and~$d$ is a (co)inductive definition name.

  A type~$\tau$ is \emph{strictly positive} if one of the following
  holds:
  \begin{itemize}
  \item $\tau$ is closed (i.e.~it contains no type variables),
  \item $\tau = A$ is a type variable,
  \item $\tau = \tau_1 \to \tau_2$ and~$\tau_1$ is closed and~$\tau_2$
    is strictly positive,
  \item $\tau = \forall i . \tau'$ and~$\tau'$ is strictly positive,
  \item $\tau = d^\infty(\vec{\alpha})$ and each~$\alpha_k$ is strictly
    positive.
  \end{itemize}

  \noindent
  By~$\SV(s)$ (resp.~$\SV(\tau)$) we denote the set of all size
  variables occurring in~$s$ (resp.~$\tau$). By~$\FV(\tau)$ we denote
  the set of all type variables occurring in~$\tau$. By~$\FSV(\tau)$
  we denote the set of all free size variables occuring in~$\tau$
  (i.e.~those not bound by any~$\forall$).

  Substitution $\tau[\tau'/A]$, $s[s'/i]$, $\tau[s'/i]$ is defined in
  the obvious way, avoiding size variable capture. We abbreviate
  simultaneous substitution $\tau[\alpha_1/A_1,\ldots,\alpha_n/A_n]$ to
  $\tau[\vec{\alpha}/\vec{A}]$.
\end{defi}

To each (co)inductive definition name~$d \in \Dc$ we associate a
unique (co)inductive definition. Henceforth, we will use (co)inductive
definitions and their names interchangeably. Remember, however, that
strictly speaking (co)inductive definitions do not occur in types,
only their names do.

\begin{defi}\label{def_coind_defs}
  A \emph{coinductive definition} for~$d \in \Dc$ is specified by a
  defining equation of the form
  \[
  d(B_1,\ldots,B_n) = \CoInd(A)\{ c_k : \vec{\sigma_k} \mid k=1,\ldots,m\}
  \]
  where~$A$ is the \emph{recursive type variable}, and
  $B_1,\ldots,B_n$ are the \emph{parameter type variables}, and $m >
  0$, and~$c_k$ is the $k$th \emph{constructor}, and~$\sigma_k^l$ is
  the $k$th constructor's $l$th \emph{argument type}, and the
  following is satisfied:
  \begin{itemize}
  \item $\sigma_k^l$ are all strictly positive,
  \item $\FV(\sigma_k^l) \subseteq \{A,B_1,\ldots,B_n\}$,
  \item $\FSV(\sigma_k^l) = \emptyset$.
  \end{itemize}
  An \emph{inductive definition} is specified analogously, but
  using~$\Ind$ instead of~$\CoInd$.

  We assume that each constructor~$c$ is associated with a unique
  (co)inductive definition~$\Def(c)$.

  We assume there is a well-founded order~$\prec$ on (co)inductive
  definitions such that for every (co)inductive definition~$d$, each
  (co)inductive definition~$d'$ occurring in a constructor argument
  type of~$d$ satisfies $d' \prec d$.
\end{defi}
The type variable~$A$ is used as a placeholder for recursive
occurrences of~$d(\vec{B})$. We often write~$\ArgTypes(c_k)$ to denote
$(\sigma_k^1,\ldots,\sigma_k^{n_k})$: the argument types of the $k$-th
constructor. We usually present (co)inductive definitions in a bit
more readable format by replacing the recursive type variable~$A$ with
the type being defined, presenting the constructor argument types in a
chain of arrow types, and adding the type being defined as the target
type of constructors. For instance, the inductive definition of lists
is specified by
\[
  \List(B) = \Ind\{\nil : \List(B),\, \cons : B \to \List(B) \to \List(B)\}.
\]
Formally, here $\sigma_1^1 = A$, $\sigma_2^1 = B$, and $\sigma_2^2 =
A$.

\begin{exa}\label{ex_nat_strm_types}
  The inductive definition of natural numbers is specified by:
  \[
  \Nat = \Ind\{ 0 : \Nat,\, S : \Nat \to \Nat \}.
  \]
  The coinductive definition of streams of natural numbers is
  specified by:
  \[
  \Strm = \CoInd\{ \cons : \Nat \to \Strm \to \Strm \}.
  \]
\end{exa}

\begin{defi}
  An expression of the form~$d(\tau_1,\ldots,\tau_n)$ is a
  \emph{(co)inductive type}, depending on whether~$d$ is an inductive or
  coinductive definition. A type of the form~$d^s(\tau_1,\ldots,\tau_n)$
  is a \emph{decorated (co)inductive type}. We drop the designator
  ``decorated'' when clear from the context. We write
  $c \in \Constr(\rho)$ to denote that~$c$ is a constructor for a
  (decorated) (co)inductive type or definition~$\rho$.
\end{defi}

In a (co)inductive type $d^s(\tau_1,\ldots,\tau_n)$, the types
$\tau_1,\ldots,\tau_n$ denote the \emph{parameters}. Intuitively, we
substitute $\tau_1,\ldots,\tau_n$ for the parameter type
variables~$B_1,\ldots,B_n$ of the (co)inductive definition~$d$.

By default,~$d_\nu$ denotes a coinductive and~$d_\mu$ an inductive
definition. We use $\mu$ for inductive and $\nu$ for coinductive
types, and~$\rho$ for (co)inductive types when it is not important if
it is inductive or coinductive. Analogously, we use~$\mu^s$, $\nu^s$,
$\rho^s$ for decorated (co)inductive types (with size~$s$). We often
omit the superscript~$\infty$ in $\rho^\infty$, overloading the
notation.

Intuitively, $\mu^s$ denotes the type of objects of an inductive
type~$\mu$ which have size at most~$s$, and~$\nu^s$ denotes the type
of objects of a coinductive type~$\nu$ which have size at least~$s$,
i.e., considered up to depth~$s$ they represent a valid object of
type~$\nu$. For a stream $\nu = \Strm$, the type $\Strm^s$ is the type
of terms~$t$ which produce (under a sufficiently long reduction
sequence) at least~$s$ initial elements of a stream. The type
e.g.~$\forall i . \Strm^i \to \Strm^s$ is the type of functions which
when given as argument a stream of size~$i$ (i.e.~with at least~$i$
initial elements well-defined) produce at least~$s$ initial elements
of a stream, where~$i$ may occur in~$s$.

Note that the parameters to (co)inductive definitions may be other
(co)inductive types with size constraints. For instance
$\List(\List^i(\tau))$ denotes the type of lists (of any length) whose
elements are lists of length at most~$i$ with elements of
type~$\tau$. Note also that the recursive type variable~$A$ may occur
as a parameter of a (co)inductive type in the type of one of the
constructors. For these two reasons we need to require that the
parameter type variables occur only strictly positively in the types
of the arguments of constructors. One could allow non-positive
occurrences of parameter type variables in general and restrict the
occurrences to strictly positive only for instantiations with types
containing free size variables or recursive type variables. This
would, however, introduce some tedious but straightforward
technicalities in the proofs.

\begin{exa}\label{ex_trees_types}
  Infinite binary trees storing natural numbers in nodes may be
  specified by:
  \[
  \BTree = \CoInd\{ \bnode : \Nat \to \BTree \to \BTree \to \BTree \}.
  \]
  Trees with potentially infinite branches but finite branching at
  each node are specified by:
  \[
  \FTree = \CoInd\{ \fnode : \Nat \to \List(\FTree) \to \FTree \}.
  \]
  Here the type~$\FTree$ itself (formally, the recursive type
  variable~$A$) occurs as a parameter of~$\List$ in the type of the
  constructor~$\fnode$.

  Infinite trees with infinite branching are specified by:
  \[
  \Tree = \CoInd\{ \node : \Nat \to (\Nat \to \Tree) \to \Tree \}.
  \]
  Here infinite branching is specified by a function from~$\Nat$
  to~$\Tree$.

  Recall the coinducutive definition of the type~$\Tree_2$ from the
  introduction:
  \[
  \Tree_2 = \CoInd\{ c_1 : \List(\Tree_2) \to \Tree_2,\, c_2 : (\Nat \to
  \Tree_2) \to \Tree_2 \}.
  \]
  In this definition both finite branching via the~$c_1$ constructor
  and infinite branching via~$c_2$ are possible. In contrast
  to~$\BTree$, $\FTree$ and~$\Tree$, the nodes of~$\Tree_2$ do not
  store any natural number values.
\end{exa}

\begin{exa}\label{ex_stream_processors_type}
  As an example of a nested higher-order (co)inductive type we
  consider stream processors
  from~\cite{HancockPattinsonGhani2009}. See
  also~\cite[Section~2.3]{AbelPientka2013}. We define two types:
  \[
    \begin{array}{rcl}
      \SPi(B) &=& \Ind\{ \get : (\Nat \to \SPi(B)) \to \SPi(B),\\ & & \quad\quad\; \mput :
      \Nat \to B \to \SPi(B) \} \\
      \SP &=& \CoInd\{ \out : \SPi(\SP) \to \SP \}
    \end{array}
  \]
  The type~$\SP$ is a type of stream processors. A stream processor
  can either read the first element from the input stream and enter a
  new state depending on the read value (the $\get$ constructor), or
  it can write an element to the output stream and enter a new state
  (the $\mput$ constructor). To ensure productivity, a stream
  processor may read only finitely many elements from the input stream
  before writing a value to the output stream. This is achieved by
  nesting the inductive type~$\SPi$ inside the coinductive type~$\SP$
  of stream processors.
\end{exa}

The well-founded order~$\prec$ on (co)inductive definitions essentially
disallows mutual (co)inductive types. They may still be represented
indirectly thanks to type parameters.

\begin{exa}
  The types~$\Odd$ and~$\Even$ of odd and even natural numbers may be
  defined as mutual inductive types:
  \[
  \begin{array}{rcl}
    \Odd &=& \Ind\{S_o : \Even \to \Odd \} \\
    \Even &=& \Ind\{0 : \Even,\, S_e : \Odd \to \Even \}
  \end{array}
  \]
  These are not valid inductive definitions in our formalism, but
  they may be reformulated as follows:
  \[
  \begin{array}{rcl}
    \Odd_0(B) &=& \Ind\{S_o : B \to \Odd_0(B) \} \\
    \Even &=& \Ind\{0 : \Even,\, S_e : \Odd_0(\Even) \to \Even \}
  \end{array}
  \]
  Now the type $\Odd$ is represented by $\Odd_0(\Even)$.
\end{exa}

In the rest of this paper by ``induction on a type~$\tau$'' we mean
induction on the lexicographic product of:
\begin{itemize}
\item the multiset extension of the well-founded order~$\prec$ on
  (co)inductive definitions occurring in the type, and
\item the size of the type.
\end{itemize}
In this order, if $c \in \Constr(\rho)$ with $\ArgTypes(c) =
(\sigma_1,\ldots,\sigma_m)$ then each~$\sigma_k$ is smaller
than~$\rho$.

\section{Rewriting semantics}\label{sec_semantics}

In this section we define our rewriting semantics. More precisely, we
define an interpretation $\valuation{\tau}{}{} \subseteq \Tb^\infty$
for each type~$\tau$.

By $\infty$ we denote a sufficiently large ordinal (see
Definition~\ref{def_semantics}), and by $\Omega$ we denote the set of
all ordinals not greater than~$\infty$. A \emph{size variable
  valuation} is a function $v : \Vc_S \to \Omega$. Any size variable
valuation~$v$ extends in a natural way to a function from size
expressions to~$\Omega$. More precisely, we define:
$v(\infty) = \infty$, $v(0) = 0$, $v(s+1) = \min(v(s) + 1, \infty)$,
$v(\min(s_1,s_2)) = \min(v(s_1),v(s_2))$,
$v(\max(s_1,s_2)) = \max(v(s_1),v(s_2))$.  To save on notation we
identify ordinals larger than~$\infty$ with~$\infty$, e.g., $\infty+1$
denotes the ordinal~$\infty$.

\begin{defi}\label{def_semantics}
  We interpret types as subsets of~$\Tb^\infty$. By~$\infty$ we
  denote an ordinal large enough so that any monotone function on
  $\Pow(\Tb^\infty)$ (the powerset of~$\Tb^\infty$) reaches
  its least and greatest fixpoint in~$\infty$ iterations. This ordinal
  exists, as we may take any ordinal larger than the cardinality
  of~$\Pow(\Tb^\infty)$.

  Given a \emph{type variable valuation}
  $\xi : \Vc_T \to \Pow(\Tb^\infty)$, a size variable valuation
  $v : V_S \to \Omega$, and a strictly positive type~$\tau$, we define
  a \emph{type
    valuation}~$\valuation{\tau}{\xi,v}{} \subseteq \Tb^\infty$. This
  is done by induction on~$\tau$. We simultaneously also define
  valuation approximations~$\valuation{\rho}{\xi,v}{\varkappa}$
  and~$\valuation{d}{\xi,v}{\varkappa}$.
  \begin{itemize}
  \item Let
    $d(B_1,\ldots,B_n) = \IndCoInd(A)\{c_k : \vec{\sigma_k} \to
    d(\vec{B}) \}\ra$ be a (co)inductive definition. We define a
    function $\Phi_{d,\xi,v} : \Pow(\Tb^\infty) \to \Pow(\Tb^\infty)$
    so that $\Phi_{d,\xi,v}(X)$ for $X \subseteq \Tb^\infty$ contains
    all terms of the form $c_k t_k^1 \ldots t_k^{n_k}$ such that
    $t_k^l \in \valuation{\sigma_k^l}{\xi[X/A],v}{}$ for
    $l=1,\ldots,n_k$.

    For a coinductive definition~$d_\nu$ and an ordinal
    $\varkappa \in \Omega$ we define the \emph{valuation
      approximation}~$\valuation{d_\nu}{\xi,v}{\varkappa} \subseteq
    \Tb^\infty$ as follows:
    \begin{itemize}
    \item $\valuation{d_\nu}{\xi,v}{0} = \Tb^\infty$,
    \item $\valuation{d_\nu}{\xi,v}{\varkappa+1} =
      \Phi_{d_\nu,\xi,v}(\valuation{d_\nu}{\xi,v}{\varkappa})$,
    \item $\valuation{d_\nu}{\xi,v}{\varkappa} =
      \bigcap_{\varkappa'<\varkappa}\valuation{d_\nu}{\xi,v}{\varkappa'}$
      if~$\varkappa$ is a limit ordinal.
    \end{itemize}
    For an inductive definition~$d_\mu$ and an ordinal
    $\varkappa \in \Omega$ we define the \emph{valuation
      approximation}~$\valuation{d_\mu}{\xi,v}{\varkappa} \subseteq
    \Tb^\infty$ by:
    \begin{itemize}
    \item $\valuation{d_\mu}{\xi,v}{0} = \emptyset$,
    \item
      $\valuation{d_\mu}{\xi,v}{\varkappa+1} =
      \Phi_{d_\mu,\xi,v}(\valuation{d_\mu}{\xi,v}{\varkappa})$,
    \item
      $\valuation{d_\mu}{\xi,v}{\varkappa} =
      \bigcup_{\varkappa'<\varkappa}\valuation{d_\mu}{\xi,v}{\varkappa'}$
      if~$\varkappa$ is a limit ordinal.
    \end{itemize}
  \item
    $\valuation{\rho}{\xi,v}{\varkappa} =
    \valuation{d}{\xi[\vec{Y}/\vec{B}],v}{\varkappa}$ where
    $\rho = d(\vec{\alpha})$ is a (co)inductive type,
    $Y_j = \valuation{\alpha_j}{\xi,v}{}$, and $\vec{B}$ are the
    parameter type variables of~$d$.
  \item $\valuation{\rho^s}{\xi,v}{} = \valuation{\rho}{\xi,v}{v(s)}$.
  \item $\valuation{A}{\xi,v}{} = \xi(A)$.
  \item $t \in \valuation{\forall i . \tau}{\xi,v}{}$ if
    $i \notin \FSV(t)$ and for every $\varkappa \in \Omega$ there
    is~$t'$ with
    $t \infred t' \in \valuation{\tau}{\xi,v[\varkappa/i]}{}$.
  \item $t \in \valuation{\alpha \to \beta}{\xi,v}{}$ if for every
    $r \in \valuation{\alpha}{\xi,v}{}$ there is~$t'$ with
    $t r \infred t' \in \valuation{\beta}{\xi,v}{}$.
  \end{itemize}
  For a closed type~$\tau$ the valuation $\valuation{\tau}{\xi,v}{}$
  does not depend on~$\xi$, so we simply write $\valuation{\tau}{v}{}$
  instead. Whenever we omit the type variable valuation we implicitly
  assume the type to be closed.
\end{defi}

In general, the interpretation~$\valuation{\tau}{}{}$ of a type~$\tau$
may contain terms which are not in normal form. This is because of the
interpretation of function types and quantification over size
variables ($\forall i$). If $\tau$ is a simple first-order
(co)inductive type whose constructor argument types contain neither
function types ($\tau_1\to\tau_2$) nor quantification over size
variables ($\forall i . \tau'$), then $\valuation{\tau}{}{}$ contains
only normal forms.

Thus, we do not show infinitary weak normalization for terms having
function types. Nonetheless, our interpretation of $t \in
\valuation{\tau_1 \to \tau_2}{}{}$ is very natural and ensures
productivity of~$t$ regarded as a function: we require that for $u \in
\valuation{\tau_1}{}{}$ there is $u' \in \valuation{\tau_2}{}{}$ with
$t u \infred u'$. Intuitively, this means that for any $u \in
\valuation{\tau_1}{}{}$ the application $t u$ reduces ``in the limit''
to a term $u' \in \valuation{\tau_2}{}{}$, using only a finite initial
part of~$u$ to produce a finite initial part of~$u'$. Moreover, it is
questionable in the first place how sensible infinitary normalization
is as a ``correctness'' criterion for terms of function types.

\begin{exa}\label{ex_nat_strm}
  Recall the definitions of the types~$\Nat$ and~$\Strm$ from
  Example~\ref{ex_nat_strm_types}:
  \[
  \begin{array}{rcl}
    \Nat &=& \Ind\{ 0 : \Nat,\, S : \Nat \to \Nat \} \\
    \Strm &=& \CoInd\{ \cons : \Nat \to \Strm \to \Strm \}
  \end{array}
  \]
  The elements of $\valuation{\Nat}{}{}$ are the terms: $0, S(0),
  S(S(0)), \ldots$. We use common number notation, e.g.~$1$ for
  $S(0)$, etc. We usually write e.g.~$1 :: 2 :: t$ instead of
  $\cons\,1\,(\cons\,2\,t)$. The elements of~$\valuation{\Strm}{}{}$
  are all infinite terms of the form $n_1 :: n_2 :: n_3 :: \ldots$
  where $n_i \in \valuation{\Nat}{}{}$.

  Consider the term
  \[
    \tail = \lambda t . \case(t; \{\cons\,x\,y \To y\})
  \]
  We have $\tail \in \valuation{\Strm \to \Strm}{}{}$. Indeed, let
  $t \in \valuation{\Strm}{}{}$. Then $t = n :: t'$ with $n \in \valuation{\Nat}{}{}$
  and $t' \in \valuation{\Strm}{}{}$. Thus
  $\tail(t) \to \case(n :: t'; \{\cons\,x\,y \To y\}) \to t' \in
  \valuation{\Strm}{}{}$.
\end{exa}

\begin{exa}\label{ex_trees}
  Recall the definitions of~$\BTree$, $\FTree$ and~$\Tree$ form
  Example~\ref{ex_trees_types}:
  \[
  \begin{array}{rcl}
    \BTree &=& \CoInd\{ \bnode : \Nat \to \BTree \to \BTree \to \BTree \} \\
    \FTree &=& \CoInd\{ \fnode : \Nat \to \List(\FTree) \to \FTree \} \\
    \Tree &=& \CoInd\{ \node : \Nat \to (\Nat \to \Tree) \to \Tree \}
  \end{array}
  \]
  The interpretation~$\valuation{\BTree}{}{}$ consists of all infinite
  terms of the form
  \[
  \bnode\, n_{1,1}\, (\bnode\, n_{2,1}\, (\ldots)\, (\ldots))
  (\bnode\, n_{2,2}\, (\ldots)\, (\ldots))
  \]
  where $n_{1,1},n_{2,1},n_{2,2},\ldots \in \valuation{\Nat}{}{}$. The
  interpretation~$\valuation{\FTree}{}{}$ consists of all potentially
  infinite terms of the form $\fnode\, n_{1,1}\, ((\fnode\, n_{2,1}\,
  (\ldots)) :: (\fnode\, n_{2,2}\, (\ldots)) :: \ldots :: \nil)$ where
  $n_{1,1},n_{2,1},n_{2,2},\ldots \in \valuation{\Nat}{}{}$. Finally,
  $\valuation{\Tree}{}{}$ consists of all terms of the form $\node\,
  n\, f$ where $n \in \valuation{\Nat}{}{}$ for every $m \in
  \valuation{\Nat}{}{}$ there is $t \in \valuation{\Tree}{}{}$ such
  that $f m \infred t$.
\end{exa}

\begin{exa}\label{ex_stream_processors}
  Recall the definition of stream processors from Example~\ref{ex_stream_processors_type}:
  \[
    \begin{array}{rcl}
      \SPi(B) &=& \Ind\{ \get : (\Nat \to \SPi(B)) \to \SPi(B),\\ & & \quad\quad\; \mput :
      \Nat \to B \to \SPi(B) \} \\
      \SP &=& \CoInd\{ \out : \SPi(\SP) \to \SP \}
    \end{array}
  \]
  An example stream processor, i.e., an example element of
  $\valuation{\SP}{}{}$ is an infinite term~$\odd$ satisfying the
  identity:
  \[
    \odd = \out (\get (\lambda x . \get (\lambda y . \mput\, x\, \odd)))
  \]
  The stream processor \texttt{odd} drops every second element of a
  stream, e.g., it transforms the stream $1 :: 2 :: 3 :: 4 :: \ldots$
  into $1 :: 3 :: 5 :: \ldots$. But e.g.~the infinite term
  \[
    \out (\get (\lambda x_1 . \get (\lambda x_2 . \get (\lambda x_3
    . \get (\ldots)))))
  \]
  is \emph{not} in~$\valuation{\SP}{}{}$, because it nests infinitely
  many~$\get$s.
\end{exa}

\begin{lem}\label{lem_val_idom}
  If $v(i) = v'(i)$ for every $i \in \FSV(\tau)$ then
  $\valuation{\tau}{\xi,v}{} = \valuation{\tau}{\xi,v'}{}$. Moreover,
  $\valuation{d}{\xi,v}{\varkappa} = \valuation{d}{\xi,v'}{\varkappa}$
  for any $v,v'$.
\end{lem}

\begin{proof}
  Follows by induction on~$\tau$, using the fact
  $\FSV(\sigma_k^l) = \emptyset$ for~$\sigma_k^l$ a constructor
  argument type as in Definition~\ref{def_coind_defs}.
\end{proof}

\begin{lem}\label{lem_val_dom}~
  \begin{enumerate}
  \item If $\xi(A) = \xi'(A)$ for $A \in \FV(\tau)$ then
    $\valuation{\tau}{\xi,v}{} =
    \valuation{\tau}{\xi',v}{}$.
  \item If $\xi(B_i) = \xi'(B_i)$ for each parameter type
    variable~$B_i$ of~$d$, then $\valuation{d}{\xi,v}{\varkappa} =
    \valuation{d}{\xi',v}{\varkappa}$.
  \end{enumerate}
\end{lem}

\begin{proof} % easy
  Induction on~$\tau$, generalizing over~$\xi$, $\xi'$ and~$v$.
\end{proof}

\begin{cor}\label{cor_phi_dom}
  If $\xi(B_i) = \xi'(B_i)$ for each parameter type variable~$B_i$
  of~$d$, then $\Phi_{d,\xi,v} = \Phi_{d,\xi',v}$.
\end{cor}

\begin{lem}\label{lem_val_subset}
  Assume $\xi \subseteq \xi'$, i.e.,
  $\xi(A) \subseteq \xi'(A)$ for all type variables~$A$.
  \begin{enumerate}
  \item If $\tau$ is strictly positive then
    $\valuation{\tau}{\xi,v}{} \subseteq \valuation{\tau}{\xi',v}{}$.
  \item If~$d$ is a (co)inductive definition then
    $\valuation{d}{\xi,v}{\varkappa} \subseteq
    \valuation{d}{\xi',v}{\varkappa}$.
  \item If $X \subseteq X'$ then
    $\Phi_{d,\xi,v}(X) \subseteq \Phi_{d,\xi',v}(X')$.  In particular,
    the function $\Phi_{d,\xi,v}$ is monotone.
  \end{enumerate}
\end{lem}

\begin{proof} % easy
  Induction on $\tau$, generalizing over~$\xi,\xi',v$.
\end{proof}

From the third point in the above lemma it follows that
$\valuation{d_\nu}{\xi,v}{\varkappa_1} \subseteq
\valuation{d_\nu}{\xi,v}{\varkappa_2}$ for
$\varkappa_2 \le \varkappa_1$, and
$\valuation{d_\mu}{\xi,v}{\varkappa_1} \subseteq
\valuation{d_\mu}{\xi,v}{\varkappa_2}$ for
$\varkappa_1 \le \varkappa_2$. Also, for a (co)inductive
definition~$d$, by the Knaster-Tarski fixpoint
theorem~\cite{Tarski1955}, the function~$\Phi_{d,\xi,v}$
has the least and greatest fixpoints, which may be obtained by
``iterating'' $\Phi_{d,\xi,v}$ starting with the empty or the full
set, respectively, as in the definition of valuation
approximations. For an inductive definition~$d_\mu$, the least
fixpoint of~$\Phi_{d_\mu,\xi,v}$ is
then~$\valuation{d_\mu}{\xi,v}{\infty}$, by how we
defined~$\infty$. Analogously, for a coinductive definition~$d_\nu$
the greatest fixpoint of~$\Phi_{d_\nu,\xi,v}$
is~$\valuation{d_\nu}{\xi,v}{\infty}$. Note that for
$\varkappa \ge \infty$ we have
$\valuation{d}{\xi,v}{\varkappa} = \valuation{d}{\xi,v}{\infty}$.

The next definition and the ensuing lemma are needed in the proof of
the approximation theorem. A reader not interested in the details of
this proof may skip the rest of this section.

\begin{defi}
  A set~$X \subseteq \Tb^\infty$ is \emph{stable} when:
  \begin{enumerate}
  \item if $t \in X$ and $t \sim_\Uc t'$ then $t' \in X$,
  \item if $t \in X$ and $t \infred t'$ then $t' \in X$.
  \end{enumerate}
  A type variable valuation~$\xi$ is stable if~$\xi(A)$ is stable for
  each type variable~$A$. The following lemma implies that the
  interpretations of closed types are in fact stable.
\end{defi}

\begin{lem}\label{lem_val_stable}
  Assume $\tau,\rho$ are strictly positive.
  \begin{enumerate}
  \item If $\xi$ is stable then so is $\valuation{\tau}{\xi,v}{}$.
  \item If $\xi$ is stable then so is
    $\valuation{\rho}{\xi,v}{\varkappa}$.
  \item If $\xi$ and $X \subseteq \Tb^\infty$ are stable then so
    is $\Phi_{d_\rho,\xi,v}(X)$.
  \end{enumerate}
\end{lem}

\begin{proof} % easy
  We show the first point by induction on~$\tau$, generalizing
  over~$\xi,v$. The remaining two points will follow directly from
  this proof.

  First assume $\tau = \rho^s$ with $\rho=d(\vec{\alpha})$. Then
  $\valuation{\tau}{\xi,v}{} = \valuation{\rho}{\xi,v}{v(s)} =
  \valuation{d}{\xi[\vec{Y}/\vec{B}],v}{v(s)}$ where
  $Y_j = \valuation{\alpha_j}{\xi,v}{}$ and each~$\alpha_j$ is
  strictly positive. By the inductive hypothesis each~$Y_j$ is
  stable. Hence $\xi_1 = \xi[\vec{Y}/\vec{B}]$ is also stable. We show
  that if $X \subseteq \Tb^\infty$ is stable then so is
  $\Phi_{d,\xi_1,v}(X)$. From this it follows by induction that
  $\valuation{\rho}{\xi,v}{\varkappa}$ is stable for
  any~$\varkappa \in \Omega$, and thus~$\valuation{\tau}{\xi,v}{}$ is
  stable. Let $t \in \Phi_{d,\xi_1,v}(X)$. Then $t = c t_1 \ldots t_n$
  where $t_k \in \valuation{\sigma_k}{\xi_1[X/A],v}{}$ and
  $c \in \Constr(\rho)$ and $\ArgTypes(c) = (\sigma_1,\ldots,\sigma_n)$. Note that
  $\xi_1[X/A]$ is stable, because~$X$
  is. Hence~$\valuation{\sigma_k}{\xi_1[X/A],v}{}$ is stable by the
  inductive hypothesis.
  \begin{enumerate}
  \item Assume $t \sim_\Uc t'$. Then $t' = c t_1' \ldots t_n'$ with
    $t_k \sim_\Uc t_k'$. We have
    $t_k' \in \valuation{\sigma_k}{\xi_1[X/A],v}{}$
    because~$\valuation{\sigma_k}{\xi_1[X/A],v}{}$ is stable. Thus
    $t' \in \Phi_{d,\xi_1,v}(X)$.
  \item Assume $t \infred t'$. Then $t' = c t_1' \ldots t_n'$ with
    $t_k \infred t_k'$. We have
    $t_k' \in \valuation{\sigma_k}{\xi_1[X/A],v}{}$
    because~$\valuation{\sigma_k}{\xi_1[X/A],v}{}$ is stable. Thus
    $t' \in \Phi_{d,\xi_1,v}(X)$.
  \end{enumerate}

  If $\tau=A$ is a type variable then
  $\valuation{\tau}{\xi,v}{} = \xi(A)$ is stable because~$\xi$ is.

  Assume $\tau = \forall i . \tau'$. Let
  $t \in \valuation{\tau}{\xi,v}{}$.
  \begin{enumerate}
  \item Assume $t \sim_\Uc t'$. Let $\varkappa \in \Omega$. There
    is~$t_0$ with
    $t \infred t_0 \in \valuation{\tau'}{\xi,v[\varkappa/i]}{}$. By
    Lemma~\ref{lem_infred_sim} there is~$t_0'$ with
    $t_0 \sim_\Uc t_0'$ and $t' \infred t_0'$. By the inductive
    hypothesis $\valuation{\tau'}{\xi,v[\varkappa/i]}{}$ is stable, so
    $t_0' \in \valuation{\tau'}{\xi,v[\varkappa/i]}{}$. Thus
    $t' \in \valuation{\tau}{\xi,v}{}$ (without loss of generality
    $i \notin \FSV(t')$).
  \item Assume $t \infred t'$. There
    is~$t_0$ with
    $t \infred t_0 \in \valuation{\tau'}{\xi,v[\varkappa/i]}{}$. By
    confluence modulo~$\Uc$ there are~$t_1,t_2$ with
    $t_0 \infred t_1 \sim_\Uc t_2$ and $t' \infred t_2$. By the
    inductive hypothesis $\valuation{\tau'}{\xi,v[\varkappa/i]}{}$ is
    stable, so $t_2 \in \valuation{\tau'}{\xi,v[\varkappa/i]}{}$. Thus
    $t' \in \valuation{\tau}{\xi,v}{}$.
  \end{enumerate}

  Finally, assume $\tau = \tau_1 \to \tau_2$ with $\tau_1$ closed
  and~$\tau_2$ strictly positive. Let
  $t \in \valuation{\tau}{\xi,v}{}$. By the inductive
  hypothesis~$\valuation{\tau_2}{\xi,v}{}$ is stable.
  \begin{enumerate}
  \item Assume $t \sim_\Uc t'$. We need to show $t' \in
    \valuation{\tau}{\xi,v}{}$. Let $r \in
    \valuation{\tau_1}{\xi,v}{}$. Then $t r \infred t_0 \in
    \valuation{\tau_2}{\xi,v}{}$. We have $t r \sim_\Uc t' r$, so by
    Lemma~\ref{lem_infred_sim} there is~$t_0'$ with $t_0 \sim_\Uc
    t_0'$ and $t' r \infred
    t_0'$. Because~$\valuation{\tau_2}{\xi,v}{}$ is stable, $t_0' \in
    \valuation{\tau_2}{\xi,v}{}$.
  \item Assume $t \infred t'$. We need to show $t' \in
    \valuation{\tau}{\xi,v}{}$. Let $r \in
    \valuation{\tau_1}{\xi,v}{}$. Then $t r \infred t_0 \in
    \valuation{\tau_2}{\xi,v}{}$. We have $t r \infred t' r$, so by
    confluence there are $t_1,t_2$ with $t_0 \infred t_1 \sim_\Uc t_2$
    and $t' r \infred t_2$. Because $\valuation{\tau_2}{\xi,v}{}$ is
    stable, $t_2 \in \valuation{\tau_2}{\xi',v}{}$.\qedhere
  \end{enumerate}
\end{proof}

\section{Approximation theorem}\label{sec_approx}

In this section we prove the approximation theorem: if
$t \infred t_n \in \valuation{\nu}{v}{n}$ for $n \in \Nbb$ then there
exists $t_\infty \in \valuation{\nu}{v}{\infty}$ such that
$t \infred t_\infty$.

The approximation theorem is an easy consequence of the following
result: if $t_n \infred t_{n+1}$ and $t_n \in \valuation{\nu}{v}{n}$
for $n \in \Nbb$, then there exists~$t_\infty$ such that
$t_0 \infred t_\infty \in \valuation{\nu}{v}{\infty}$. If~$\nu$ is a
simple coinductive type, e.g., it is a stream with a single
constructor~$c$ where $\ArgTypes(c) = (\sigma, A)$, the type~$\sigma$
is closed, and~$A$ is the recursive type variable of~$\nu$, then the
argument is not complicated. It follows from the assumption that
$t_{n+1} = c u_{n+1} w_{n+1}$ with
$u_{n+1} \in \valuation{\sigma}{v}{}$,
$w_{n+1} \in \valuation{\nu}{v}{n}$ and $w_{n+1} \infred w_{n+2}$. We
coinductively construct~$w_\infty$ with
$w_1 \to^{2\infty} w_\infty \in \valuation{\nu}{v}{\infty}$ (note that
$\valuation{\nu}{v}{\infty}$ treated as a unary relation may be
defined coinductively). Take $t_\infty = c u_1 w_\infty$. We have
$t_0 \to^{2\infty} t_\infty \in \valuation{\nu}{v}{\infty}$, which
suffices by Lemma~\ref{lem_infred_two}. This reasoning captures the
gist of the argument. With higher-order (co)inductive types the core
idea remains the same but significant technical complications occur
because of the alternation of least and greatest fixpoints in the
definition of~$\valuation{-}{\xi,v}{}$. We construct the
term~$t_\infty$ by coinduction, and show $t_0 \infred t_\infty$ by
coinduction, and then show $t_\infty \in \valuation{\nu}{v}{\infty}$
by an inductive argument. To be able to even state an appropriately
generalized inductive hypothesis, we first need some definitions.

A reader not interested in the infinitary rewriting details of the
proof of the approximation theorem may skip directly to
Theorem~\ref{thm_approx}.

\begin{defi}\label{def_tau_sequence}
  Let $\tau$ be a strictly positive type and
  $\Xi = \{\xi_n\}_{n\in\Nbb}$ a family of type variable valuations. A
  \emph{$\tau,\Xi$-sequence} (with~$v$) is a sequence of
  terms~$\{t_n\}_{n\in\Nbb}$ satisfying
  $t_n \in \valuation{\tau}{\xi_n,v}{}$ and $t_n \infred t_{n+1}$ for
  $n\in\Nbb$.

  By $\Xi^\nu_{v} = \{\xi_n^\nu\}_{n\in\Nbb}$ we denote the family of
  type variable valuations such that
  $\xi_n^\nu(A) = \valuation{\nu}{v}{n}$ for all~$A$ and $n \in
  \Nbb$. We usually write~$\Xi^\nu$ instead of~$\Xi^\nu_{v}$ when~$v$
  is irrelevant or clear from the context. If
  $\Tc = \{\tau_A\}_{A\in V_T}$ is a family of strictly positive types
  and $\Xi = \{\xi_n\}_{n\in\Nbb}$ a family of type variable
  valuations, then $\Xi\rval{\Tc}_{v}$ denotes the family
  $\{\xi_n'\}_{n\in\Nbb}$ where
  $\xi_n'(A) = \valuation{\tau_A}{\xi_n,v}{}$. Again, the
  subscript~$v$ is usually omitted.

  A family~$\Xi$ of type variable valuations is
  \emph{$\nu$-hereditary} (with $v$) if $\Xi = \Xi^\nu_v$ or,
  inductively, $\Xi = \Xi'\rval{\Tc}_v$ for some
  $\nu$-hereditary~$\Xi'$ and a family~$\Tc$ of strictly positive
  types.

  A \emph{heredity derivation}~$D$ is either $\emptyset$, or,
  inductively, a pair $(D',\Tc)$ where~$D'$ is a heredity derivation
  and $\Tc$ a family of strictly positive types. The $\nu$-hereditary
  family~$\Xi^D$ \emph{determined by} a heredity derivation~$D$ is
  defined inductively: $\Xi^\emptyset = \Xi^\nu$ and
  $\Xi^{(D,\Tc)} = \Xi^D\rval{\Tc}$.

  A family~$\Xi = \{\xi_n\}_{n\in\Nbb}$ is \emph{stable} if
  each~$\xi_n$ is stable.
\end{defi}

For the sake of readability we usually talk about $\nu$-hereditary
families, but we always implicitly assume that for any given
$\nu$-hereditary family~$\Xi$ we are given a fixed heredity
derivation~$D$ such that $\Xi=\Xi^D$.

\begin{lem}\label{lem_hereditary_stable}
  Any $\nu$-hereditary family~$\Xi$ is stable.
\end{lem}

\begin{proof}
  By induction on the definition of a $\nu$-hereditary family, using
  Lemma~\ref{lem_val_stable}.
\end{proof}

\begin{lem}\label{lem_hereditary_val}
  If a family~$\Xi$ determined by a heredity derivation~$D$ is
  $\nu$-hereditary with $v$ and the size variable~$i$ is fresh, i.e.,
  it does not occur in~$\nu$ or any of the types in the type families
  in~$D$, then~$\Xi$ is $\nu$-hereditary with $v[\varkappa/i]$ and
  determined by the same heredity derivation~$D$.
\end{lem}

\begin{proof}
  Induction on~$D$. If $D=\emptyset$
  then~$\Xi=\Xi^\nu_{v}=\Xi^\nu_{v[\varkappa/i]}$ by
  Lemma~\ref{lem_val_idom}, because~$i$ does not occur in~$\nu$. If
  $D=(D',\Tc)$ and $\Xi=\Xi^{D'}\rval{\Tc}$, then by the inductive
  hypothesis~$\Xi^{D'}$ is $\nu$-hereditary with $v[\varkappa/i]$ and
  determined by the heredity derivation~$D'$. Assuming
  $\Xi=\{\xi_n\}_{n\in\Nbb}$ and $\Xi^{D'}=\{\xi_n'\}_{n\in\Nbb}$, we
  have $\xi_n(A) = \valuation{\tau_A}{\xi_n',v}{} =
  \valuation{\tau_A}{\xi_n',v[\varkappa/i]}{}$ by
  Lemma~\ref{lem_val_idom} because $i \notin \FSV(\tau_A)$. So~$\Xi$
  is $\nu$-hereditary with $v[\varkappa/i]$ and determined by~$D$.
\end{proof}

\begin{lem}\label{lem_nu_form}
  If $\{t_n\}_{n\in\Nbb}$ is a $A,\Xi^\nu$-sequence, then
  $t_{n+1} = c t_{n+1}^1 \ldots t_{n+1}^m$ for $n \in \Nbb$, and
  $\{t_{n+1}^k\}_{n\in\Nbb}$ is a $\sigma_k,\Xi'$-sequence for each
  $k=1,\ldots,m$ where $c \in \Constr(\nu)$ and
  $\ArgTypes(c) = (\sigma_1,\ldots,\sigma_m)$ and
  $\nu = d_\nu(\vec{\alpha})$ and $\Xi' = \Xi^\nu\rval{\Tc}$ where
  $\Tc=\{\tau_{A'}\}_{A'\in V_T}$ and $\tau_{B_j} = \alpha_j$ and
  $\tau_{A'} = A'$ for $A'\notin\{B_1,\ldots,B_l\}$ and
  $B_1,\ldots,B_l$ are the parameter type variables of~$d_\nu$.
\end{lem}

\begin{proof}
  Let $\Xi' = \{\xi_n'\}_{n\in\Nbb}$. We have
  \[
  \begin{array}{rcl}
    t_{n+1} &\in&
    \valuation{A}{\xi_{n+1}^\nu,v}{} \\
    &=& \valuation{\nu}{v}{n+1} \\
    &=& \valuation{d_\nu}{\xi,v}{n+1} \\
    &=& \Phi_{d_\nu,\xi,v}(\valuation{d_\nu}{\xi,v}{n}) \\
    &=& \Phi_{d_\nu,\xi,v}(\valuation{\nu}{v}{n})
  \end{array}
  \]
  where $\xi(B_j) = \valuation{\alpha_j}{v}{}$ and $B_1,\ldots,B_l$
  are the parameter type variables of~$d_\nu$. Then
  $t_{n+1} = c_{n+1} t_{n+1}^1 \ldots t_{n+1}^{m_{n+1}}$ with
  $t_{n+1}^k \in
  \valuation{\sigma_k^{n+1}}{\xi[\valuation{\nu}{v}{n}/A],v}{}$ where
  $\ArgTypes(c_{n+1}) =
  (\sigma_1^{n+1},\ldots,\sigma_{m_{n+1}}^{n+1})$ and~$A$ is the
  recursive type variable of~$d_\nu$. Since $t_n \infred t_{n+1}$ for
  $n \in \Nbb$ we must have $c_{n+1} = c$ and $m_{n+1} = m$ and
  $\sigma_k^{n+1} = \sigma_k$ for fixed~$c,m,\sigma_k$ not depending
  on~$n$. Also $t_{n+1}^k \infred t_{n+2}^k$ for $k=1,\ldots,m$ and
  $n \in \Nbb$. Because $\xi[\valuation{\nu}{v}{n}/A]$ and $\xi_n'$
  are identical on $\{A,B_1,\ldots,B_l\}$, by Lemma~\ref{lem_val_dom}
  we have $t_{n+1}^k \in \valuation{\sigma_k}{\xi_n',v}{}$. Thus
  $\{t_{n+1}^k\}_{n\in\Nbb}$ is a $\sigma_k,\Xi'$-sequence.
\end{proof}

\begin{lem}\label{lem_rho_infty_form}
  If $\tau = d^\infty(\vec{\alpha})$ and $\{t_n\}_{n\in\Nbb}$ is a
  $\tau,\Xi$-sequence, then $t_n = c t_n^1 \ldots t_n^m$ and
  $\{t_n^k\}_{n\in\Nbb}$ is a $\sigma_k,\Xi'$-sequence for each
  $k=1,\ldots,m$ where $c \in \Constr(d)$ and
  $\ArgTypes(c) = (\sigma_1,\ldots,\sigma_m)$ and
  $\Xi' = \Xi\rval{\Tc}$ where $\Tc=\{\tau_{A'}\}_{A'\in V_T}$ and
  $\tau_A = \tau$ and $\tau_{B_j} = \alpha_j$ and $\tau_{A'} = A'$ for
  $A'\notin\{A,B_1,\ldots,B_l\}$ and $B_1,\ldots,B_l$ are the
  parameter type variables of~$d$ and~$A$ is the recursive type
  variable of~$d$.
\end{lem}

\begin{proof}
  The proof is analogous to the proof of Lemma~\ref{lem_nu_form}, but
  using the fact that $\valuation{d}{\xi,v}{\infty} =
  \Phi_{d,\xi,v}(\valuation{d}{\xi,v}{\infty})$.
\end{proof}

\begin{defi}\label{def_t_infty}
  Let $S^\nu$ be the set of triples $(\tau,\Xi,\{t_n\}_{n\in\Nbb})$
  such that~$\tau$ is strictly positive, $\Xi = \{\xi_n\}_{n\in\Nbb}$
  is $\nu$-hereditary, and $\{t_n\}_{n\in\Nbb}$ is a
  $\tau,\Xi$-sequence. By corecursion we define a function $f^\nu :
  S^\nu \to \Tb^\infty$. Let $\{t_n\}_{n\in\Nbb}$ be a
  $\tau,\Xi$-sequence. First note that if $\tau = A$ then we may
  assume $\Xi = \Xi^\nu$, because as long as $\tau=A$ and $\Xi =
  \Xi'\rval{\Tc}$, the sequence $\{t_n\}_{n\in\Nbb}$ is also a
  $\tau_A,\Xi'$-sequence, so we may use the definition for the case
  $\tau=\tau_A$ and $\Xi=\Xi'$.
  \begin{itemize}
  \item If $\tau$ is closed then $f^\nu(\tau,\Xi,\{t_n\}_{n\in\Nbb}) =
    t_0$.
  \item If $\tau = A$ then without loss of generality $\Xi = \Xi^\nu$
    and by Lemma~\ref{lem_nu_form} for $n \in \Nbb$ we have $t_{n+1} =
    c t_{n+1}^1 \ldots t_{n+1}^m$ and $\{t_{n+1}^k\}_{n\in\Nbb}$ is a
    $\sigma_k,\Xi'$-sequence for each~$k=1,\ldots,m$. Then define
    $f^\nu(\tau,\Xi,\{t_n\}_{n\in\Nbb}) = c r_1 \ldots r_m$ where $r_k
    = f^\nu(\sigma_k,\Xi',\{t_{n+1}^i\}_{n\in\Nbb})$.
  \item If $\tau = d^\infty(\vec{\alpha})$ then by
    Lemma~\ref{lem_rho_infty_form} we have
    $t_n = c t_n^1 \ldots t_n^m$ and $\{t_n^k\}_{n\in\Nbb}$ is a
    $\sigma_k,\Xi'$-sequence for each $k=1,\ldots,m$. Then define
    $f^\nu(\tau,\Xi,\{t_n\}_{n\in\Nbb}) = c r_1 \ldots r_m$ where
    $r_k = f^\nu(\sigma_k,\Xi',\{t_n^k\}_{n\in\Nbb})$.
  \item If $\tau = \forall i . \tau'$ then
    $f^\nu(\tau,\Xi,\{t_n\}_{n\in\Nbb}) = t_0$.
  \item If $\tau = \tau_1 \to \tau_2$ then
    $f^\nu(\tau,\Xi,\{t_n\}_{n\in\Nbb}) = t_0$.
  \end{itemize}
  We usually denote $f^\nu(\tau,\Xi,\{t_n\}_{n\in\Nbb})$ by $t_\infty$
  when $\tau,\Xi$ and $\{t_n\}_{n\in\Nbb}$ are clear from the context.
\end{defi}

\begin{lem}\label{lem_infred_f_nu}
  If $\Xi$ is $\nu$-hereditary and $\{t_n\}_{n\in\Nbb}$ is a
  $\tau,\Xi$-sequence then $t_0 \infred t_\infty$.
\end{lem}

\begin{proof}
  By Lemma~\ref{lem_infred_two} it suffices to show $t_0 \to^{2\infty}
  t_\infty$. We proceed by coinduction. By the definition
  of~$t_\infty$ there are the following possibilities.
  \begin{itemize}
  \item If $\tau$ is closed then $t_\infty = t_0$ so $t_0
    \to^{2\infty} t_\infty$.
  \item If $\tau = A$ then without loss of generality $\Xi=\Xi^\nu$
    and for $n \in \Nbb$ we have
    $t_{n+1} = c t_{n+1}^1 \ldots t_{n+1}^m$ and
    $\{t_{n+1}^k\}_{n\in\Nbb}$ is a $\sigma_k,\Xi'$-sequence
    for~$k=1,\ldots,m$. Then $t_\infty = c r_1 \ldots r_m$ with
    $r_k = f^\nu(\sigma_k,\Xi',\{t_{n+1}^k\}_{n\in\Nbb})$. By the
    coinductive hypothesis $t_1^k \to^{2\infty} r_k$. Because
    $t_0 \infred c t_1^1 \ldots t_1^m$, we have
    $t_0 \to^{2\infty} t_\infty$.
  \item If $\tau = d^\infty(\vec{\alpha})$ then
    $t_n = c t_n^1 \ldots t_n^m$ and $\{t_n^k\}_{n\in\Nbb}$ is a
    $\sigma_k,\Xi'$-sequence for each $k=1,\ldots,m$. Then
    $t_\infty = c r_1 \ldots r_m$ where
    $r_k = f^\nu(\sigma_k,\Xi',\{t_n^k\}_{n\in\Nbb})$. By the
    coinductive hypothesis $t_0^k \to^{2\infty} r_k$, so
    $t_0 \to^{2\infty} t_\infty$.
  \item If $\tau = \forall i . \tau'$ or $\tau = \tau_1 \to \tau_2$
    then $t_\infty = t_0$, so $t_0 \to^{2\infty} t_\infty$.\qedhere
  \end{itemize}
\end{proof}

We want to show that if $\Xi$ is $\nu$-hereditary and
$\{t_n\}_{n\in\Nbb}$ is a $\tau,\Xi$-sequence, then
$t_\infty \in \bigcap_{n\in\Nbb}\valuation{\tau}{\xi_n,v}{}$
(Corollary~\ref{cor_complete}). Together with the above lemma and some
auxiliary results this will imply the approximation theorem
(Theorem~\ref{thm_approx}). First, we need a few more definitions and
auxiliary lemmas.

\begin{defi}
  Let $\Xi=\{\xi_n\}_{n\in\Nbb}$ and $\Xi'=\{\xi_n'\}_{n\in\Nbb}$. We
  write $\Xi \subseteq \Xi'$ if $\xi_n \subseteq \xi_n'$ for $n \in
  \Nbb$.
\end{defi}

\begin{lem}\label{lem_approx_family}
  If $\Xi \subseteq \Xi'$ and $\{t_n\}_{n\in\Nbb}$ is a
  $\tau,\Xi$-sequence, then $\{t_n\}_{n\in\Nbb}$ is also a
  $\tau,\Xi'$-sequence.
\end{lem}

\begin{proof}
  Follows from definitions and Lemma~\ref{lem_val_subset}.
\end{proof}

\begin{lem}\label{lem_sequence}
  If $t \infred t_n \in \valuation{\tau_n}{\xi_n,v}{}$ and~$\xi_n$ is
  stable for $n \in \Nbb$ then there exists a sequence of terms
  $\{t_n'\}_{n\in\Nbb}$ such that $t \infred t_0'$ and $t_n' \in
  \valuation{\tau_n}{\xi_n,v}{}$ and $t_n' \infred t_{n+1}'$ for $n
  \in \Nbb$.
\end{lem}

\begin{proof}
  By induction we define the terms $w_n$ and~$t_n'$ such that $t_n
  \infred w_n \sim_\Uc t_n'$ and $\{t_n'\}_{n\in\Nbb}$ satisfies the
  required properties. See Figure~\ref{fig_sequence}. We take $t_0' =
  w_0 = t_0$. For the inductive step, assume $w_n$ and~$t_n'$ are
  defined. By Lemma~\ref{lem_infred_concat} and confluence
  modulo~$\Uc$ there are~$w_{n+1}$ and~$w_{n+1}'$ such that $t_{n+1}
  \infred w_{n+1} \sim_\Uc w_{n+1}'$ and $w_n \infred w_{n+1}'$. By
  Lemma~\ref{lem_infred_sim} there is~$t_{n+1}'$ with $t_n' \infred
  t_{n+1}'$ and $w_{n+1}' \sim_\Uc t_{n+1}'$. By
  Lemma~\ref{lem_sim_concat} we have $w_{n+1} \sim_\Uc
  t_{n+1}'$. Because~$\xi_{n+1}$ is stable, by
  Lemma~\ref{lem_val_stable} so
  is~$\valuation{\tau_{n+1}}{\xi_{n+1},v}{}$. Since $t_{n+1} \in
  \valuation{\tau_{n+1}}{\xi_{n+1},v}{}$ and $t_{n+1} \infred w_{n+1}
  \sim_\Uc t_{n+1}'$ we obtain $t_{n+1}' \in
  \valuation{\tau_{n+1}}{\xi_{n+1},v}{}$.
\end{proof}

\begin{figure}[ht]
  \centerline{
    \xymatrix{
      & & t \ar@{->}[dll]^>>{\infty} \ar@{->}[d]^>>{\infty} \ar@{->}[dr]^>>{\infty}
      \ar@{->}[drr]^>>{\infty} \ar@{->}[drrr]^>>{\infty} & & & \\
      t_0 \ar@{->}[d]^>>{\infty} & \ldots & t_n \ar@{->}[d]^>>{\infty} & t_{n+1}
      \ar@{->}[d]^>>{\infty} & t_{n+2}
      \ar@{->}[d]^>>{\infty} & \ldots \\
      t_0 \ar@{}[d]|{{\displaystyle\wr}} & \ldots & w_n \ar@{}[d]|{{\displaystyle\wr}} \ar@{->}[dr]^>>{\infty} & w_{n+1}
      \ar@{}[d]|{{\displaystyle\wr}} \ar@{->}[dr]^>>{\infty} & w_{n+2}
      \ar@{}[d]|{{\displaystyle\wr}} & \ldots \\
      t_0 \ar@{}[d]|{{\displaystyle\wr}} & \ldots & w_n' \ar@{}[d]|{{\displaystyle\wr}} & w_{n+1}'
      \ar@{}[d]|{{\displaystyle\wr}} & w_{n+2}' \ar@{}[d]|{{\displaystyle\wr}} & \ldots \\
      t_0 \ar@{->}[r]^>>{\infty} & \ldots \ar@{->}[r]^>>{\infty} & t_n' \ar@{->}[r]^>>{\infty} & t_{n+1}' \ar@{->}[r]^>>{\infty} &
      t_{n+2}' \ar@{->}[r]^>>{\infty} & \ldots
    }
  }
  \caption{Proof of Lemma~\ref{lem_sequence}.}\label{fig_sequence}
\end{figure}

\begin{defi}\label{def_complete}
  A $\nu$-hereditary~$\Xi = \{\xi_n\}_{n\in\Nbb}$ is
  \emph{semi-complete} with $Z,\iota$ if $Z \subseteq \Xi$ is stable
  and for every type variable~$A$ and every $A,Z$-sequence
  $\{t_n\}_{n\in\Nbb}$ (which is also a $A,\Xi$-sequence by
  Lemma~\ref{lem_approx_family}) we have $t_\infty =
  f^\nu(A,\Xi,\{t_n\}_{n\in\Nbb}) \in \iota(A)$. The family~$\Xi$ is
  \emph{complete} if it is semi-complete with $\Xi,\xi_m$ for each
  $m\in\Nbb$.
\end{defi}

\begin{rem}\label{rem_complete_valuation}
  Note that the definition of ``semi-complete'' depends on the
  implicit size variable valuation~$v$, through~$\Xi$ and the
  function~$f^\nu$. Let $\Xi$ be $\nu$-hereditary (with $v$) and
  semi-complete with $Z,\iota$, with the implicit
  valuation~$v$. Let~$i$ be a fresh size variable. Then by
  Lemma~\ref{lem_hereditary_val} the family~$\Xi$ is $\nu$-hereditary
  with~$v[\varkappa/i]$ and determined by the same heredity
  derivation. It is also semi-complete with $Z,\iota$, with the
  implicit valuation~$v[\varkappa/i]$. This is because if~$\Xi$ is
  $\nu$-hereditary with $v[\varkappa/i]$ and $\{t_n\}_{n\in\Nbb}$ a
  $\tau,\Xi$-sequence with $v[\varkappa/i]$, then it follows from
  Definition~\ref{def_t_infty} and the statements of
  Lemma~\ref{lem_nu_form} and Lemma~\ref{lem_rho_infty_form} that only
  the type~$\tau$, the heredity derivation and the
  sequence~$\{t_n\}_{n\in\Nbb}$ determine the value
  of~$f^\nu(\tau,\Xi,\{t_n\}_{n\in\Nbb})$. Also note that the property
  of being an $A,Z$-sequence does not depend on~$v$, because~$A$ is a
  type variable.
\end{rem}

We are now going to show that if $\Xi=\{\xi_n\}_{n\in\Nbb}$ is
complete and $\{t_n\}_{n\in\Nbb}$ is a $\tau,\Xi$-sequence, then
$t_\infty = f^\nu(\tau,\Xi,\{t_n\}_{n\in\Nbb}) \in
\bigcap_{n\in\Nbb}\valuation{\tau}{\xi_n,v}{}$
(Corollary~\ref{cor_compl}). This is a consequence of the following a
bit more general lemma. Its proof is rather long and technical, and
therefore delegated to an appendix to make the overall structure of
the proof of the approximation theorem clearer.

\begin{lem}\label{lem_compl}
  If $\Xi=\{\xi_n\}_{n\in\Nbb}$ is $\nu$-hereditary with~$v$ and
  semi-complete with~$Z,\iota$, and $\{t_n\}_{n\in\Nbb}$ is a
  $\tau,Z$-sequence (and thus a $\tau,\Xi$-sequence by
  Lemma~\ref{lem_approx_family}), then:
  \[
    t_\infty =
    f^\nu(\tau,\Xi,\{t_n\}_{n\in\Nbb}) \in \valuation{\tau}{\iota,v}{}.
  \]
\end{lem}

\begin{cor}\label{cor_compl}
  If $\Xi=\{\xi_n\}_{n\in\Nbb}$ is complete and $\{t_n\}_{n\in\Nbb}$
  is a $\tau,\Xi$-sequence, then $t_\infty =
  f^\nu(\tau,\Xi,\{t_n\}_{n\in\Nbb}) \in
  \bigcap_{n\in\Nbb}\valuation{\tau}{\xi_n,v}{}$.
\end{cor}

We are now going to show that every $\nu$-hereditary family~$\Xi$ is
complete. To achieve this we show that $\Xi^\nu$ is complete
(Corollary~\ref{cor_complete_compose}), and that if~$\Xi$ is complete
then so is~$\Xi\rval{\Tc}$ (Lemma~\ref{lem_xi_nu_complete}).

\begin{lem}\label{lem_semi_complete_compose}
  If $\Xi$ is semi-complete with $Z,\iota$ then $\Xi\rval{\Tc}$ is
  semi-complete with $Z\rval{\Tc},\iota'$ where $\Tc=\{\tau_A\}_{A\in
    V_T}$ and $\iota'(A) = \valuation{\tau_A}{\iota,v}{}$.
\end{lem}

\begin{proof}
  Let $\Xi = \{\xi_n\}_{n\in\Nbb}$ and $\Xi' = \Xi\rval{\Tc} =
  \{\xi_n'\}_{n\in\Nbb}$ and $Z = \{\zeta_n\}_{n\in\Nbb}$ and $Z' =
  Z\rval{\Tc} = \{\zeta_n'\}_{n\in\Nbb}$. We have $\zeta_n'(A) =
  \valuation{\tau_A}{\zeta_n,v}{} \subseteq
  \valuation{\tau_A}{\xi_n,v}{} = \xi_n'(A)$ by
  Lemma~\ref{lem_val_subset} because $Z \subseteq \Xi$ and thus
  $\zeta_n \subseteq \xi_n$. Hence $Z' \subseteq \Xi'$. Let
  $\{t_n\}_{n\in\Nbb}$ be a $A,Z'$-sequence, i.e., $t_n \infred
  t_{n+1}$ and $t_n \in \valuation{A}{\zeta_n',v}{} = \zeta_n'(A) =
  \valuation{\tau_A}{\zeta_n,v}{}$ for $n \in \Nbb$. Then
  $\{t_n\}_{n\in\Nbb}$ is also a $\tau_A,Z$-sequence. Because~$\Xi$ is
  semi-complete with $Z,\iota$, by Lemma~\ref{lem_compl} we have
  $t_\infty = f^\nu(A,\Xi',\{t_n\}_{n\in\Nbb}) =
  f^\nu(\tau_A,\Xi,\{t_n\}_{n\in\Nbb}) \in
  \valuation{\tau_A}{\iota,v}{} = \iota'(A)$.
\end{proof}

\begin{cor}\label{cor_complete_compose}
  If $\Xi$ is complete then so is $\Xi\rval{\Tc}$.
\end{cor}

\begin{lem}\label{lem_xi_nu_complete}
  $\Xi^\nu$ is complete.
\end{lem}

\begin{proof}
  We show by induction on~$m\in\Nbb$ that $\Xi^\nu$ is semi-complete
  with $\Xi^\nu,\xi^\nu_m$. We have $\Xi^\nu \subseteq
  \Xi^\nu$. Also~$\Xi^\nu$ is stable. Let $\{t_n\}_{n\in\Nbb}$ be a
  $A,\Xi^\nu$-sequence. We need to show
  $t_\infty = f^\nu(A,\Xi^\nu,\{t_n\}_{n\in\Nbb}) \in \xi^\nu_m(A) =
  \valuation{\nu}{v}{m}$. If $m=0$ then
  $\valuation{\nu}{v}{m} = \Tb^\infty$, so
  $t_\infty \in \xi^\nu_m(A)$. Assume $m=m'+1$. We have
  $t_n \in \xi^\nu_n(A) = \valuation{\nu}{v}{n}$. Then by
  Lemma~\ref{lem_nu_form} we have
  $t_{n+1} = c t_{n+1}^1 \ldots t_{n+1}^k$ for $n \in \Nbb$, and
  $\{t_{n+1}^i\}_{n\in\Nbb}$ is a $\sigma_i,\Xi'$-sequence for each
  $i=1,\ldots,k$ where $c \in \Constr(\nu)$ and
  $\ArgTypes(c) = (\sigma_1,\ldots,\sigma_k)$ and
  $\nu = d_\nu(\vec{\alpha})$ and $\Xi' = \Xi^\nu\rval{\Tc}$ where
  $\Tc=\{\tau_{A}\}_{A\in V_T}$ and $\tau_{B_j} = \alpha_j$ and
  $\tau_{A} = A$ for $A\notin\{B_1,\ldots,B_l\}$ and $B_1,\ldots,B_l$
  are the parameter type variables of~$d_\nu$. By the inductive
  hypothesis $\Xi^\nu$ is semi-complete with
  $\Xi^\nu,\xi^\nu_{m'}$. By Lemma~\ref{lem_semi_complete_compose} we
  conclude that~$\Xi'$ is semi-complete with $\Xi',\iota$ where
  $\iota(A) = \valuation{\tau_{A}}{\xi^\nu_{m'},v}{}$, i.e.,
  $\iota =
  \xi^\nu_{m'}[\valuation{\alpha_1}{v}{}/B_1,\ldots,\valuation{\alpha_l}{v}{}/B_l]$. Because
  $\{t_{n+1}^i\}_{n\in\Nbb}$ is a $\sigma_i,\Xi'$-sequence, by
  Lemma~\ref{lem_compl} we have
  $t_\infty^i \in \valuation{\sigma_i}{\iota,v}{}$. Thus
  $t_\infty = c t_\infty^1 \ldots t_\infty^k \in
  \Phi_{d_\nu,\iota,v}(\iota(A')) =
  \Phi_{d_\nu,\iota,v}(\valuation{\nu}{v}{m'}) =
  \valuation{\nu}{v}{m}$, where~$A'$ is the recursive type variable
  of~$d_\nu$. Hence $t_\infty \in \xi^\nu_m(A)$.
\end{proof}

\begin{cor}\label{cor_nu_hereditary_complete}
  Every $\nu$-hereditary family is complete.
\end{cor}

\begin{proof}
  Follows by induction from Lemma~\ref{lem_xi_nu_complete} and
  Corollary~\ref{cor_complete_compose}.
\end{proof}

\begin{cor}\label{cor_complete}
  If $\Xi=\{\xi_n\}_{n\in\Nbb}$ is $\nu$-hereditary and
  $\{t_n\}_{n\in\Nbb}$ is a $\tau,\Xi$-sequence, then
  $t_\infty = f^\nu(\tau,\Xi,\{t_n\}_{n\in\Nbb}) \in
  \bigcap_{n\in\Nbb}\valuation{\tau}{\xi_n,v}{}$.
\end{cor}

\begin{proof}
  Follows from Corollary~\ref{cor_compl} and
  Corollary~\ref{cor_nu_hereditary_complete}.
\end{proof}

We are now going to show that
$\valuation{\nu}{v}{\omega} = \valuation{\nu}{v}{\infty}$, i.e.,
$\omega$ iterations suffice to reach the fixpoint for any coinductive
type. For this we need the following lemma about intersection of
valuations. We define $\bigcap_{n\in\Nbb}\xi_n$ by
$(\bigcap_{n\in\Nbb}\xi_n)(A) = \bigcap_{n\in\Nbb}\xi_n(A)$ for
any~$A$.

\begin{lem}\label{lem_intersection}
  If $\Xi=\{\xi_n\}_{n\in\Nbb}$ is complete then
  $\bigcap_{n\in\Nbb}\valuation{\tau}{\xi_n,v}{}\subseteq\valuation{\tau}{\bigcap_{n\in\Nbb}\xi_n,v}{}$
  for any strictly positive~$\tau$.
\end{lem}

\begin{proof}
  Induction on~$\tau$. The proof is similar to the proof the auxiliary
  Lemma~\ref{lem_val_intersection} in Appendix~\ref{app_proofs}. We
  treat three cases that differ more substantially.
  \begin{itemize}
  \item If $\tau=d_\mu^s(\vec{\alpha})$ then
    $\bigcap_{n\in\Nbb}\valuation{\tau}{\xi_n,v}{} =
    \bigcap_{n\in\Nbb}\valuation{\mu}{\xi_n,v}{v(s)}$ where
    $\mu=d_\mu(\vec{\alpha})$. By induction on~$\varkappa$ we show
    $\bigcap_{n\in\Nbb}\valuation{\mu}{\xi_n,v}{\varkappa} \subseteq
    \valuation{\mu}{\bigcap_{n\in\Nbb}\xi_n,v}{\varkappa}$. There are
    three cases.
    \begin{enumerate}
    \item $\varkappa=0$. Then
      $\bigcap_{n\in\Nbb}\valuation{\mu}{\xi_n,v}{\varkappa}
      = \bigcap_{n\in\Nbb}\emptyset = \emptyset =
      \valuation{\mu}{\bigcap_{n\in\Nbb}\xi_n,v}{\varkappa}$.
    \item $\varkappa=\varkappa'+1$. Let $t \in
      \bigcap_{n\in\Nbb}\valuation{\mu}{\xi_n,v}{\varkappa}$. Then $t
      = c u_1 \ldots u_k$ with $u_i \in
      \bigcap_{n\in\Nbb}\valuation{\sigma_i}{\xi_n',v}{}$ where~$A$ is
      the recursive type variable of~$d_\mu$ and $c \in \Constr(d_\mu)$ and
      $\ArgTypes(c) = (\sigma_1,\ldots,\sigma_k)$ and $B_1,\ldots,B_l$ are the
      parameter type variables of~$d_\mu$ and $\Xi'=\Xi\rval{\Tc}$ and
      $\Xi'=\{\xi_n'\}_{n\in\Nbb}$ and $\Tc=\{\tau_A\}_{A\in V_T}$ and
      $\tau_{B_j}=\alpha_j$ and $\tau_A = \mu^i$ and $\tau_A =
      A'$ for $A' \notin \{A,B_1,\ldots,B_l\}$ where $i$ is a fresh
      size variable such that $v(i) = \varkappa'$ (by
      Lemma~\ref{lem_val_idom} we may assume such a size variable
      exists). So~$\Xi'$ is also complete by
      Corollary~\ref{cor_complete_compose}. By the main inductive
      hypothesis $u_i \in \valuation{\sigma_i}{\xi,v}{}$ where $\xi =
      \bigcap_{n\in\Nbb}\xi_n'$. We have $\xi(A) =
      \bigcap_{n\in\Nbb}\valuation{\mu}{\xi_n,v}{\varkappa'} \subseteq
      \valuation{\mu}{\bigcap_{n\in\Nbb}\xi_n,v}{\varkappa'}$ by the
      inductive hypothesis. Also $(\bigcap_{n\in\Nbb}\xi_n')(B_j) =
      \bigcap_{n\in\Nbb}\xi_n'(B_j) =
      \bigcap_{n\in\Nbb}\valuation{\alpha_j}{\xi_n',v}{} \subseteq
      \valuation{\alpha_j}{\bigcap_{n\in\Nbb}\xi_n',v}{} =
      \valuation{\alpha_j}{\bigcap_{n\in\Nbb}\xi_n,v}{}$ by the
      inductive hypothesis and Lemma~\ref{lem_val_dom}, because we may
      assume $B_1,\ldots,B_l\notin \FV(\alpha_j)$. Hence
      \[
      \begin{array}{rcll}
      \xi &=&
      (\bigcap_{n\in\Nbb}\xi_n)[&\valuation{\mu}{\bigcap_{n\in\Nbb}\xi_n,v}{\varkappa'}/A,\\
      &&&\valuation{\alpha_1}{\bigcap_{n\in\Nbb}\xi_n,v}{}/B_1,\\
      &&&\ldots,\\
      &&&\valuation{\alpha_l}{\bigcap_{n\in\Nbb}\xi_n,v}{}/B_l].
      \end{array}
      \]
      Therefore
      $t \in
      \valuation{\mu}{\bigcap_{n\in\Nbb}\xi_n,v}{\varkappa}$.
    \item $\varkappa$ is a limit ordinal. Let
      $t\in\bigcap_{n\in\Nbb}\valuation{\mu}{\xi_n,v}{\varkappa}=
      \bigcap_{n\in\Nbb}\bigcup_{\varkappa_n<\varkappa}\valuation{\mu}{\xi_n,v}{\varkappa_n}$. Then
      for each $n\in\Nbb$ there is $\varkappa_n<\varkappa$ with
      $t\in\valuation{\mu}{\xi_n,v}{\varkappa_n}$, i.e., $t \in
      \bigcap_{n\in\Nbb}\valuation{\mu}{\xi_n,v}{\varkappa_n}$. We
      have $\varkappa_n>0$ is a successor ordinal for $n\in\Nbb$,
      because $\valuation{\mu}{\xi_n,v}{0} = \emptyset$. Because~$\Xi$
      is stable by Lemma~\ref{lem_hereditary_stable}, using
      Lemma~\ref{lem_mu_varkappa} we conclude $t \in
      \bigcap_{n\in\Nbb}\valuation{\mu}{\xi_n,v}{\varkappa_0}$. Then
      $t \in \valuation{\mu}{\bigcap_{n\in\Nbb}\xi_n,v}{\varkappa_0}
      \subseteq \valuation{\mu}{\bigcap_{n\in\Nbb}\xi_n,v}{\varkappa}$
      by an argument as in the previous point.
    \end{enumerate}
  \item If $\tau=\forall i . \tau'$ then let
    $t \in \bigcap_{n\in\Nbb}\valuation{\tau}{\xi_n,v}{}$. Let
    $\varkappa \in \Omega$. For $n \in \Nbb$ there is~$t_n$ with
    $t \infred t_n \in \valuation{\tau'}{\xi_n,v[\varkappa/i]}{}$. By
    Lemma~\ref{lem_hereditary_stable} and Lemma~\ref{lem_sequence}
    there exists a sequence of terms $\{t_n'\}_{n\in\Nbb}$ such that
    $t \infred t_0'$ and
    $t_n' \in \valuation{\tau'}{\xi_n,v[\varkappa/i]}{}$ and
    $t_n' \infred t_{n+1}'$ for $n \in \Nbb$. Thus
    $\{t_n'\}_{n\in\Nbb}$ is a $\tau',\Xi$-sequence (with
    $v[\varkappa/i]$). Because~$\Xi$ is complete, by
    Corollary~\ref{cor_compl} there is~$t^\varkappa$ with
    $t \infred t^\varkappa \in
    \bigcap_{n\in\Nbb}\valuation{\tau'}{\xi_n,v[\varkappa/i]}{}$. By
    the inductive hypothesis
    $t^\varkappa \in
    \valuation{\tau'}{\bigcap_{n\in\Nbb}\xi_n,v[\varkappa/i]}{}$. Since
    $\varkappa\in\Omega$ was arbitrary, this implies
    $t \in \valuation{\tau}{\bigcap_{n\in\Nbb}\xi_n,v}{}$.
  \item If $\tau=\tau_1\to\tau_2$ then let
    $t \in \bigcap_{n\in\Nbb}\valuation{\tau_1 \to \tau_2}{\xi_n,v}{}$
    and $w \in \valuation{\tau_1}{\bigcap_{n\in\Nbb}\xi_n,v}{}$. We
    have $w \in \bigcap_{n\in\Nbb}\valuation{\tau_1}{\xi_n,v}{}$ by
    Lemma~\ref{lem_val_subset}. Hence there exists a sequence of terms
    $\{w_n\}_{n\in\Nbb}$ with
    $t w \infred w_n \in \valuation{\tau_2}{\xi_n,v}{}$. By
    Lemma~\ref{lem_hereditary_stable} and Lemma~\ref{lem_sequence}
    there exists a sequence of terms $\{w_n'\}_{n\in\Nbb}$ such that
    $t w \infred w_0'$ and $w_n' \in \valuation{\tau_2}{\xi_n,v}{}$
    and $w_n' \infred w_{n+1}'$ for $n \in \Nbb$. Thus
    $\{w_n'\}_{n\in\Nbb}$ is a $\tau_2,\Xi$-sequence. Because~$\Xi$ is
    complete, by Corollary~\ref{cor_compl} we have
    $w_\infty' \in
    \bigcap_{n\in\Nbb}\valuation{\tau_2}{\xi_n,v}{}$. By the inductive
    hypothesis
    $w_\infty' \in
    \valuation{\tau_2}{\bigcap_{n\in\Nbb}\xi_n,v}{}$. By
    Lemma~\ref{lem_infred_f_nu} and Lemma~\ref{lem_infred_concat} we
    also have $t w \infred w_\infty'$. This shows
    $t \in \valuation{\tau}{\bigcap_{n\in\Nbb}\xi_n,v}{}$.\qedhere
  \end{itemize}
\end{proof}

The following lemma shows that for a coinductive type~$\nu$ we
have~$\valuation{\nu}{v}{\omega} =
\valuation{\nu}{v}{\infty}$. Because we allow only strictly
positive coinductive types, $\omega$ iterations suffice to reach the
fixpoint. A similar result was already obtained in
e.g.~\cite{Abel2003}.

\begin{lem}\label{lem_nu_fixpoint}
  $\valuation{\nu}{v}{\omega} = \valuation{\nu}{v}{\infty}$.
\end{lem}

\begin{proof}
  It suffices to show
  $\valuation{\nu}{v}{\omega}\subseteq\valuation{\nu}{v}{\omega+1}$. So
  let $t \in \valuation{\nu}{v}{\omega}$. Then
  $t\in\valuation{\nu}{v}{m}$ for each $m\in\Nbb$. So
  $t = c u_1 \ldots u_k$ where
  $u_i \in \valuation{\sigma_i}{\xi_m',v}{}$ for $m \in \Nbb$ where
  $\Xi'=\{\xi_m'\}_{m\in\Nbb}$ and $\Xi'=\Xi^\nu\rval{\Tc}$ and
  $\Tc=\{\tau_A\}_{A \in V_T}$ and
  $\tau_{B_j} = \valuation{\alpha_j}{v}{}$ and $\tau_A = A$ for
  $A \notin \{B_1,\ldots,B_l\}$ and $\nu=d_\nu(\vec{\alpha})$ and
  $B_1,\ldots,B_l$ are the parameter type variables of~$d_\nu$. Note
  that~$\Xi'$ is complete by
  Corollary~\ref{cor_nu_hereditary_complete}. Hence by
  Lemma~\ref{lem_intersection} we have
  $u_i \in \valuation{\sigma_i}{\bigcap_{m\in\Nbb}\xi_m',v}{}$. Let
  $\xi'=\bigcap_{m\in\Nbb}\xi_m'$. We have
  $\xi'(A) = \bigcap_{m\in\Nbb}\xi_m'(A) =
  \bigcap_{m\in\Nbb}\xi_m^\nu(A) =
  \bigcap_{m\in\Nbb}\valuation{\nu}{v}{m} =
  \valuation{\nu}{v}{\omega}$ where $A$ is the recursive type variable
  of~$d_\nu$, and $\xi'(B_j) = \valuation{\alpha_j}{v}{}$. Therefore
  $t \in \valuation{\nu}{v}{\omega+1}$.
\end{proof}

Finally, we prove the approximation theorem. Lemma~\ref{lem_sequence},
Lemma~\ref{lem_infred_f_nu}, Corollary~\ref{cor_complete} and
Lemma~\ref{lem_nu_fixpoint} are used in the proof.

\begin{thm}[Approximation Theorem]\label{thm_approx}
  If $t \infred t_n \in \valuation{\nu}{v}{n}$ for $n \in \Nbb$ then
  there exists $t_\infty \in \valuation{\nu}{v}{\infty}$ such that
  $t \infred t_\infty$.
\end{thm}

\begin{proof}
  By Lemma~\ref{lem_sequence} there exists a sequence of terms
  $\{t_n\}_{n\in\Nbb}$ such that $t \infred t_0$ and
  $t_n \in \valuation{\nu}{v}{n}$ and $t_n \infred t_{n+1}$ for
  $n \in \Nbb$. Hence $\{t_n\}_{n\in\Nbb}$ is a
  $A,\Xi^\nu_v$-sequence. By Lemma~\ref{lem_infred_f_nu} we have
  $t_0 \infred t_\infty$, and hence $t \infred t_\infty$ by
  Lemma~\ref{lem_infred_concat}. By Corollary~\ref{cor_complete} we
  have
  $t_\infty \in \bigcap_{n\in\Nbb}\valuation{A}{\xi_n^\nu,v}{} =
  \bigcap_{n\in\Nbb}\xi_n^\nu(A) =
  \bigcap_{n\in\Nbb}\valuation{\nu}{v}{m} =
  \valuation{\nu}{v}{\omega}$. Also
  $\valuation{\nu}{v}{\omega} = \valuation{\nu}{v}{\infty}$ by
  Lemma~\ref{lem_nu_fixpoint}, so
  $t_\infty \in \valuation{\nu}{v}{\infty}$.
\end{proof}

We now precisely formulate the result about approximations of infinite
objects informally described in the introduction: if for every
approximation~$u_n$ of size~$n$ of an infinite object~$u$ the
application~$t u_n$ reduces to an approximation of an infinite object
of the right type, with the result approximations getting larger
as~$n$ gets larger, then there is a reduction starting from~$t u$
which ``in the limit'' produces an infinite object of the right
type. We show that this follows from the approximation theorem.

First, we show that a weak version of this is a direct consequence of
Theorem~\ref{thm_approx}.

\begin{prop}\label{prop_approx_2}
  Let $t \in \Tb^\infty$ and let $f : \Nbb \to \Nbb$ be such that
  $\lim_{n\to\infty}f(n) = \infty$. Assume that for every $n \in \Nbb$
  and every $u_n \in \valuation{\nu_1}{}{n}$ there is~$w_n$ with
  $t u_n \infred w_n \in \valuation{\nu_2}{}{f(n)}$. Then
  $t \in \valuation{\nu_1 \to \nu_2}{}{}$, i.e., for every
  $u \in \valuation{\nu_1}{}{}$ there is~$w$ with
  $t u \infred w \in \valuation{\nu_2}{}{}$.
\end{prop}

\begin{proof}
  Let $u \in \valuation{\nu_1}{}{}$ Because
  $\valuation{\nu_1}{}{} = \valuation{\nu_1}{}{\infty} \subseteq
  \valuation{\nu_1}{}{n}$, for each $n \in \Nbb$ there is~$w_n$ with
  $t u \infred w_n \in \valuation{\nu_2}{}{f(n)}$. Because
  $\lim_{n\to\infty}f(n) = \infty$, we may choose a strictly
  increasing subsequence $\{f(n_k)\}_{k\in\Nbb}$ from the sequence
  $\{f(n)\}_{n\in\Nbb}$. Then $f(n_k) \ge k$ for $k \in \Nbb$. Hence
  $\valuation{\nu_2}{}{f(n_k)} \subseteq \valuation{\nu_2}{}{k}$. This
  implies that for each $k \in \Nbb$ there is~$w_{n_k}$ with
  $t u \infred w_{n_k} \in \valuation{\nu_2}{}{k}$. Now by
  Theorem~\ref{thm_approx} there is~$w$ with
  $t u \infred w \in \valuation{\nu_2}{}{\infty}$.
\end{proof}

The above result is, however, a bit unsatisfying in that the valuation
approximations~$\valuation{\nu_1}{}{n}$ contain too many terms, i.e.,
they contain all terms which nest at least~$n$ constructors of the
coinductive type~$\nu_1$. In particular, the infinite object~$u$ is an
approximation of itself, on which the above proof relies. It would be
closer to informal intuition to weaken the hypothesis in
Proposition~\ref{prop_approx_2} by requiring the approximants of
size~$n$ to nest exactly~$n$ constructors of the approximated
coinductive type.

\begin{defi}
  Let $\bot = (\lambda x . x x) (\lambda x . x x)$. Note that~$\bot$
  is the only reduct of~$\bot$.

  For a coinductive definition~$d_\nu$ and $n \in \Nbb$ we define the
  \emph{strict valuation
    approximation}~$\valuation{d_\nu}{\bot,\xi,v}{n} \subseteq
  \Tb^\infty$ as follows:
  $\valuation{d_\nu}{\bot\xi,v}{0} = \{\bot\}$,
  $\valuation{d_\nu}{\bot,\xi,v}{n+1} =
  \Phi_{d_\nu,\xi,v}(\valuation{d_\nu}{\bot,\xi,v}{n})$. We set
  $\valuation{\nu}{\bot,\xi,v}{n} =
  \valuation{d}{\bot,\xi[\vec{Y}/\vec{B}],v}{n}$ where
  $\nu = d_\nu(\vec{\alpha})$ is a coinductive type,
  $Y_j = \valuation{\alpha_j}{\xi,v}{}$, and~$\vec{B}$ are the
  parameter type variables of~$d_\nu$.

  The relation~$\succ$ is defined coinductively.
  \[
    \begin{array}{c}
      \infer={t \succ \bot}{} \quad\quad \infer={x \succ x}{}
      \quad\quad \infer={c \succ c}{} \\ \\
      \infer={\lambda x . t \succ \lambda x . t'}{t \succ t'} \quad\quad \infer={t_1 t_2 \succ t_1'
      t_2'}{t_1 \succ t_1' & t_2 \succ t_2'} \\ \\
      \infer={\case(t; \{c_k\vec{x}\To t_k\}) \succ \case(t';
      \{c_k\vec{x}\To t_k'\})}{t \succ t' & t_k \succ t_k'}
    \end{array}
  \]
  In other words, $t \succ t'$ if~$t'$ is~$t$ with some subterms
  replaced by~$\bot$. If $t \succ t'$, $t \in \valuation{\nu}{}{}$ and
  $t' \in \valuation{\nu}{\bot}{n}$ then~$t'$ is an \emph{approximant}
  of~$t$ of size~$n$.
\end{defi}

\begin{lem}\label{lem_succ_red}
  If $t \succ t' \to u'$ then there is~$u$ with
  $t \to^\equiv u \succ u'$.
\end{lem}

\begin{proof}
  Induction on $t' \to u'$.
\end{proof}

\begin{lem}\label{lem_succ_infty}
  If $t \succ t' \infred u'$ then there is~$u$ with
  $t \infred u \succ u'$.
\end{lem}

\begin{proof}
  By coinduction, analysing $t' \infred u'$ and using
  Lemma~\ref{lem_succ_red}. More precisely, one defines an appropriate
  function
  $f : \Tb^\infty \times \Tb^\infty \times \Tb^\infty \to \Tb^\infty$
  by corecursion and shows $t \infred f(t, t', u')$ and
  $f(t, t', u') \succ u'$ by coinduction separately.
\end{proof}

A set $X \subseteq \Tb^\infty$ is \emph{approximation expansion
  closed} if $t' \in X$ and $t \succ t'$ imply $t \in X$.

\begin{lem}\label{lem_succ_val}
  Assume $\xi(A)$ is approximation expansion closed for
  every~$A$. Then $\valuation{\tau}{\xi,v}{}$ is approximation
  expansion closed.
\end{lem}

\begin{proof}
  Induction on~$\tau$, using Lemma~\ref{lem_succ_infty} for the cases
  $\tau = \tau_1\to\tau_2$ and $\tau = \forall i . \tau'$.
\end{proof}

\begin{thm}\label{thm_approx_2}
  Let $t \in \Tb^\infty$ and let $f : \Nbb \to \Nbb$ be such that
  $\lim_{n\to\infty}f(n) = \infty$. Let $u \in \valuation{\nu_1}{}{}$.
  If for every $n \in \Nbb$ and every
  $u_n \in \valuation{\nu_1}{\bot}{n}$ with $u \succ u_n$ there
  is~$w_n$ with $t u_n \infred w_n \in \valuation{\nu_2}{}{f(n)}$,
  then there is~$w$ with $t u \infred w \in \valuation{\nu_2}{}{}$.
\end{thm}

\begin{proof}
  Let $n \in \Nbb$ and let $u_n \in \valuation{\nu_1}{\bot}{n}$ be
  such that $u \succ u_n$. There is~$w_n$ with
  $t u_n \infred w_n \in \valuation{\nu_2}{}{f(n)}$. We have
  $t u \succ t u_n$. By Lemma~\ref{lem_succ_infty} there is~$v_n$ with
  $t u \infred v_n \succ w_n$. By Lemma~\ref{lem_succ_val} we have
  $v_n \in \valuation{\nu_2}{}{f(n)}$. Now, because
  $\lim_{n\to\infty}f(n)=\infty$, by an argument like the one in the
  proof of Proposition~\ref{prop_approx_2}, we may conclude that there
  is~$w$ with $t u \infred w \in \valuation{\nu_2}{}{}$.
\end{proof}

\section{The type system~$\lambda^\Diamond$}\label{sec_type_system}

In this section we define the type system~$\lambda^\Diamond$ which
provides a syntactic correctness criterion for finite terms decorated
with type information. In the next section we use the approximation
theorem to prove soundness: if a finite decorated term~$t$ has
type~$\tau$ in the system~$\lambda^\Diamond$ then its erasure
infinitarily reduces to a $t' \in \valuation{\tau}{}{}$.

\emph{Decorated terms} are given by:
\[
  \begin{array}{rcl}
    t &::=& x \mid c \mid \lambda x : \tau . t \mid t t \mid t s \mid
            \Lambda i . t \mid \case(t; \{ c_k
            \vec{x} \To t_k \}) \mid \fix f : \tau . t \mid \cofix^j f
            : \tau . t
  \end{array}
\]
where~$x \in \Vc$, and~$c,c_k \in \Cc$, and~$\tau$ is a type, and~$j$
is a size variable, and~$s$ is a size expression.

We define $s_1 \le s_2$ iff $v(s_1) \le v(s_2)$ for every size
variable valuation~$v$.

The function $\tgt$ that gives the \emph{target} of a type is
defined as follows:
\begin{itemize}
\item $\tgt(A) = A$, $\tgt(\rho^s) = \rho^s$,
\item $\tgt(\tau_1 \to \tau_2) = \tgt(\tau_2)$,
\item $\tgt(\forall i . \tau) = \tgt(\tau)$.
\end{itemize}
By $\chgtgt(\tau,\alpha)$ we denote the type~$\tau$ with the target
exchanged for~$\alpha$. Formally, $\chgtgt(\tau,\alpha)$ is defined
inductively:
\begin{itemize}
\item $\chgtgt(A,\alpha) = \alpha$, $\tgt(\rho^s,\alpha) = \alpha$,
\item $\chgtgt(\tau_1 \to \tau_2,\alpha) = \tau_1 \to \chgtgt(\tau_2,\alpha)$,
\item $\chgtgt(\forall i . \tau,\alpha) = \forall i . \chgtgt(\tau,\alpha)$.
\end{itemize}
Note that free size variables in~$\alpha$ may be captured as a
result of this operation.

A \emph{context}~$\Gamma$ is a finite map from type variables to
types. We write $\Gamma, x : \alpha$ to denote the context~$\Gamma'$
such that $\Gamma'(x) = \alpha$ and $\Gamma'(y) = \Gamma(y)$ for
$x \ne y$. A \emph{judgement} has the form
$\Gamma \proves t : \alpha$. The rules of the type
system~$\lambda^\Diamond$ are presented in
Figure~\ref{fig_type_system}. Figure~\ref{fig_subtyping} defines the
subtyping relation used in Figure~\ref{fig_type_system}. A closed
decorated term~$t$ is \emph{typable} if $\proves t : \tau$ for
some~$\tau$. In Figure~\ref{fig_type_system} \emph{all types are
  assumed to be closed} (i.e.~they don't contain free type variables,
but may contain free size variables). In Figure~\ref{fig_type_system}
the type variable~$A$ denotes the recursive type variable of the
(co)inductive definition considered in a given rule, and~$\vec{B}$
denote the parameter type variables.

\begin{figure*}
  \[
  \begin{array}{c}
    \infer[(\text{ax})]{\Gamma, x : \tau \proves x : \tau}{} \quad
    \infer[(\text{sub})]{\Gamma \proves t : \tau'}{\Gamma \proves t :
    \tau & \tau \sqsubseteq \tau'} \\ \\
    \infer[(\text{con})]{\Gamma \proves c t_1 \ldots t_n : \rho^{s+1}}{
      \begin{array}{lll}
        \ArgTypes(c) = (\sigma_1, \ldots, \sigma_n) & \Def(c) = d & \rho =
        d(\vec{\tau}) \\
        \multicolumn{3}{c}{\Gamma \proves t_k :
        \sigma_k[\rho^s/A][\vec{\tau}/\vec{B}] \text{ for } k = 1,\ldots,n}
      \end{array}} \\ \\
    \infer[\text{(lam)}]{\Gamma \proves (\lambda x : \alpha . t) : \alpha \to \beta}{\Gamma, x : \alpha \proves t : \beta} \quad
    \infer[\text{(app)}]{\Gamma \proves t t' : \beta}{\Gamma \proves t : \alpha \to \beta & \Gamma \proves t' : \alpha} \\ \\
    \infer[\text{(inst)}]{\Gamma \proves t s : \tau[s/i]}{\Gamma \proves t : \forall i . \tau} \quad
    \infer[\text{(gen)}]{\Gamma \proves \Lambda i . t : \forall i . \tau}{\Gamma \proves t :
    \tau & i \notin \FSV(\Gamma)} \\ \\
    \infer[(\text{case})]{\Gamma \proves \case(t; \{c_k \vec{x_k} \To t_k \mid
    k=1,\ldots,n\}) : \tau}
    {
      \begin{array}{lll}
        \ArgTypes(c_k) = (\sigma_k^1, \ldots, \sigma_k^{n_k}) &
        \delta_k^l = \sigma_k^l[\rho^{s}/A][\vec{\tau}/\vec{B}] & \rho = d(\vec{\tau}) \\
        \Gamma \proves t : \rho^{s+1} & \multicolumn{2}{l}{\Gamma, x_k^1 : \delta_k^1, \ldots, x_k^{n_k} : \delta_k^{n_k}
        \proves t_k : \tau}
      \end{array}} \\ \\
    \infer[\text{(fix)}]{\Gamma \proves (\fix f : \forall
    j_1 \ldots j_n . \mu \to \tau . t)
      : \forall j_1 \ldots j_n . \mu \to \tau}{
      \Gamma, f : \forall j_1 \ldots j_n . \mu^i \to \tau \proves
      t : \forall j_1 \ldots j_n . \mu^{i+1} \to \tau & i \notin \FSV(\Gamma,\mu,\tau,j_1,\ldots,j_n)
    } \\ \\
    \infer[\text{(cofix)}]{\Gamma \proves (\cofix^j f : \tau . t) : \tau}{
    \Gamma, f : \chgtgt(\tau, \nu^{\min(s,j)}) \proves t :
    \chgtgt(\tau, \nu^{\min(s,j+1)}) & \begin{array}{cc}\multicolumn{2}{c}{\tgt(\tau) = \nu^s} \\
    j \notin \FSV(\Gamma) & j \notin \SV(\tau) \end{array}
    }
  \end{array}
  \]
  \caption{Rules of the type system~$\lambda^\Diamond$}\label{fig_type_system}
\end{figure*}

\begin{figure}
  \[
  \begin{array}{c}
    \infer{A \sqsubseteq A}{} \\ \\
    \infer{d_\mu^s(\vec{\alpha}) \sqsubseteq d_\mu^{s'}(\vec{\beta})}{\alpha_k \sqsubseteq \beta_k & s \le s'} \quad
    \infer{d_\nu^s(\vec{\alpha}) \sqsubseteq d_\nu^{s'}(\vec{\beta})}{\alpha_k \sqsubseteq \beta_k & s \ge s'} \\ \\
    \infer{\forall i . \tau \sqsubseteq \forall i . \tau'}{\tau \sqsubseteq \tau'} \quad
    \infer{\alpha \to \beta \sqsubseteq \alpha' \to \beta'}{\alpha'
      \sqsubseteq \alpha & \beta \sqsubseteq \beta'}
  \end{array}
  \]
  \caption{Subtyping rules}\label{fig_subtyping}
\end{figure}

We now briefly explain the typing rules. The rules (ax), (sub), (lam),
(app), (inst), (gen) are standard. The rule (con) allows to type
constructors of (co)inductive types. It states that if each
argument~$t_k$ of the constructor~$c$ of a (co)inductive type~$\rho$
may be assigned an appropriate type with the size of the recursive
occurrences of~$\rho$ being~$s$, then $c t_1 \ldots t_n$ has type
$\rho^{s+1}$. For instance, for the type of lists of natural
numbers~$\List(\Nat)$, the rule (con) says that if $x : \Nat$ and
$y : \List^i(\Nat)$ then $\cons\,x\,y : \List^{i+1}(\Nat)$.

The (case) rule allows to type case expressions. If the decorated
term~$t$ that is matched on has a (co)inductive type~$\rho^{s+1}$, and
for each $k=1,\ldots,n$ under the assumption that the arguments of the
constructor~$c_k$ have appropriate types (with the recursive
occurrences of~$\rho$ having size~$s$) the branch~$t_k$ may be given
the type~$\tau$, then the case expression has type $\tau$.

The (fix) rule allows to type recursive fixpoint definitions. It
essentially requires that we may type the body~$t$ under the
assumption that~$f$ already ``works'' for smaller elements.

The (cofix) rule allows to type corecursive fixpoint
definitions. Essentially, it requires that we may type the body~$t$
under the assumption that~$f$ already produces a smaller coinductive
object, i.e., that if~$f$ produces an object defined up to depth~$j$
then~$t$ produces an object defined up to depth~$j+1$. The size
variable $j$ in $\cofix^j f : \tau . t$ may occur
in~$t$. Example~\ref{ex_stream_processors_2} below shows how this may
be used.

\begin{defi}
  Let $\Ys = (\lambda x . \lambda f . f (x x f)) (\lambda x . \lambda
  f . f (x x f))$ be the Turing fixpoint combinator. Note that $\Ys t
  \reduces t (\Ys t)$ for any term~$t$.

  The \emph{erasure}~$\erase{t}$ of a decorated term~$t$ is defined
  inductively:
  \begin{itemize}
  \item $\erase{x} = x$, $\erase{c} = c$,
  \item $\erase{\lambda x : \tau . t} = \lambda x . \erase{t}$, $\erase{\Lambda i . t} = \erase{t}$,
  \item $\erase{t_1 t_2} = \erase{t_1} \erase{t_2}$, $\erase{t s} = \erase{t}$,
  \item $\erase{\case(t; \{c_k\vec{x} \To t_k\})} = \case(\erase{t};
    \{c_k\vec{x} \To \erase{t_k}\})$,
  \item $\erase{\fix f : \tau . t} = \Ys (\lambda f . \erase{t})$,
    $\erase{\cofix f : \tau . t} = \Ys (\lambda f . \erase{t})$.
  \end{itemize}
\end{defi}

\subsection{Examples}

In this section, we give a few examples of typing derivations in the
system~$\lambda^\Diamond$. For the sake of readability, we only
indicate how to derive the typings. It is straightforward but tedious
to translate the examples into the exact formalism
of~$\lambda^\Diamond$.

\begin{exa}
  We reuse the definitions of~$\Nat$ and~$\Strm$ from
  Example~\ref{ex_nat_strm_types} (see also
  Example~\ref{ex_nat_strm}). Consider the function
  \[
    \tail = \Lambda i . \lambda s : \Strm^{i+1} . \case(s; \{\cons\,x\,t \To t\})
  \]
  To type $\tail$ we use the (gen), (lam) and (case) rules. Assume
  $s : \Strm^{i+1}$. To type the match we need to type the
  branch. Assuming $x : \Nat$ and $t : \Strm^{i}$ we have
  $t : \Strm^{i}$, so the match has type $\Strm^{i}$ by the (case)
  rule. Hence, by the (lam) and (gen) rules we obtain
  $\proves \tail : \forall i . \Strm^{i+1} \to \Strm^{i}$.

  Similarly, the function
  \[
    \head = \Lambda i . \lambda s : \Strm^{i+1} . \case(s; \{\cons\,x\,t \To x\})
  \]
  may be assigned the type $\forall i . \Strm^{i+1} \to \Nat$.
\end{exa}

\begin{exa}\label{ex_stream_processors_2}
  We return to the stream processors from
  Example~\ref{ex_stream_processors_type} and
  Example~\ref{ex_stream_processors}.

  The function~$\run$ which runs a stream processor on a stream is
  defined by:
  \[
    \begin{array}{rcl}
      \run &=& \cofix \run : \SP \to \Strm \to
               \Strm . \\ & & \lambda x : \SP . \lambda y : \Strm
                              .\\ & & \case(x;\{\out\,z \To \runi\,z\,y\})
    \end{array}
  \]
  where
  \[
    \begin{array}{rcl}
      \runi &=& \fix \runi : \SPi(\SP) \to \Strm \to
                \Strm^{j+1} . \\ & & \lambda z : \SPi(\SP)
    . \lambda y : \Strm . \\ & & \case(z; \{\get\,f \To \runi\,(f (\head\,\infty\,
    y))\, (\tail\,\infty\,y),\; \mput\,n\,x' \To n :: \run\,x'\,y\})
    \end{array}
  \]
  Recall that~$\Strm$ with no decorations is an abbreviation for
  $\Strm^\infty$.

  We have $\proves \run : \SP \to \Strm \to \Strm$. Indeed, to use the
  (cofix) typing rule assume
  \[
    \run : \SP \to \Strm \to \Strm^{j}
  \]
  and $x : \SP$ and $y : \Strm$.
  \begin{itemize}
  \item To type $\runi$ assume
    $\runi : \SPi^k(\SP) \to \Strm \to
    \Strm^{j+1}$ and $z : \SPi^{k+1}(\SP)$ and
    $y : \Strm$. We apply the (case) rule to type the match
    inside~$\runi$. For this purpose we need to type both branches.
    \begin{itemize}
    \item Assuming $f : \Nat \to \SPi^k(\SP)$, we have
      $\runi\,(f (\head\,\infty\,y))\,(\tail\,\infty\, y) :
      \Strm^{j+1}$ by the (inst), (sub) and (app) rules (note that
      $\infty + 1 \le \infty$ on size expressions).
    \item Assuming $n : \Nat$ and $x' : \SP$, we have
      $\run\,x'\,y : \Strm^{j}$ by the (app) rule. Thus
      $n :: \run\,x'\,y : \Strm^{j+1}$ by (con).
    \end{itemize}
    Hence the match has type $\Strm^{j+1}$ by (case). Thus
    \[
      \runi : \SPi(\SP) \to \Strm \to \Strm^{j+1}
    \]
    by the (fix) typing rule.
  \item To type the match inside $\run$ we use the (case) rule. Under
    the assumption $z : \SPi(\SP)$ the term $\runi\,z\,y$ has type
    $\Strm^{j+1}$ by the (app) rule. Hence the match has type
    $\Strm^{j+1}$ by the (case) rule.
  \end{itemize}
  Now using the (lam) rule we conclude that under the assumption
  \[
    \run : \SP \to \Strm \to \Strm^{j}
  \]
  the body of~$\run$ may be typed with
  $\SP \to \Strm \to \Strm^{j+1}$. Hence
  $\run : \SP \to \Strm \to \Strm$ by the (cofix) rule.
\end{exa}

\begin{rem}
  Strictly speaking, it is possible to type non-productive terms in
  our system. For instance, the term $t = \cofix f : \Strm^0 . f$ has
  type~$\Strm^0$. However, this is not a problem and it agrees with an
  intuitive interpretation of the type system: if $\proves t :
  \Strm^0$ then~$t$ should produce at least~$0$ elements of a stream,
  which does not put any restrictions on~$t$. One could exclude such
  terms by requiring that~$s$ in~$\nu^s$ in the (cofix) typing rule
  should tend to infinity when the sizes of the arguments having
  coinductive types tend to infinity. We did not see a compelling
  reason to incorporate this requirement explicitly into the type
  system.
\end{rem}

\subsection{Type checking}

Type checking in~$\lambda^\Diamond$ is decidable and
coNP-complete. Each decorated term has a minimal type, and there
exists a polynomial algorithm to infer (a compact representation of)
the minimal type. Type checking then reduces to deciding the subtyping
relation between the minimal type and the type being checked.

The proof of the following theorem and the details of the type
checking algorithm may be found in Appendix~\ref{app_checking}. We
only briefly mention this theorem as an interesting ancillary
result. We move the details to an appendix, because this result has no
connection with the infinitary rewriting semantics which is the main
theme of this paper.

\begin{thm}\label{thm_type_checking}
  Type checking in the system~$\lambda^\Diamond$ is
  coNP-complete. More precisely, given $\Gamma,t,\tau$ the problem of
  checking whether $\Gamma \proves t : \tau$ is coNP-complete.
\end{thm}

Despite the theoretically high complexity, we believe that the type
checking algorithm is practical. It is based on a polynomial reduction
of the type-checking problem to the validity of a set of constraints
in quantifier-free Presburger arithmetic. Deciding the validity of the
constraints is coNP-complete~\cite{BoroshTreybing1976,Haase2014}, but
in practice may probably be checked using an SMT-solver such
as~Z3~\cite{MouraBjoerner2008} or~CVC4~\cite{CVC4paper}.

\section{Soundness}\label{sec_soundness}

In this section we show soundness: if $\proves t : \tau$ then there
is~$t' \in \valuation{\tau}{}{}$ with $\erase{t} \infred t'$. We show
that soundness of the (cofix) typing rule follows from the
approximation theorem. This is the main result of the present
section. The justification of the remaining rules
of~$\lambda^\Diamond$ is straightforward if a bit tedious.

We first prove a lemma justifying the correctness of the (cofix)
typing rule. This lemma follows from the approximation theorem.

\begin{lem}\label{lem_cofix}
  Let $r = \Ys (\lambda f . t)$ with $\tgt(\tau) = \nu^s$. Let
  $r_0 = r$ and $r_{n+1} = t[r_n/f]$ for $n \in \Nbb$. Let
  $\tau' = \chgtgt(\tau,\nu^{\min(s,j)})$ where
  $j \notin \SV(s,\nu,\tau)$. If for every $n \in \Nbb$ there
  is~$r_n'$ with $r_n \infred r_n' \in \valuation{\tau'}{v[n/j]}{}$,
  then there is~$r'$ with $r \infred r' \in \valuation{\tau}{v}{}$.
\end{lem}

\begin{proof}
  Note that $r \infred r_n$ for $n \in \Nbb$ follows by
  induction, using Lemma~\ref{lem_infred_subst}. Thus also
  $r \infred r_n'$ for $n\in\Nbb$ by
  Lemma~\ref{lem_infred_concat}.

  Without loss of generality assume
  $\tau = \forall i_1 . \nu_1^{i_1} \to \forall i_2 . \nu_2^{i_2} \to
  \nu^s$. Let $\varkappa_1,\varkappa_2 \in \Omega$ and let
  $u_1 \in \valuation{\nu_1}{v[\varkappa_1/i_1]}{\varkappa_1}$ and
  $u_2 \in
  \valuation{\nu_2}{v[\varkappa_1/i_1,\varkappa_2/i_2]}{\varkappa_2}$. Then
  because $j\notin\SV(s,\nu,\tau)$, using Lemma~\ref{lem_val_idom}, we
  conclude that for every $n \in \Nbb$ there is~$r_n''$ with
  $r_n' u_1 u_2 \infred r_n'' \in
  \valuation{\nu}{v[\varkappa_1/i_1,\varkappa_2/i_2]}{\min(m,n)}$
  where $m = v[\varkappa_1/i_1,\varkappa_2/i_2](s)$. It suffices to
  find~$r'$ with
  $r u_1 u_2 \infred r' \in
  \valuation{\nu}{v[\varkappa_1/i_1,\varkappa_2/i_2]}{m}$.

  First assume $m < \omega$. Then
  $r_m'' \in \valuation{\nu}{v[\varkappa_1/i_1,\varkappa_2/i_2]}{m}$,
  so we may take $r' = r_m''$, because
  $r u_1 u_2 \infred r_m' u_1 u_2 \infred r_m''$. So assume
  $m \ge \omega$. Then for every $n \in \Nbb$ we have
  $r_n'' \in
  \valuation{\nu}{v[\varkappa_1/i_1,\varkappa_2/i_2]}{n}$. Since
  $r u_1 u_2 \infred r_n' u_1 u_2 \infred r_n''$ for
  $n \in \Nbb$, by Lemma~\ref{lem_infred_concat} and
  Theorem~\ref{thm_approx} there is~$r'$ with
  $r u_1 u_2 \infred r' \in
  \valuation{\nu}{v[\varkappa_1/i_1,\varkappa_2/i_2]}{\infty}$.
\end{proof}

The next lemma is needed for the justification of the~(sub) subtyping
rule.

\begin{lem}\label{lem_subtyping}
  If $\tau \sqsubseteq \tau'$ then $\valuation{\tau}{v}{} \subseteq
  \valuation{\tau'}{v}{}$.
\end{lem}

\begin{proof}
  Induction on $\tau$. If
  $\tau = d_\mu^s(\vec{\alpha}) \sqsubseteq \tau' =
  d_\mu^{s'}(\vec{\beta})$ then $s \le s'$ and
  $\alpha_i \sqsubseteq \beta_i$. We have
  $\valuation{\tau}{v}{} = \valuation{d_\mu(\vec{\alpha})}{v}{v(s)} =
  \valuation{d_\mu}{\xi,v}{v(s)}$ where
  $\xi(B_i) = \valuation{\alpha_i}{v}{}$. By the inductive hypothesis
  $\valuation{\alpha_i}{v}{} \subseteq \valuation{\beta_i}{v}{}$. Let
  $\xi'(B_i) = \valuation{\beta_i}{v}{}$. Then $\xi \subseteq
  \xi'$. By Lemma~\ref{lem_val_subset} we obtain
  $\valuation{d_\mu}{\xi,v}{v(s)} \subseteq
  \valuation{d_\mu}{\xi',v}{v(s)}$. Also $v(s) \le v(s')$ and
  $\valuation{\tau'}{v}{} = \valuation{d_\mu}{\xi',v}{v(s')}$. Thus
  $\valuation{\tau}{v}{v(s)} = \valuation{d_\mu}{\xi,v}{v(s)}
  \subseteq \valuation{d_\mu}{\xi',v}{v(s)} \subseteq
  \valuation{d_\mu}{\xi',v}{v(s')} = \valuation{\tau'}{v}{}$.

  If
  $\tau = d_\nu^s(\vec{\alpha}) \sqsubseteq \tau' =
  d_\nu^{s'}(\vec{\beta})$ then the argument is analogous to the
  previous case.

  If
  $\tau = \forall i . \tau_1 \sqsubseteq \tau' = \forall i . \tau_1'$
  then $\tau_1 \sqsubseteq \tau_1'$. Thus by the inductive hypothesis
  $\valuation{\tau_1}{v[\varkappa/i]}{} \sqsubseteq
  \valuation{\tau_1'}{v[\varkappa/i]}{}$ for $\varkappa \in
  \Omega$. Hence
  $\valuation{\tau}{v}{} \subseteq \valuation{\tau'}{v}{}$.

  Finally, assume
  $\tau = \tau_1 \to \tau_2 \sqsubseteq \tau' = \tau_1' \to
  \tau_2'$. Then $\tau_1' \sqsubseteq \tau_1$ and
  $\tau_2 \sqsubseteq \tau_2'$. Let $t \in
  \valuation{\tau}{v}{}$. Then for every
  $r \in \valuation{\tau_1}{v}{}$ there is~$t'$ with
  $t r \infred t' \in \valuation{\tau_2}{v}{}$. Let
  $r \in \valuation{\tau_1'}{v}{}$. Since
  $\valuation{\tau_1'}{v}{} \subseteq \valuation{\tau_1}{v}{}$ by the
  inductive hypothesis, there exists~$t'$ with
  $t r \infred t' \in \valuation{\tau_2}{v}{}$. But
  $\valuation{\tau_2}{v}{} \subseteq \valuation{\tau_2'}{v}{}$ by the
  inductive hypothesis. Hence $t \in \valuation{\tau'}{v}{}$.
\end{proof}

\begin{thm}[Soundness]\label{thm_soundness}
  If $\Gamma \proves t : \tau$ with
  $\Gamma = x_1:\tau_1,\ldots,x_n:\tau_n$ then for every size variable
  valuation $v : \Vc_S \to \Omega$ and all
  $t_1 \in \valuation{\tau_1}{v}{}$, \ldots,
  $t_n \in \valuation{\tau_n}{v}{}$ there exists~$t'$ such that
  $\erase{t}[t_1/x_1,\ldots,t_n/x_n] \infred t' \in
  \valuation{\tau}{v}{}$.
\end{thm}

\begin{proof}
  By induction on the length of the derivation of the typing
  judgement, using Lemma~\ref{lem_cofix} and
  Lemma~\ref{lem_subtyping}. Lemma~\ref{lem_cofix} is needed to
  justify the~(cofix) typing rule. The proof is rather long but
  straightforward. The details may be found in an appendix.
\end{proof}

\section{Conclusions}

We introduced an infinitary rewriting semantics for strictly positive
nested higher-order (co)inductive types. This may be seen as a
refinement and generalization of the notion of productivity in term
rewriting to a setting with higher-order functions and with data
specified by nested higher-order inductive and coinductive
definitions. We showed an approximation theorem: $t \infred t_n \in
\valuation{\nu}{v}{n}$ for $n \in \Nbb$ then there exists $t_\infty
\in \valuation{\nu}{v}{\infty}$ such that $t \infred t_\infty$,
where~$\nu$ is a coinductive type.

In the second part of the paper, we defined a type
system~$\lambda^\Diamond$ combining simple types with nested
higher-order (co)inductive types, and using size restrictions
similarly to systems with sized types. We showed how to use the
approximation theorem to prove soundness: if a finite decorated
term~$t$ has type~$\tau$ in the system then its erasure infinitarily
reduces to a $t' \in \valuation{\tau}{}{}$. Together with confluence
modulo~$\Uc$ of the infinitary reduction relation and the stability
of~$\valuation{\tau}{}{}$, this implies that any finite typable term
has a well-defined interpretation in the right type. This provides an
operational interpretation of typable terms which takes into account
the ``limits'' of infinite reduction sequences.

In particular, if a decorated term~$t$ has in the
system~$\lambda^\Diamond$ a simple (co)inductive type~$\rho$ such that
$\valuation{\rho}{}{}$ contains only normal forms, then the
term~$\erase{t}$ is infinitarily weakly normalizing. It then follows
from~\cite{KetemaSimonsen2010} that any outermost-fair, possibly
infinite but weakly continuous, reduction sequence starting
from~$\erase{t}$ ends in a normal form. Intuitively, this means that
any ``fair'' reduction strategy always produces a normal form ``in the
limit''. For instance, $\Strm$ mentioned in the introduction is such a
type, i.e., all terms in~$\valuation{\Strm}{}{}$ are normal forms. If
all elements of~$\valuation{\rho}{}{}$ are additionally finite, as
e.g.~with $\rho = \Nat$, then~$\erase{t}$ is in fact finitarily weakly
normalizing.

We have not shown infinitary weak normalization for terms having
function types. Nonetheless, our interpretation of $t \in
\valuation{\tau_1 \to \tau_2}{}{}$ is very natural and ensures the
productivity of~$t$ regarded as a function: we require that for $u \in
\valuation{\tau_1}{}{}$ there is $u' \in \valuation{\tau_2}{}{}$ with
$t u \infred u'$.

In general, it seems desirable to strengthen our rewriting semantics
so as to require \emph{all} maximal (in some sense) infinitary
reduction sequences to yield a term of the right type, not just the
existence of such a reduction. Or one would want to prove strong
infinitary normalization for erasures of typable terms. This, however,
does not seem easy to establish at present.

Our proof of the approximation theorem is classical. We do not expect
any significant problems to arise in an attempt to constructivise our
development, but we did not pay enough attention to constructivity
issues to claim this with complete certainty.

\bibliography{biblio}{}
\bibliographystyle{plain}

\appendix

\clearpage
\section{Proofs for Section~\ref{sec_approx}}\label{app_proofs}

This section provides the proof of Lemma~\ref{lem_compl}. First, we
need two auxiliary lemmas, which are needed only for the proof of
Lemma~\ref{lem_compl} (they are not used outside of this appendix).

\begin{lem}\label{lem_val_intersection} % needed only for lem_compl
  If~$\tau$ is strictly positive and $\xi_1,\xi_2$ are stable then
  $\valuation{\tau}{\xi_1,v}{} \cap \valuation{\tau}{\xi_2,v}{} =
  \valuation{\tau}{\xi_1\cap\xi_2,v}{}$, where we define
  $(\xi_1\cap\xi_2)(A) = \xi_1(A) \cap \xi_2(A)$ for any type
  variable~$A$.
\end{lem}

\begin{proof}
  Induction on~$\tau$. Note that it suffices to show
  $\valuation{\tau}{\xi_1,v}{}\cap\valuation{\tau}{\xi_2,v}{}\subseteq
  \valuation{\tau}{\xi_1\cap\xi_2,v}{}$, because the inclusion
  in the other direction follows from Lemma~\ref{lem_val_subset}
  (noting that $\xi_1\cap\xi_2 \subseteq \xi_i$ for $i=1,2$).
  \begin{itemize}
  \item If $\tau$ is closed then
    $\valuation{\tau}{\xi_1,v}{}\cap\valuation{\tau}{\xi_2,v}{}
    =
    \valuation{\tau}{\xi_1,v}{}\cap\valuation{\tau}{\xi_1,v}{}
    = \valuation{\tau}{\xi_1\cap\xi_2,v}{}$ by
    Lemma~\ref{lem_val_dom}.
  \item If $\tau=A$ then
    $\valuation{\tau}{\xi_1,v}{}\cap\valuation{\tau}{\xi_2,v}{}
    = \xi_1(A)\cap\xi_2(A) =
    \valuation{\tau}{\xi_1\cap\xi_2,v}{}$.
  \item If $\tau=d_\mu^s(\vec{\alpha})$ then
    $\valuation{\tau}{\xi_1,v}{}\cap\valuation{\tau}{\xi_2,v}{} =
    \valuation{d_\mu}{\xi_1',v}{v(s)}\cap\valuation{d_\mu}{\xi_2',v}{v(s)}$
    where $\xi_n'(B_j) = \valuation{\alpha_j}{\xi_n,v}{}$ and
    $\xi_n'(A') = \xi_n(A')$ for $A \notin \{B_1,\ldots,B_l\}$ and
    $B_1,\ldots,B_l$ are the parameter type variables of~$d_\mu$. By
    induction on~$\varkappa$ we show
    $\valuation{d_\mu}{\xi_1',v}{\varkappa}\cap\valuation{d_\mu}{\xi_2',v}{\varkappa}
    \subseteq \valuation{d_\mu}{\xi_1'\cap\xi_2',v}{\varkappa}$. There
    are three cases.
    \begin{enumerate}
    \item $\varkappa=0$. Then
      $\valuation{d_\mu}{\xi_1',v}{\varkappa}\cap\valuation{d_\mu}{\xi_2',v}{\varkappa}
      = \emptyset\cap\emptyset = \emptyset =
      \valuation{d}{\xi_1'\cap\xi_2',v}{\varkappa}$.
    \item $\varkappa=\varkappa'+1$. Let $t \in
      \valuation{d_\mu}{\xi_1',v}{\varkappa}\cap\valuation{d_\mu}{\xi_2',v}{\varkappa}$. Then
      $t = c u_1 \ldots u_k$ with
      \[
      u_i \in
      \bigcap_{n\in\{1,2\}}\valuation{\sigma_i}{\xi_n'[\valuation{d_\mu}{\xi_n',v}{\varkappa'}/A],v}{}
      \]
      where~$A$ is the recursive type variable of~$d_\mu$ and $c \in
      \Constr(d_\mu)$ and $\ArgTypes(c) = (\sigma_1,\ldots,\sigma_k)$. By the main
      inductive hypothesis $u_i \in
      \valuation{\sigma_i}{\xi,v}{}$ where
      \[
        \xi =
        \bigcap_{n\in\{1,2\}}\xi_n'[\valuation{d_\mu}{\xi_n',v}{\varkappa'}/A].
      \]
      We have
      \[
      \xi(A) =
      \bigcap_{n\in\{1,2\}}\valuation{d_\mu}{\xi_n',v}{\varkappa'} =
      \valuation{d_\mu}{\xi_1'\cap\xi_2',v}{\varkappa'}
      \]
      by the inductive hypothesis. Hence
      \[
      \xi =
      (\xi_1'\cap\xi_2')[\valuation{d_\mu}{\xi_1'\cap\xi_2',v}{\varkappa'}/A].
      \]
      Therefore $t \in
      \valuation{d_\mu}{\xi_1'\cap\xi_2',v}{\varkappa}$.
    \item $\varkappa$ is a limit ordinal. Let
      \[
      t\in\bigcap_{n\in\{1,2\}}\valuation{d_\mu}{\xi_n',v}{\varkappa}=
      \bigcap_{n\in\{1,2\}}\bigcup_{\varkappa_n<\varkappa}\valuation{d_\mu}{\xi_n',v}{\varkappa_n}.
      \]
      Then for each $n\in\{1,2\}$ there is $\varkappa_n<\varkappa$ with
      $t\in\valuation{d_\mu}{\xi_n',v}{\varkappa_n}$. We have
      $\varkappa_n>0$ is a successor ordinal for $n=1,2$, because
      $\valuation{d_\mu}{\xi_n',v}{0} = \emptyset$. Thus $t = c u_1
      \ldots u_k$ with $u_i \in
      \bigcap_{n\in\{1,2\}}\valuation{\sigma_i}{\xi_n'[\valuation{d_\mu}{\xi_n',v}{\varkappa_n-1}/A],v}{}$
      where~$A$ is the recursive type variable of~$d_\mu$ and $c \in
      \Constr(d_\mu)$ and $\ArgTypes(c) = (\sigma_1,\ldots,\sigma_k)$. By the main
      inductive hypothesis $u_i \in \valuation{\sigma_i}{\xi,v}{}$
      where $\xi =
      \bigcap_{n\in\{1,2\}}\xi_n'[\valuation{d_\mu}{\xi_n',v}{\varkappa_n-1}/A]$. Without
      loss of generality assume $\varkappa_1 \le \varkappa_2$. We have
      $\xi(A) =
      \bigcap_{n\in\{1,2\}}\valuation{d_\mu}{\xi_n',v}{\varkappa_n-1}\subseteq\valuation{d_\mu}{\xi_1',v}{\varkappa_2-1}\cap\valuation{d_\mu}{\xi_2',v}{\varkappa_2-1}$. Hence
      by the inductive hypothesis $\xi(A) \subseteq
      \valuation{d_\mu}{\xi_1'\cap\xi_2',v}{\varkappa_2-1}$. Thus by
      Lemma~\ref{lem_val_subset} we have $u_i \in
      \valuation{\sigma_i}{\xi',v}{}$ where $\xi' =
      (\xi_1'\cap\xi_2')[\valuation{d_\mu}{\xi_1'\cap\xi_2',v}{\varkappa_2-1}/A]$. Therefore
      $t \in \valuation{d_\mu}{\xi_1'\cap\xi_2',v}{\varkappa_2}
      \subseteq \valuation{d_\mu}{\xi_1'\cap\xi_2',v}{\varkappa}$.
    \end{enumerate}
    We have thus shown
    $\valuation{\tau}{\xi_1,v}{}\cap\valuation{\tau}{\xi_2,v}{}
    \subseteq \valuation{d_\mu}{\xi_1'\cap\xi_2',v}{v(s)}$. Note that
    $\xi_1',\xi_2'$ are stable by Lemma~\ref{lem_val_stable}. We have
    $(\xi_1'\cap\xi_2')(B_j) = \xi_1'(B_j)\cap\xi_2'(B_j) =
    \valuation{\alpha_j}{\xi_1',v}{}\cap
    \valuation{\alpha_j}{\xi_2',v}{} =
    \valuation{\alpha_j}{\xi_1'\cap\xi_2',v}{} =
    \valuation{\alpha_j}{\xi_1\cap\xi_2,v}{}$ by the inductive
    hypothesis and Lemma~\ref{lem_val_dom}, because we may assume
    $B_1,\ldots,B_l\notin \FV(\alpha_j)$. Hence
    $\valuation{\tau}{\xi_1,v}{}\cap\valuation{\tau}{\xi_2,v}{}
    \subseteq \valuation{d_\mu(\vec{\alpha})}{\xi_1\cap\xi_2,v}{v(s)}
    = \valuation{\tau}{\xi_1\cap\xi_2,v}{}$ by
    Lemma~\ref{lem_val_dom}.
  \item If $\tau=d_\nu^s(\vec{\alpha})$ then
    $\valuation{\tau}{\xi_1,v}{}\cap\valuation{\tau}{\xi_2,v}{} =
    \valuation{d_\nu}{\xi_1',v}{v(s)}\cap\valuation{d_\nu}{\xi_2',v}{v(s)}$
    where $\xi_n'(B_j) = \valuation{\alpha_j}{\xi_n,v}{}$ and
    $\xi_n'(A') = \xi_n(A')$ for $A \notin \{B_1,\ldots,B_l\}$ and
    $B_1,\ldots,B_l$ are the parameter type variables of~$d_\nu$. Note
    that $\xi_1',\xi_2'$ are stable by
    Lemma~\ref{lem_val_stable}. First, by induction on~$\varkappa$ we
    show
    $\valuation{d_\nu}{\xi_1',v}{\varkappa}\cap\valuation{d_\nu}{\xi_2',v}{\varkappa}
    \subseteq \valuation{d_\nu}{\xi_1'\cap\xi_2',v}{\varkappa}$. If
    $\varkappa=0$ then
    $\valuation{d_\nu}{\xi_1',v}{\varkappa}\cap\valuation{d_\nu}{\xi_2',v}{\varkappa}
    = \Tb^\infty \cap \Tb^\infty = \Tb^\infty =
    \valuation{d_\nu}{\xi_1'\cap\xi_2',v}{\varkappa}$. So assume
    $\varkappa=\varkappa'+1$. Let
    $t \in
    \valuation{d_\nu}{\xi_1',v}{\varkappa}\cap\valuation{d_\nu}{\xi_2',v}{\varkappa}$. Then
    $t = c u_1 \ldots u_k$ with
    $u_i \in
    \bigcap_{n\in\{1,2\}}\valuation{\sigma_i}{\xi_n'[\valuation{d_\nu}{\xi_n',v}{\varkappa'}/A],v}{}$
    where~$A$ is the recursive type variable of~$d_\nu$ and
    $c \in \Constr(d_\nu)$ and $\ArgTypes(c)=(\sigma_1,\ldots,\sigma_k)$. By
    Lemma~\ref{lem_val_stable} and the main inductive hypothesis
    $u_i \in \valuation{\sigma_i}{\xi,v}{}$ where
    $\xi =
    \bigcap_{n\in\{1,2\}}\xi_n'[\valuation{d_\nu}{\xi_n',v}{\varkappa'}/A]$. By
    the inductive hypothesis
    \[
    \xi(A) = \valuation{d_\nu}{\xi_1',v}{\varkappa'} \cap
    \valuation{d_\nu}{\xi_2',v}{\varkappa'} =
    \valuation{d_\nu}{\xi_1'\cap\xi_2',v}{\varkappa'}.
    \]
    Hence $t \in
    \valuation{d_\nu}{\xi_1'\cap\xi_2',v}{\varkappa}$. Finally, assume
    $\varkappa$ is a limit ordinal. Then
    \[
    \begin{array}{rcl}
      \bigcap_{n\in\{1,2\}}\valuation{d_\nu}{\xi_n',v}{\varkappa} &=&
      \bigcap_{n\in\{1,2\}}\bigcap_{\varkappa'<\varkappa}\valuation{d_\nu}{\xi_n',v}{\varkappa'}\\
      &=&
      \bigcap_{\varkappa'<\varkappa}\bigcap_{n\in\{1,2\}}\valuation{d_\nu}{\xi_n',v}{\varkappa'}\\
      &\subseteq&
      \bigcap_{\varkappa'<\varkappa}\valuation{d_\nu}{\xi_1'\cap\xi_2',v}{\varkappa'} \\
      &=& \valuation{d_\nu}{\xi_1'\cap\xi_2',v}{\varkappa}.
    \end{array}
    \]
    Now like in the previous point
    \[
    \begin{array}{rcl}
      (\xi_1'\cap\xi_2')(B_j) &=&
      \xi_1'(B_j)\cap\xi_2'(B_j) \\ &=&
      \valuation{\alpha_j}{\xi_1',v}{}\cap
      \valuation{\alpha_j}{\xi_2',v}{} \\ &=&
      \valuation{\alpha_j}{\xi_1'\cap\xi_2',v}{}\\ &=&
      \valuation{\alpha_j}{\xi_1\cap\xi_2,v}{}
    \end{array}
    \]
    by the inductive hypothesis and Lemma~\ref{lem_val_dom}. Hence
    \[
    \begin{array}{rcl}
      \valuation{\tau}{\xi_1,v}{}\cap\valuation{\tau}{\xi_2,v}{}
      &=&
      \valuation{d_\nu}{\xi_1',v}{m}\cap\valuation{d_\nu}{\xi_2',v}{\varkappa} \\
      &\subseteq& \valuation{d_\nu}{\xi_1'\cap\xi_2',v}{\varkappa} \\
      &=& \valuation{d_\nu(\vec{\alpha})}{\xi_1\cap\xi_2,v}{\varkappa} \\
      &=& \valuation{\tau}{\xi_1\cap\xi_2,v}{}
    \end{array}
    \]
    by Lemma~\ref{lem_val_dom}.
  \item Suppose $\tau = \forall i . \tau'$. Let
    $t \in \valuation{\tau}{\xi_1,v}{} \cap
    \valuation{\tau}{\xi_2,v}{}$ and $\varkappa \in \Omega$. There
    are~$t_1,t_2$ with
    $t \infred t_1 \in \valuation{\tau'}{\xi_1,v[\varkappa/i]}{}$ and
    $t \infred t_2 \in \valuation{\tau'}{\xi_2,v[\varkappa/i]}{}$. By
    confluence modulo~$\Uc$ there are $t_1',t_2'$ such that
    $t_1 \infred t_1' \sim_\Uc t_2'$ and $t_2 \infred t_2'$. Using
    Lemma~\ref{lem_val_stable} we obtain
    $t_2' \in \valuation{\tau'}{\xi_1,v[\varkappa/i]}{} \cap
    \valuation{\tau'}{\xi_2,v[\varkappa/i]}{}$. Hence
    $t \infred t_2' \in
    \valuation{\tau'}{\xi_1\cap\xi_2,v[\varkappa/i]}{}$ by the
    inductive hypothesis. Thus
    $t \in \valuation{\tau}{\xi_1\cap\xi_2,v}{}$. This shows
    $\valuation{\tau}{\xi_1,v}{}\cap\valuation{\tau}{\xi_2,v}{}\subseteq\valuation{\tau}{\xi_1\cap\xi_2,v}{}$.
  \item Suppose $\tau=\tau_1\to\tau_2$. Let
    $t \in \valuation{\tau_1 \to \tau_2}{\xi_1,v}{} \cap
    \valuation{\tau_1 \to \tau_2}{\xi_2,v}{}$. Let
    $w \in \valuation{\tau_1}{\xi_1\cap\xi_2,v}{}$. We have
    $w \in \valuation{\tau_1}{\xi_1,v}{} \cap
    \valuation{\tau_1}{\xi_2,v}{}$. Hence there are $w_1,w_2$ with
    $t w \infred w_1 \in \valuation{\tau_2}{\xi_1,v}{}$ and
    $t w \infred w_2 \in \valuation{\tau_2}{\xi_2,v}{}$. By confluence
    modulo~$\Uc$ there are $w_1',w_2'$ such that $w_1'\sim_\Uc w_2'$
    and $w_i \infred w_i'$. By Lemma~\ref{lem_val_stable} both
    $\valuation{\tau_2}{\xi_1,v}{}$,$\valuation{\tau_2}{\xi_2,v}{}$
    are stable, and thus so is
    $\valuation{\tau_2}{\xi_1,v}{} \cap
    \valuation{\tau_2}{\xi_2,v}{}$. Hence
    $w_1' \in \valuation{\tau_2}{\xi_1,v}{} \cap
    \valuation{\tau_2}{\xi_2,v}{}$. By the inductive hypothesis
    $w_1' \in \valuation{\tau_2}{\xi_1\cap\xi_2,v}{}$. This shows
    $\valuation{\tau}{\xi_1,v}{}\cap\valuation{\tau}{\xi_2,v}{}\subseteq\valuation{\tau}{\xi_1\cap\xi_2,v}{}$.\qedhere
  \end{itemize}
\end{proof}

\begin{lem}\label{lem_mu_varkappa} % needed only for lem_compl
  If $\xi_1,\xi_2$ are stable then
  $\valuation{\mu}{\xi_1,v}{\varkappa_1} \cap
  \valuation{\mu}{\xi_2,v}{\varkappa_2} \subseteq
  \valuation{\mu}{\xi_2,v}{\varkappa_1}$.
\end{lem}

\begin{proof}
  Induction on~$\varkappa_1$. We may assume $\varkappa_1 <
  \varkappa_2$. If $\varkappa_1 = 0$ then
  $\valuation{\mu}{\xi_1,v}{\varkappa_1} \cap
  \valuation{\mu}{\xi_2,v}{\varkappa_2} = \emptyset \cap
  \valuation{\mu}{\xi_2,v}{\varkappa_2} = \emptyset =
  \valuation{\mu}{\xi_2,v}{\varkappa_1}$.

  If $\varkappa_1$ is a limit ordinal then there exists
  $\varkappa_0 < \varkappa_1$ with
  $t \in \valuation{\mu}{\xi,v}{\varkappa_0}$ and we may use the
  inductive hypothesis.

  If $\varkappa_1=\varkappa_1'+1$ then we may
  assume $\varkappa_2=\varkappa_2'+1$. Let $t \in
  \valuation{\mu}{\xi_1,v}{\varkappa_1} \cap
  \valuation{\mu}{\xi_2,v}{\varkappa_2}$. Then $t = c t_1
  \ldots t_k$ with $t_i \in \valuation{\sigma_i}{\zeta_1,v}{}
  \cap \valuation{\sigma_i}{\zeta_2,v}{}$, $c \in \Constr(\mu)$,
  $\ArgTypes(c) = (\sigma_1,\ldots,\sigma_k)$ and (using
  Lemma~\ref{lem_val_dom})
  \[
  \zeta_l=\xi_l[\valuation{\mu}{\xi_l,v}{\varkappa_l'}/A,\valuation{\alpha_1}{\xi_l,v}{}/B_1,\ldots,\valuation{\alpha_k}{\xi_l,v}{}/B_k]
  \]
  for $l=1,2$, and $\mu = d_\mu(\vec{\alpha})$, and $B_1,\ldots,B_k$
  are the parameter type variables of~$d_\mu$, and~$A$ is the
  recursive type variable of~$d_\mu$. By Lemma~\ref{lem_val_stable}
  the valuations~$\zeta_1,\zeta_2$ are stable. By
  Lemma~\ref{lem_val_intersection} we have $t_i \in
  \valuation{\sigma_i}{\zeta_1\cap\zeta_2,v}{}$. Using the inductive
  hypothesis, Lemma~\ref{lem_val_intersection} and
  Lemma~\ref{lem_val_subset} we conclude that $t_i \in
  \valuation{\sigma_i}{\zeta,v}{}$ where
  \[
  \zeta=\xi'[\valuation{\mu}{\xi_2,v}{\varkappa_1'}/A,\valuation{\alpha_1}{\xi_1\cap\xi_2,v}{}/B_1,\ldots,\valuation{\alpha_k}{\xi_1\cap\xi_2,v}{}/B_k].
  \]
  By Lemma~\ref{lem_val_subset} we have $\zeta\subseteq\zeta'$ where
  \[
  \zeta'=\xi_2[\valuation{\mu}{\xi_2,v}{\varkappa_1'}/A,\valuation{\alpha_1}{\xi_2,v}{}/B_1,\ldots,\valuation{\alpha_k}{\xi_2,v}{}/B_k].
  \]
  Hence $t_i \in \valuation{\sigma_i}{\zeta',v}{}$ by
  Lemma~\ref{lem_val_subset}. But this by definition implies
  $t \in \valuation{\mu}{\xi_2,v}{\varkappa_1}$.
\end{proof}

{ \renewcommand{\thethm}{\ref{lem_compl}}
\begin{lem}
  If $\Xi=\{\xi_n\}_{n\in\Nbb}$ is $\nu$-hereditary with~$v$ and
  semi-complete with~$Z,\iota$, and $\{t_n\}_{n\in\Nbb}$ is a
  $\tau,Z$-sequence (and thus a $\tau,\Xi$-sequence by
  Lemma~\ref{lem_approx_family}), then
  \[
    t_\infty = f^\nu(\tau,\Xi,\{t_n\}_{n\in\Nbb}) \in
    \valuation{\tau}{\iota,v}{}.
  \]
\end{lem}
\addtocounter{thm}{-1}}

\begin{proof}
  We proceed by induction on~$\tau$. So let $Z=\{\zeta_n\}_{n\in\Nbb}$
  be stable and let $\Xi=\{\xi_n\}_{n\in\Nbb}$ be $\nu$-hereditary
  with~$v$ and semi-complete with~$Z,\iota$, and let
  $\{t_n\}_{n\in\Nbb}$ be a $\tau,Z$-sequence. By the definition
  of~$t_\infty$ there are the following possibilities.
  \begin{itemize}
  \item If $\tau$ is closed then $t_\infty = t_0 \in
    \valuation{\tau}{\xi_0,v}{} =
    \valuation{\tau}{\iota,v}{}$ by Lemma~\ref{lem_val_dom}.
  \item If $\tau = A$ then $t_\infty \in \iota(A) =
    \valuation{\tau}{\iota,v}{}$ because~$\Xi$ is semi-complete
    with~$Z,\iota$.
  \item If $\tau = \mu^\infty$ with $\mu = d_\mu(\vec{\alpha})$ then
    let $\Xi' = \Xi\rval{\Tc}$ where $\Tc=\{\tau_{A'}\}_{A'\in V_T}$
    with $\tau_A = \tau$, $\tau_{B_j} = \alpha_j$ and $\tau_{A'} = A'$
    for $A' \notin \{A,B_1,\ldots,B_l\}$, where $B_1,\ldots,B_l$ are
    the parameter type variables of~$d_\mu$, and~$A$ is the recursive
    type variable of~$d_\mu$. Note that~$\Xi'$ is $\nu$-hereditary,
    because~$\Xi$ is. Let $Z'=\{\zeta_n'\}_{n\in\Nbb}$ where
    $\zeta_n' =
    \zeta_n[\valuation{\alpha_1}{\zeta_n,v}{}/B_1,\ldots,\valuation{\alpha_l}{\zeta_n,v}{}/B_l]$.
    Note that $Z' \subseteq \Xi'$ follows from
    Lemma~\ref{lem_val_subset}, because $\zeta_n \subseteq \xi_n$ and
    thus
    $\valuation{\alpha_j}{\zeta_n,v}{} \subseteq
    \valuation{\alpha_j}{\xi_n,v}{}$. Also,~$Z'$ is stable by
    Lemma~\ref{lem_val_stable}, because~$Z$ is. Let
    $\iota'(A') = \valuation{\tau_{A'}}{\iota,v}{}$ for any~$A'$. We
    show the following.
    \begin{itemize}
    \item[$(\star)$] Let $X=\{\chi_n\}_{n\in\Nbb}$ be such
      that $\chi_n(A') = \zeta_n'(A')$ for $A'\ne A$. If~$\Xi'$ is
      semi-complete with~$X,\iota'$ then~$\Xi'$ is semi-complete with
      $X',\iota'$ where $X'=\{\chi_n'\}_{n\in\Nbb}$ and
      \[
      \chi_n'=\chi_n[\Phi_{d_\mu,\zeta_n',v}(\chi_n(A))/A].
      \]
    \end{itemize}
    First note that because $\Xi'$ is semi-complete with $X,\iota'$ we
    have $X \subseteq \Xi'$, so $\chi_n(A) \subseteq \xi_n'(A) =
    \valuation{\tau_A}{\xi_n,v}{} = \valuation{\mu}{\xi_n,v}{\infty} =
    \valuation{d_\mu}{\xi_n',v}{\infty}$ by Lemma~\ref{lem_val_dom}
    because $\xi_n'(B_j) = \valuation{\alpha_j}{\xi_n,v}{}$. Also
    $\zeta_n' \subseteq \xi_n'$ because $Z' \subseteq \Xi'$. Therefore
    \[
    \begin{array}{rcl}
      \chi_n'(A) &=& \Phi_{d_\mu,\zeta_n',v}(\chi_n(A)) \\
      &\subseteq&
      \Phi_{d_\mu,\xi_n',v}(\valuation{d_\mu}{\xi_n',v}{\infty}) \\
      &=& \valuation{d_\mu}{\xi_n',v}{\infty} \\
      &=& \xi_n'(A)
    \end{array}
    \]
    by Lemma~\ref{lem_val_subset}. Thus $X' \subseteq \Xi'$. Note
    that~$X'$ is stable by the third point in
    Lemma~\ref{lem_val_stable}. It remains to show that for any~$A'$
    and any $A',X'$-sequence $\{w_n\}_{n\in\Nbb}$ we have
    $w_\infty = f^\nu(A',\Xi',\{w_n\}_{n\in\Nbb}) \in \iota'(A')$. If
    $A'\ne A$ then $\{w_n\}_{n\in\Nbb}$ is also a $A',X$-sequence, so
    $w_\infty \in \iota'(A')$ follows from the fact that~$\Xi'$ is
    semi-complete with $X,\iota'$. If $A'=A$ then
    $w_n \in \Phi_{d_\mu,\zeta_n',v}(\chi_n(A))$ for $n \in
    \Nbb$. Therefore there exists~$c \in \Constr(\mu)$ such that
    $w_n = c w_n^1 \ldots w_n^k$ and $w_n^i \infred w_{n+1}^i$ and
    $w_n^i \in \valuation{\sigma_i}{\zeta_n'[\chi_n(A)/A],v}{}$ where
    $\ArgTypes(c) = (\sigma_1,\ldots,\sigma_k)$. Because
    $\chi_n(A')=\zeta_n'(A')$ for $A'\ne A$, we have
    $\zeta_n'[\chi_n(A)/A] = \chi_n$. Hence
    $w_n^i \in \valuation{\sigma_i}{\chi_n,v}{}$. Thus
    $\{w_n^i\}_{n\in\Nbb}$ is a $\sigma_i,X$-sequence. Note
    that~$\sigma_i$ is smaller than~$\tau$. Because~$\Xi'$ is
    semi-complete with $X,\iota'$, by the inductive hypothesis we have
    $w_\infty^i = f^\nu(\sigma_i,\Xi',\{w_n^i\}_{n\in\Nbb}) \in
    \valuation{\sigma_i}{\iota',v}{}$. Note that
    \[
    \begin{array}{rcl}
      w_\infty &=& f^\nu(A,\Xi',\{w_n\}_{n\in\Nbb}) \\
      &=& f^\nu(\tau_{A},\Xi,\{w_n\}_{n\in\Nbb}) \\
      &=& f^\nu(\tau,\Xi,\{w_n\}_{n\in\Nbb}) \\
      &=& c w_\infty^1 \ldots w_\infty^k.
    \end{array}
    \]
    Hence $w_\infty \in \Phi_{d_\mu,\iota',v}(\iota'(A))$. We
    have $\iota'(A) = \valuation{\tau}{\iota,v}{} =
    \valuation{d_\mu}{\iota',v}{\infty}$ by
    Lemma~\ref{lem_val_dom} because $\iota'(B_j) =
    \valuation{\alpha_j}{\iota,v}{}$ for $j=1,\ldots,l$. Hence
    $w_\infty \in
    \Phi_{d_\mu,\iota',v}(\valuation{d_\mu}{\iota',v}{\infty})
    = \valuation{d_\mu}{\iota',v}{\infty} =
    \valuation{\tau}{\iota,v}{} = \iota'(A)$. We have thus
    shown~$(\star)$.

    Let $Z^\varkappa = \{\zeta^\varkappa_n\}_{n\in\Nbb}$ be such that
    $\zeta^\varkappa_n=\zeta_n'[\valuation{\mu}{\zeta_n,v}{\varkappa}/A]$. Then
    $Z^\varkappa \subseteq \Xi'$ follows from
    Lemma~\ref{lem_val_subset}. Also~$Z^\varkappa$ is stable by
    Lemma~\ref{lem_val_stable}, because~$Z,Z'$ are. By induction
    on~$\varkappa$ we show that~$\Xi'$ is semi-complete with
    $Z^\varkappa,\iota'$. We distinguish three cases.
    \begin{itemize}
    \item $\varkappa=0$. Then $\zeta^\varkappa_n =
      \zeta_n'[\valuation{\mu}{\zeta_n,v}{0}/A] =
      \zeta_n'[\emptyset/A]$. We show that for every~$A'$ and every
      $A',Z^0$-sequence $\{w_n\}_{n\in\Nbb}$ we have $w_\infty =
      f^\nu(\tau_{A'},\Xi,\{w_n\}_{n\in\Nbb}) \in \iota'(A') =
      \valuation{\tau_{A'}}{\iota,v}{}$. Let $\{w_n\}_{n\in\Nbb}$ be a
      $A',Z^0$-sequence. If $A'\ne A$ then
      $\zeta_n^0(A')=\zeta_n'(A')=\valuation{\tau_{A'}}{\zeta_n,v}{}$,
      so $\{w_n\}_{n\in\Nbb}$ is a $\tau_{A'},Z$-sequence. Moreover,
      $\tau_{A'}=\alpha_j$ or $\tau_{A'}=A'$, so~$\tau_{A'}$ has
      smaller size than~$\tau$. Thus by the (main) inductive
      hypothesis we have $f^\nu(\tau_{A'},\Xi,\{w_n\}_{n\in\Nbb}) \in
      \valuation{\tau_{A'}}{\iota,v}{}$, i.e., $w_\infty \in
      \iota'(A')$. If $A'=A$ then $\zeta^0_n(A) = \emptyset$, so there
      does not exist a $A,Z^0$-sequence.
    \item $\varkappa=\varkappa'+1$. Then
      \[
      \begin{array}{rcl}
        \zeta^\varkappa_n &=&
        \zeta_n'[\valuation{\mu}{\zeta_n,v}{\varkappa'+1}/A] \\
        &=&
        \zeta_n'[\Phi_{d_\mu,\zeta_n',v}(\valuation{\mu}{\zeta_n,v}{\varkappa'})/A]
        \\
        &=&
        \zeta^{\varkappa'}_n[\Phi_{d_\mu,\zeta_n',v}(\zeta^{\varkappa'}_n(A))/A].
      \end{array}
      \]
      By the inductive hypothesis~$\Xi'$ is semi-complete with
      $Z^{\varkappa'},\iota'$. Hence~$\Xi'$ is semi-complete with
      $Z^\varkappa,\iota'$ by~$(\star)$.
    \item $\varkappa$ is a limit ordinal. We need to show that for
      all~$A'$ and every
      $A',Z^\varkappa$-sequence~$\{w_n\}_{n\in\Nbb}$ we have
      \[
      \begin{array}{rcl}
        w_\infty &=& f^\nu(A',\Xi',\{w_n\}_{n\in\Nbb}) \\
        &=&
        f^\nu(\tau_{A'},\Xi,\{w_n\}_{n\in\Nbb}) \in \iota'(A').
        \end{array}
      \]
      If $A'\ne A$ then the argument is the same as the one used in
      showing that~$\Xi'$ is semi-complete with~$Z^0,\iota'$. So
      assume $A'=A$. Then $w_n \in \zeta^\varkappa_n(A) =
      \valuation{\mu}{\zeta_n,v}{\varkappa}$ for $n \in
      \Nbb$. Since~$\varkappa$ is a limit ordinal, for each $n \in
      \Nbb$ there is~$\varkappa_n < \varkappa$ such that $w_n \in
      \valuation{\mu}{\zeta_n,v}{\varkappa_n}$. Because $w_0 \infred
      w_n$ for $n \in \Nbb$ and $\zeta_0,\zeta_n$ are stable, by
      Lemma~\ref{lem_mu_varkappa} we obtain $w_n \in
      \valuation{\mu}{\zeta_n,v}{\varkappa_0}$ for $n \in \Nbb$. So
      $\{w_n\}_{n\in\Nbb}$ is a $A,Z^{\varkappa_0}$-sequence. By the
      inductive hypothesis~$\Xi'$ is semi-complete with
      $Z^{\varkappa_0},\iota'$, so $w_\infty \in \iota'(A)$.
    \end{itemize}

    Now taking $\varkappa=\infty$ we conclude that~$\Xi'$ is
    semi-complete with $Z^\infty,\iota'$. Note that $\zeta^\infty_n(A)
    = \valuation{\tau}{\zeta_n,v}{}$ for $n\in\Nbb$. Hence
    $\{t_n\}_{n\in\Nbb}$ is a $A,Z^\infty$-sequence, because it is a
    $\tau,Z$-sequence and $\tau_A = \tau$. Therefore $t_\infty =
    f^\nu(\tau,\Xi,\{t_n\}_{n\in\Nbb}) \in \iota'(A) =
    \valuation{\tau}{\iota,v}{}$.
  \item If $\tau = {\nu_0}^\infty$ with ${\nu_0} =
    d_{\nu_0}(\vec{\alpha})$ then let $\Xi' = \Xi\rval{\Tc}$ where
    $\Tc=\{\tau_{A'}\}_{A'\in V_T}$ with $\tau_A = \tau$, $\tau_{B_j}
    = \alpha_j$ and $\tau_{A'} = A'$ for $A' \notin
    \{A,B_1,\ldots,B_l\}$, where $B_1,\ldots,B_l$ are the parameter
    type variables of~$d_{\nu_0}$, and~$A$ is the recursive type
    variable of~$d_{\nu_0}$. Note that~$\Xi'$ is $\nu$-hereditary,
    because~$\Xi$ is. Let $Z'=\{\zeta_n'\}_{n\in\Nbb}$ where
    $\zeta_n'(A') = \valuation{\tau_{A'}}{\zeta_n,v}{}$ for
    all~$A'$. Note that $Z' \subseteq \Xi'$ follows from
    Lemma~\ref{lem_val_subset}, because $\zeta_n \subseteq \xi_n$ and
    thus $\valuation{\tau_{A'}}{\zeta_n,v}{} \subseteq
    \valuation{\tau_{A'}}{\xi_n,v}{}$. Also~$Z'$ is stable by
    Lemma~\ref{lem_val_stable}, because~$Z$ is. Let $\iota'(A') =
    \valuation{\tau_{A'}}{\iota,v}{}$ for any~$A'$. We show the
    following.
    \begin{itemize}
    \item[$(\star)$] Let $\iota_0$ be a type variable valuation such
      that $\iota_0(A') = \iota'(A')$ for $A'\ne A$. If~$\Xi'$ is
      semi-complete with~$Z',\iota_0$ then~$\Xi'$ is semi-complete
      with $Z',\iota_1$ where
      $\iota_1=\iota_0[\Phi_{d_{\nu_0},\iota_0,v}(\iota_0(A))/A]$.
    \end{itemize}
    Since $Z' \subseteq \Xi'$ and~$Z'$ is stable, it suffices to show
    that for every $A',Z'$-sequence~$\{w_n\}_{n\in\Nbb}$ we have
    $w_\infty = f^\nu(A',\Xi',\{w_n\}_{n\in\Nbb}) \in \iota_1(A')$. So
    let~$\{w_n\}_{n\in\Nbb}$ be a $A',Z'$-sequence,
    i.e.~$w_n \in \valuation{A'}{\zeta_n',v}{} = \zeta_n'(A')$ and
    $w_n \infred w_{n+1}$. If $A'\ne A$ then
    $\iota_1(A') = \iota_0(A')$. Because $\Xi'$ is semi-complete with
    $Z',\iota_0$, we have $w_\infty \in \iota_0(A') = \iota_1(A')$. If
    $A'=A$ then
    $w_n \in \valuation{\tau}{\zeta_n,v}{} =
    \Phi_{d_{\nu_0},\zeta_n',v}(\valuation{\tau}{\zeta_n,v}{})$. Hence
    $w_n = c w_n^1 \ldots w_n^k$ and
    $w_n^i \in \valuation{\sigma_i}{\zeta_n',v}{}$ and
    $w_n^i \infred w_{n+1}^i$ where $c \in \Constr({\nu_0})$,
    $\ArgTypes(c) = (\sigma_1,\ldots,\sigma_k)$. Thus
    $\{w_n^i\}_{n\in\Nbb}$ is a $\sigma_i,Z'$-sequence. Because~$\Xi'$
    is semi-complete with $Z',\iota_0$ and~$\sigma_i$ is smaller
    than~$\tau$, by the inductive hypothesis
    $w_\infty^i = f^\nu(\sigma_i,\Xi',\{w_n^i\}_{n\in\Nbb}) \in
    \valuation{\sigma_i}{\iota_0,v}{}$. Note that
    $w_\infty = c w_\infty^1 \ldots w_\infty^k$ by
    Definition~\ref{def_t_infty}. Hence
    $w_\infty \in \Phi_{d_{\nu_0},\iota_0,v}(\iota_0(A)) = \iota_1(A)$
    by Corollary~\ref{cor_phi_dom}. We have thus shown~$(\star)$.

    Let $\iota_0=\iota'[\Tb^\infty/A]$. We show that~$\Xi'$ is
    semi-complete with~$Z',\iota_0$. We have already shown $Z'
    \subseteq \Xi'$ and that~$Z'$ is stable. So let
    $\{w_n\}_{n\in\Nbb}$ be a $A',Z'$-sequence. We show $w_\infty =
    f^\nu(A',\Xi',\{w_n\}_{n\in\Nbb}) \in \iota_0(A')$. We have $w_n
    \in \valuation{A'}{\zeta_n',v}{} =
    \valuation{\tau_{A'}}{\zeta_n,v}{}$. If $A' \ne A$ then $w_n \in
    \zeta_n'(A') = \valuation{\tau_{A'}}{\zeta_n,v}{}$, so
    $\{w_n\}_{n\in\Nbb}$ is a
    $\tau_{A'},Z$-sequence. Since~$\tau_{A'}$ is smaller than~$\tau$
    (because $\tau_{A'} = A'$ or $\tau_{A'} = \alpha_j$) and~$\Xi$ is
    semi-complete with~$Z,\iota$, by the inductive hypothesis
    $w_\infty = f^\nu(A',\Xi',\{w_n\}_{n\in\Nbb}) =
    f^\nu(\tau_{A'},\Xi,\{w_n\}_{n\in\Nbb}) \in
    \valuation{\tau_{A'}}{\iota,v}{} = \iota'(A') = \iota_0(A')$. If
    $A'=A$ then $\iota_0(A) = \Tb^\infty$, so $w_\infty \in
    \iota_0(A)$.

    Now let
    $\iota_\varkappa=\iota'[\valuation{\nu_0}{\iota,v}{\varkappa}/A]$
    for an ordinal $\varkappa$ (recall that
    $\valuation{\nu_0}{\iota,v}{0} = \Tb^\infty$). We show by
    induction on~$\varkappa$ that~$\Xi'$ is semi-complete with
    $Z',\iota_\varkappa$. For $\varkappa=0$ we have shown this in the
    previous paragraph. If $\varkappa=\varkappa'+1$ then this follows
    from~$(\star)$ because
    $\Phi_{d_{\nu_0},\iota_{\varkappa'},v}(\iota_{\varkappa'}(A)) =
    \Phi_{d_{\nu_0},\iota_{\varkappa'},v}(\valuation{\nu_0}{\iota,v}{\varkappa'})
    =
    \Phi_{d_{\nu_0},\iota',v}(\valuation{d_{\nu_0}}{\iota',v}{\varkappa'})
    = \valuation{\nu_0}{\iota,v}{\varkappa'+1}$ by
    Corollary~\ref{cor_phi_dom} and Lemma~\ref{lem_val_dom} (note that
    $\iota_\varkappa(B_j) = \iota'(B_j) =
    \valuation{\alpha_j}{\iota,v}{}$). So let $\varkappa$ be a limit
    ordinal. We have already shown $Z' \subseteq \Xi'$ and that~$Z'$
    is stable. So let $\{w_n\}_{n\in\Nbb}$ be a $A',Z'$-sequence. By
    the inductive hypothesis $\Xi'$ is semi-complete with
    $Z',\iota_{\varkappa'}$ for $\varkappa'<\varkappa$. So if $A'=A$
    then $w_\infty \in
    \bigcap_{\varkappa'<\varkappa}\iota_{\varkappa'}(A) =
    \bigcap_{\varkappa'<\varkappa}\valuation{\nu_0}{\iota,v}{\varkappa'}
    = \valuation{\nu_0}{\iota,v}{\varkappa} = \iota_\varkappa(A)$. If
    $A'\ne A$ then $w_\infty \in \iota_0(A') = \iota_\varkappa(A')$.

    Now because $\{t_n\}_{n\in\Nbb}$ is a $\tau,Z$-sequence and
    $\tau_A=\tau$, the sequence $\{t_n\}_{n\in\Nbb}$ is also a
    $A,Z'$-sequence. Because~$\Xi'$ is semi-complete with
    $Z',\iota_\infty$ and
    \[
    f^\nu(A,\Xi',\{t_n\}_{n\in\Nbb}) =
    f^\nu(\tau,\Xi,\{t_n\}_{n\in\Nbb}) = t_\infty,
    \]
    by Definition~\ref{def_complete} we have $t_\infty \in
    \iota_\infty(A) = \valuation{\nu_0}{\iota,v}{\infty} =
    \valuation{\tau}{\iota,v}{}$.
  \item If $\tau = \forall i . \tau'$ with $i$ fresh then
    $t_\infty = t_0$. We need to show
    $t_0 \in \valuation{\tau}{\iota,v}{}$. Let $\varkappa \in
    \Omega$. For $n\in\Nbb$ we have
    $t_0 \infred t_n \in \valuation{\tau}{\zeta_n,v}{}$, so for each
    $n\in\Nbb$ there exists~$t_n'$ with
    $t_0 \infred t_n' \in
    \valuation{\tau'}{\zeta_n,v[\varkappa/i]}{}$. Because~$Z$ is
    stable, by Lemma~\ref{lem_sequence} there is a sequence
    $\{t_n''\}_{n\in\Nbb}$ such that $t_0 i \infred t_0''$ and
    $t_n'' \infred t_{n+1}''$ and
    $t_n'' \in \valuation{\tau'}{\zeta_n,v[\varkappa/i]}{}$ for
    $n\in\Nbb$. Hence $\{t_n''\}_{n\in\Nbb}$ is a $\tau',Z$-sequence
    (with $v[\varkappa/i]$). The family~$\Xi$ is $\nu$-hereditary
    with~$v[\varkappa/i]$ by
    Lemma~\ref{lem_hereditary_val}. Because~$\Xi$ is also
    semi-complete with $Z,\iota$, by
    Remark~\ref{rem_complete_valuation} and the inductive hypothesis
    there is~$t^\varkappa$ with
    $t_0 \infred t^\varkappa \in
    \valuation{\tau'}{\iota,v[\varkappa/i]}{}$.
    Because~$\varkappa \in \Omega$ was arbitrary, this implies
    $t_0 \in \valuation{\tau}{\iota,v}{}$.
  \item If $\tau = \tau_1 \to \tau_2$ with $\tau_1$ closed
    and~$\tau_2$ strictly positive, then $t_\infty = t_0$. We need to
    show $t_0 \in \valuation{\tau}{\iota,v}{}$. For $n\in\Nbb$ we have
    $t_0 \infred t_n \in \valuation{\tau}{\zeta_n,v}{}$. Let
    $r \in \valuation{\tau_1}{\iota,v}{}$. Because~$\tau_1$ is closed,
    by Lemma~\ref{lem_val_dom} we have
    $r \in \valuation{\tau_1}{\zeta_n,v}{}$ for each $n\in\Nbb$. Hence
    for each $n \in \Nbb$ there is~$t_n'$ with
    $t_0 r \infred t_n r \infred t_n' \in
    \valuation{\tau_2}{\zeta_n,v}{}$. Because~$Z$ is stable, by
    Lemma~\ref{lem_sequence} there is a sequence
    $\{t_n''\}_{n\in\Nbb}$ such that $t_0 r \infred t_0''$ and
    $t_n'' \infred t_{n+1}''$ and
    $t_n'' \in \valuation{\tau_2}{\zeta_n,v}{}$ for $n\in\Nbb$. Hence
    $\{t_n''\}_{n\in\Nbb}$ is a $\tau_2,Z$-sequence. By the inductive
    hypothesis there is $t_\infty' \in \valuation{\tau_2}{\iota,v}{}$
    with $t_0'' \infred t_\infty'$. Since $t_0 r \infred t_0''$, also
    $t_0 r \infred t_\infty'$. We have thus shown
    $t_0 \in \valuation{\tau}{\iota,v}{}$.\qedhere
  \end{itemize}
\end{proof}

\clearpage
\section{Proofs for Section~\ref{sec_soundness}}\label{app_proofs_soundness}

This section provides the details of the proof of
Theorem~\ref{thm_soundness}. First, we need two auxiliary lemmas,
which are needed only for the proof of Theorem~\ref{thm_soundness}
(they are not used outside of this appendix).

\begin{lem}\label{lem_val_subst} % used only in thm_soundness
  $\valuation{\tau[\tau'/A]}{\xi,v}{} =
  \valuation{\tau}{\xi[\valuation{\tau'}{\xi,v}{}/A],v}{}$.
\end{lem}

\begin{proof}
  Induction on~$\tau$. Let $\xi' =
  \xi[\valuation{\tau'}{\xi,v}{}/A]$. If $\tau = \rho^s$ and
  $\rho = d(\vec{\alpha})$ then
  $\valuation{\tau[\tau'/A]}{\xi,v}{} =
  \valuation{d}{\xi[\vec{X}/\vec{B}],v}{v(s)}$ where
  $X_j = \valuation{\alpha_j[\tau'/A]}{\xi,v}{}$ and $\vec{B}$ are the
  parameter type variables of~$d$. By the inductive hypothesis
  $X_j = \valuation{\alpha_j}{\xi',v}{}$. By Lemma~\ref{lem_val_dom}
  we have
  $\valuation{\tau[\tau'/A]}{\xi,v}{} =
  \valuation{d}{\xi[\vec{X}/\vec{B}],v}{v(s)} =
  \valuation{d}{\xi'[\vec{X}/\vec{B}],v}{v(s)} =
  \valuation{\tau}{\xi',v}{}$.

  If $\tau=A$ then $\valuation{\tau[\tau'/A]}{\xi,v}{} =
  \valuation{\tau'}{\xi,v}{} =
  \valuation{\tau}{\xi',v}{}$.

  If $\tau = \forall i . \tau_1$ then let
  $t \in \valuation{\tau[\tau'/A]}{\xi,v}{}$. Let
  $\varkappa \in \Omega$. There is~$t'$ such that
  $t \infred t' \in
  \valuation{\tau_1[\tau'/A]}{\xi,v[\varkappa/i]}{}$.  By the
  inductive hypothesis
  $t' \in
  \valuation{\tau_1}{\xi[\valuation{\tau'}{\xi,v[\varkappa/i]}{}/A],v[\varkappa/i]}{}$.
  This implies
  $t \in \valuation{\tau}{\xi[\valuation{\tau'}{\xi,v}{}/A],v}{}$,
  using Lemma~\ref{lem_val_idom} (we may assume
  $i \notin \FSV(\tau')$). The other direction is analogous.

  If $\tau=\tau_1\to\tau_2$ then let
  $t \in \valuation{\tau[\tau'/A]}{\xi,v}{}$. Let
  $r \in \valuation{\tau_1}{\xi',v}{}$. Then
  $r \in \valuation{\tau_1[\tau'/A]}{\xi,v}{}$ by the inductive
  hypothesis. Thus
  $t r \infred t' \in \valuation{\tau_2[\tau'/A]}{\xi,v}{}$. By the
  inductive hypothesis $t' \in \valuation{\tau_2}{\xi',v}{}$. The
  inclusion in the other direction is analogous.
\end{proof}

\begin{lem}\label{lem_ival_subst} % used only in thm_soundness
  $\valuation{\tau[s/i]}{v}{} = \valuation{\tau}{v[v(s)/i]}{}$.
\end{lem}

\begin{proof}
  Induction on~$\tau$.
\end{proof}

{ \renewcommand{\thethm}{\ref{thm_soundness}}
\begin{thm}[Soundness]
  If $\Gamma \proves t : \tau$ with
  $\Gamma = x_1:\tau_1,\ldots,x_n:\tau_n$ then for every size variable
  valuation $v : \Vc_S \to \Omega$ and all
  $t_1 \in \valuation{\tau_1}{v}{}$, \ldots,
  $t_n \in \valuation{\tau_n}{v}{}$ there exists~$t'$ such that
  $\erase{t}[t_1/x_1,\ldots,t_n/x_n] \infred t' \in
  \valuation{\tau}{v}{}$.
\end{thm}
\addtocounter{thm}{-1}}

\begin{proof}
  By induction on the length of the derivation of the typing
  judgement. We consider the last rule in the derivation.

  \begin{itemize}
  \item[(ax)] If $\Gamma, x : \tau \proves x : \tau$ then the claim
    follows directly from definitions.

  \item[(sub)] Assume $x_1:\tau_1,\ldots,x_n:\tau_n\proves t : \tau'$
    because of $x_1:\tau_1,\ldots,x_n:\tau_n\proves t : \tau$ and
    $\tau \sqsubseteq \tau'$. Let
    $t_1 \in
    \valuation{\tau_1}{v}{}$,\ldots,$t_n\in\valuation{\tau_n}{v}{}$. By
    the inductive hypothesis there is~$t'$ with
    $\erase{t}[t_1/x_1,\ldots,t_n/x_n] \infred t' \in
    \valuation{\tau}{v}{}$. By Lemma~\ref{lem_subtyping} we also have
    $t' \in \valuation{\tau'}{v}{}$.

  \item[(con)] Assume $\Gamma \proves c r_1 \ldots r_n : \rho^{s+1}$
    because of
    $\Gamma \proves r_k : \sigma_k[\rho^s/A][\vec{\alpha}/\vec{B}]$
    for $k=1,\ldots,n$ and $\ArgTypes(c) = (\sigma_1,\ldots,\sigma_n)$
    and $\Def(c) = d$ and $\rho = d(\vec{\alpha})$ and
    $\Gamma=x_1:\tau_1,\ldots,x_m:\tau_m$. Let
    $t_1\in\valuation{\tau_1}{v}{}$,\ldots,$t_m\in\valuation{\tau_m}{v}{}$. By
    the inductive hypothesis for $k=1,\ldots,n$ there is~$r_k'$ with
    $\erase{r_k}[t_1/x_1,\ldots,t_n/x_n] \infred r_k' \in
    \valuation{\sigma_k[\rho^s/A][\vec{\alpha}/\vec{B}]}{v}{}$. By
    Lemma~\ref{lem_val_subst} and Lemma~\ref{lem_val_dom} we have
    $r_k' \in
    \valuation{\sigma_k}{\xi[\valuation{\rho}{v}{v(s)}/A],v}{}$ where
    $\xi(B_j) = \valuation{\alpha_j}{v}{}$. Hence
    \[
      \erase{c r_1 \ldots r_n}[t_1/x_1,\ldots,t_m/x_m] \infred c r_1'
      \ldots r_n' \in \valuation{\rho}{v}{v(s+1)}.
    \]

  \item[(lam)] Assume
    $\Gamma \proves (\lambda x : \alpha . t) : \alpha \to \beta$
    because of $\Gamma,x:\alpha\proves t:\beta$ and
    $\Gamma=x_1:\tau_1,\ldots,x_n:\tau_n$. Let
    $t_1\in\valuation{\tau_1}{v}{}$,\ldots,
    $t_n\in\valuation{\tau_n}{v}{}$. Let
    $r \in \valuation{\alpha}{v}{}$. By the inductive hypothesis there
    is~$t'$ with
    $\erase{t}[t_1/x_1,\ldots,t_n/x_n,r/x] \infred t'
    \in\valuation{\beta}{v}{}$. Hence
    \[
      \erase{\lambda x : \alpha . t}[t_1/x_1,\ldots,t_n/x_n] r \infred t'
      \in \valuation{\beta}{v}{}.
    \]
    This implies
    $\erase{\lambda x : \alpha . t}[t_1/x_1,\ldots,t_n/x_n] \in
    \valuation{\alpha \to \beta}{v}{}$.

  \item[(app)] Assume $\Gamma \proves t t' : \beta$ because of
    $\Gamma \proves t : \alpha \to \beta$ and
    $\Gamma \proves t' : \alpha$ and
    $\Gamma=x_1:\tau_1,\ldots,x_n:\tau_n$. Let
    $t_1\in\valuation{\tau_1}{v}{}$,\ldots,
    $t_n\in\valuation{\tau_n}{v}{}$. By the inductive hypothesis there
    are~$r,r'$ such that
    \[
      \erase{t}[t_1/x_1,\ldots,t_n/x_n] \infred r \in
      \valuation{\alpha\to\beta}{v}{}
    \]
    and
    \[
      \erase{t'}[t_1/x_1,\ldots,t_n/x_n] \infred r' \in
      \valuation{\alpha}{v}{}.
    \]
    Hence there is~$r''$ with
    $\erase{t t'}[t_1/x_1,\ldots,t_n/x_n] \infred r r' \infred r'' \in
    \valuation{\beta}{v}{}$.

  \item[(inst)] Assume $\Gamma \proves t s : \tau[s/i]$ because of
    $\Gamma \proves t : \forall i . \tau$, where
    $\Gamma = x_1:\tau_1,\ldots,x_n:\tau_n$. Let
    $t_1\in\valuation{\tau_1}{v}{}$,\ldots,$t_n\in\valuation{\tau_n}{v}{}$. By
    the inductive hypothesis
    $\erase{t}[t_1/x_1,\ldots,t_n/x_n] \infred t' \in
    \valuation{\forall i . \tau}{v}{}$. Hence there is~$t''$ with
    $t' \infred t'' \in \valuation{\tau}{v[v(s)/i]}{}$. So
    $\erase{t}[t_1/x_1,\ldots,t_n/x_n] \infred t'' \in
    \valuation{\tau[s/i]}{v}{}$ by Lemma~\ref{lem_ival_subst}.

  \item[(gen)] Assume
    $\Gamma \proves \Lambda i . t : \forall i . \tau$ because of
    $\Gamma \proves t : \tau$ with $i \notin \FSV(\Gamma)$ and
    $\Gamma = x_1 : \tau_1,\ldots,x_n:\tau_n$. Let
    $t_1\in\valuation{\tau_1}{v}{}$,\ldots,$t_n\in\valuation{\tau_n}{v}{}$. Let
    $\varkappa \in \Omega$. Since $i \notin \FSV(\Gamma)$ by
    Lemma~\ref{lem_val_idom} we have
    $t_k \in \valuation{\tau_k}{v[\varkappa/i]}{}$ for
    $k=1,\ldots,n$. By the inductive hypothesis there is~$r_\varkappa$
    with
    $\erase{\Lambda i . t}[t_1/x_1,\ldots,t_n/x_n] =
    \erase{t}[t_1/x_1,\ldots,t_n/x_n] \infred r_\varkappa \in
    \valuation{\tau}{v[\varkappa/i]}{}$. This implies
    \[
      \erase{t}[t_1/x_1,\ldots,t_n/x_n] \in \valuation{\forall i
        . \tau}{v}{}.
    \]

  \item[(case)] Assume
    $\Gamma \proves \case(t; \{c_k \vec{x_k} \To t_k \mid
    k=1,\ldots,n\}) : \tau$ because of $\Gamma \proves t : \rho^{s+1}$
    and
    $\Gamma, x_k^1 : \delta_k^1, \ldots, x_k^{n_k} : \delta_k^{n_k}
    \proves t_k : \tau$ and
    $\ArgTypes(c_k) = (\sigma_k^1, \ldots, \sigma_k^{n_k})$ and
    $\delta_k^l = \sigma_k^l[\rho^{s}/A][\vec{\alpha}/\vec{B}]$ and
    $\rho = d(\vec{\alpha})$ and
    $\Gamma=x_1:\tau_1,\ldots,x_m:\tau_m$. Let
    $r_1\in\valuation{\tau_1}{v}{}$,\ldots,$r_m\in\valuation{\tau_m}{v}{}$. By
    the inductive hypothesis there is~$u$ with
    $\erase{t}[r_1/x_1,\ldots,r_m/x_m] \infred u \in
    \valuation{\rho^{s+1}}{v}{} = \valuation{\rho}{v}{v(s+1)} =
    \valuation{\rho}{v}{v(s) + 1}$ (note that we may have
    $v(s+1) = v(s) = \infty$, but then the last equation still holds
    because~$\infty$ is the fixpoint ordinal). Hence
    $u = c_k u_1 \ldots u_{n_k}$ where
    $u_i \in \valuation{\sigma_k^i}{\xi,v}{}$ and
    $\xi(B_j) = \valuation{\alpha_j}{v}{}$ and
    $\xi(A) = \valuation{\rho}{v}{v(s)}$. Then by
    Lemma~\ref{lem_val_subst} we have
    $u_i \in \valuation{\sigma_k^i}{\xi,v}{} =
    \valuation{\delta_k^i}{v}{}$. Hence by the inductive hypothesis
    there is~$w$ with
    $\erase{t_k}[r_1/x_1,\ldots,r_m/x_m,u_1/x_k^1,\ldots,u_{n_k}/x_k^{n_k}]
    \infred w \in \valuation{\tau}{v}{}$. Note that also (we may
    assume $x_k^1,\ldots,x_k^{n_k}\notin\FV(r_1,\ldots,r_m)$):
    \[
      \begin{array}{rll}
        \multicolumn{3}{l}{\case(\erase{t}; \{c_k \vec{x_k}
        \To \erase{t_k} \mid k=1,\ldots,n\})[r_1/x_1,\ldots,r_m/x_m]} \\
        \quad &\infred& \case(u; \{c_k \vec{x_k}
                        \To \erase{t_k}[r_1/x_1,\ldots,r_m/x_m]\}) \\
              &\to&
                    \erase{t_k}[r_1/x_1,\ldots,r_m/x_m,u_1/x_k^1,\ldots,u_{n_k}/x_k^{n_k}]
        \\
              &\infred& w.
      \end{array}
    \]

  \item[(fix)] Assume
    $\Gamma \proves (\fix f : \forall j_1 \ldots j_m . \mu \to \tau
    . t) : \forall j_1 \ldots j_m . \mu \to \tau$
    because of
    \[
      \Gamma,f : \forall j_1 \ldots j_m . \mu^i \to \tau \proves t :
      \forall j_1 \ldots j_m . \mu^{i+1} \to \tau
    \]
    where $i \notin \FSV(\Gamma,\mu,\tau,j_1,\ldots,j_m)$ and
    $\Gamma = x_1:\tau_1,\ldots,x_n:\tau_n$. Let
    $t_1\in\valuation{\tau_1}{v}{},\ldots,t_n\in\valuation{\tau_n}{v}{}$. Let
    $t' = \erase{t}[t_1/x_1,\ldots,t_n/x_n]$ and
    $r = \Ys (\lambda f . t')$. Note that
    $r = \erase{\fix f : \forall j_1 \ldots j_m . \mu \to \tau
      . t}[t_1/x_1,\ldots,t_n/x_n]$. By induction
    on~$\varkappa \in \Omega$ we show
    $r \in \valuation{\forall j_1 \ldots j_m . \mu^i \to
      \tau}{v[\varkappa/i]}{}$. Let
    $\varkappa_1,\ldots,\varkappa_m\in\Omega$. We need to show that
    for every $u \in \valuation{\mu}{v'}{\varkappa}$ there is~$r'$
    with $r u \infred r' \in \valuation{\tau}{v'}{}$, where
    $v' = v[\varkappa_1/j_1,\ldots,\varkappa_n/j_n]$ (recall
    $i \notin \FSV(\mu,\tau)$). There are three cases.
    \begin{itemize}
    \item $\varkappa = 0$. Then
      $\valuation{\mu}{v'}{\varkappa} = \emptyset$.
    \item $\varkappa = \varkappa'+1$. By the inductive hypothesis
      for~$\varkappa$ we have
      $r \in \valuation{\forall j_1 \ldots j_m . \mu^i \to
        \tau}{v[\varkappa'/i]}{}$. By the main inductive hypothesis
      $t'[r/f] \in \valuation{\forall j_1 \ldots j_m . \mu^i \to
        \tau}{v[\varkappa/i]}{}$. Let
      $u \in \valuation{\mu}{v'}{\varkappa}$. Then there is~$r'$ with
      $r u \reduces t'[r/f] u \infred r' \in \valuation{\tau}{v'}{}$.
    \item $\varkappa$ is a limit ordinal. Let
      $u \in \valuation{\mu}{v'}{\varkappa}$. Then
      $u \in \valuation{\mu}{v'}{\varkappa'}$ for some
      $\varkappa' < \varkappa$. By the inductive hypothesis
      for~$\varkappa$ we have
      $r \in \valuation{\forall j_1 \ldots j_m . \mu^i \to
        \tau}{v[\varkappa'/i]}{}$. Hence, there is~$r'$ with
      $r u \infred r' \in \valuation{\tau}{v'}{}$.
    \end{itemize}
    We have thus shown that
    $r \in \valuation{\forall j_1 \ldots j_m . \mu^i \to
      \tau}{v[\varkappa/i]}{}$ for all $\varkappa \in \Omega$. In
    particular, this holds for $\varkappa = \infty$, which implies
    $r \in \valuation{\forall j_1 \ldots j_m . \mu \to \tau}{v}{}$
    because $i \notin \FSV(\mu,\tau)$.

  \item[(cofix)] Assume
    $\Gamma \proves (\cofix^j f : \tau . t) : \tau$ because
    $\Gamma,f:\chgtgt(\tau,\nu^{\min(s,j)}) \proves t :
    \chgtgt(\tau,\nu^{\min(s,j+1)})$ and $\tgt(\tau) = \nu^s$ and
    $j\notin\FSV(\Gamma)$ and $j \notin \SV(\tau)$ and
    $\Gamma = x_1:\tau_1,\ldots,x_n:\tau_n$. Let
    $t_1\in\valuation{\tau_1}{v}{}$,\ldots,$t_n\in\valuation{\tau_n}{v}{}$. Let
    $t' = \erase{t}[t_1/x_1,\ldots,t_n/x_n]$ and
    $r = \Ys (\lambda f . t')$. Let $r_0 = r$ and
    $r_{n+1} = t'[r_n/f]$ for $n \in \Nbb$. Let
    $\tau' = \chgtgt(\tau,\nu^{\min(s,j)})$ and
    $\tau'' = \chgtgt(\tau,\nu^{\min(s,j+1)})$.

    By induction on~$n$ we show that for each~$n$ there is~$r_n'$ with
    $r_n \infred r_n' \in \valuation{\tau'}{v[n/j]}{}$. For $n=0$, we
    have $r_0 = r \in \valuation{\tau'}{v[0/j]}{}$ directly from
    definitions and the fact that $j \notin \SV(\tau)$, because
    $\valuation{\nu^{\min(s,j)}}{v'[0/j]}{} =
    \valuation{\nu}{v'[0/j]}{0} = \Tb^\infty$ for any~$v'$. So
    assume $r_n \infred r_n' \in \valuation{\tau'}{v[n/j]}{}$. Because
    $\Gamma,f:\tau' \proves t : \tau''$, by the main inductive
    hypothesis there is~$t''$ with
    $t'[r_n'/f] \infred t'' \in \valuation{\tau''}{v[n/j]}{}$. We have
    $r_{n+1} = t'[r_n/f] \infred t'[r_n'/f] \infred t''$. Take
    $r_{n+1}' = t''$. Because
    $r_{n+1}' \in \valuation{\tau''}{v[n/j]}{}$ and
    $j\notin\SV(\tau)$, it follows from definitions and
    Lemma~\ref{lem_ival_subst} that
    $r_{n+1}' \in \valuation{\tau'}{v[n+1/j]}{}$.

    We have thus shown that for each $n \in\Nbb$ there exists~$r_n'$
    such that
    \[
      r_n \infred r_n' \in
      \valuation{\chgtgt(\tau,\nu^{\min(s,j)})}{v[n/j]}{}.
    \]
    Now by Lemma~\ref{lem_cofix} there is~$r'$ with
    \[
      \erase{\cofix^j f : \tau . t}[t_1/x_1,\ldots,t_n/x_n] = r
      \infred r' \in \valuation{\tau}{v}{}.\qedhere
    \]
  \end{itemize}
\end{proof}

\clearpage
\section{Type checking and type inference}\label{app_checking}

In this section we show that type checking in~$\lambda^\Diamond$ is
decidable and coNP-complete. First, we show that each decorated term
has a minimal type. We give an algorithm to infer the minimal
type. Type checking then reduces to deciding the subtyping relation
between the minimal type and the type being checked.

\subsection{Minimal typing}

In this section we show that if $t$ is typable in a context~$\Gamma$,
then there exists a minimal type $\Tc(\Gamma;t)$ such that
$\Gamma \proves t : \Tc(\Gamma;t)$ and for every type~$\tau$ with
$\Gamma \proves t : \tau$ we have $\Tc(\Gamma;t) \sqsubseteq \tau$. To
define $\Tc(\Gamma;t)$ we first need the definitions of the
operations~$\sqcup$ and~$\sqcap$ on types.

\begin{defi}
  We define $\tau_1 \sqcup \tau_2$ and $\tau_1 \sqcap \tau_2$
  inductively.
  \begin{itemize}
  \item $(\alpha \to \beta) \sqcup (\alpha' \to \beta') = (\alpha
    \sqcap \alpha') \to (\beta \sqcup \beta')$.
  \item $(\alpha \to \beta) \sqcap (\alpha' \to \beta') = (\alpha
    \sqcup \alpha') \to (\beta \sqcap \beta')$.
  \item
    $(\forall i . \alpha) \sqcup (\forall i . \alpha') = \forall i
    . \alpha \sqcup \alpha'$.
  \item
    $(\forall i . \alpha) \sqcap (\forall i . \alpha') = \forall i
    . \alpha \sqcap \alpha'$.
  \item
    $(d_\mu^s(\vec{\alpha})) \sqcup (d_\mu^{s'}(\vec{\beta})) =
    d_\mu^{\max(s,s')}(\vec{\gamma})$ where
    $\gamma_i=\alpha_i\sqcup\beta_i$.
  \item
    $d_\mu^s(\vec{\alpha}) \sqcap d_\mu^{s'}(\vec{\beta}) =
    d_\mu^{\min(s,s')}(\vec{\gamma})$ where
    $\gamma_i=\alpha_i\sqcap\beta_i$.
  \item
    $d_\nu^s(\vec{\alpha}) \sqcup d_\nu^{s'}(\vec{\beta}) =
    d_\nu^{\min(s,s')}(\vec{\gamma})$ where
    $\gamma_i=\alpha_i\sqcup\beta_i$.
  \item
    $d_\nu^s(\vec{\alpha}) \sqcap d_\nu^{s'}(\vec{\beta}) =
    d_\nu^{\max(s,s')}(\vec{\gamma})$ where
    $\gamma_i=\alpha_i\sqcap\beta_i$.
  \end{itemize}
\end{defi}

\begin{lem}\label{lem_meet_bound}
  $\tau_1 \sqcap \tau_2 \sqsubseteq \tau_i \sqsubseteq \tau_1 \sqcup
  \tau_2$.
\end{lem}

\begin{proof}
  By induction on~$\tau_i$.
\end{proof}

\begin{lem}\label{lem_subtype_subst}
  If~$\tau$ is strictly positive and $\tau \sqsubseteq \tau'$ and
  $\alpha \sqsubseteq \beta$ then
  $\tau[\alpha/A] \sqsubseteq \tau'[\beta/A]$. Conversely, if
  $\tau[\alpha/A] \sqsubseteq \gamma$
  (resp.~$\gamma \sqsubseteq \tau[\alpha/A]$) then
  $\gamma = \tau'[\beta/A]$ with $\tau \sqsubseteq \tau'$
  (resp.~$\tau' \sqsubseteq \tau$) and $\alpha \sqsubseteq \beta$
  (resp.~$\beta \sqsubseteq \alpha$).
\end{lem}

\begin{proof}
  Induction on~$\tau$.
\end{proof}

\begin{lem}\label{lem_size_subst_monotone_2}
  If $s_1 \le s_2$ then $s[s_1/i] \le s[s_2/i]$ and
  $s_1[s/i] \le s_2[s/i]$.
\end{lem}

\begin{proof}
  Induction on the structure of size expressions.
\end{proof}

\begin{defi}
  A subexpression occurrence~$s_0$ in a size expression~$s$ is
  \emph{superfluous} in~$s$ if it does not occur within a
  subexpression of the form~$s_1+1$. For a size expression~$s$
  satisfying $s \ge 1$ we define the size expression~$\overline{s}$ as
  follows. Let $s'$ be obtained from $s$ by replacing with~$0$ each
  superfluous occurrence of a size variable. Then by using obvious
  identities on size expressions ($\max(0, s) = s$,
  $\max(s_1+1,s_2+1)=\max(s_1,s_2)+1$, $\max(\infty,s) = \infty$, etc)
  transform~$s'$ into $\overline{s} + 1$ (note that $s' = 0$ is not
  possible because $s \ge 1$). For any size expression~$s$ we define
  the size expression~$\underline{s}$ as follows. Let~$s''$ be
  obtained from~$s$ by replacing with $i+1$ each superfluous
  occurrence of~$i$. Then by using the identities on size expressions
  transform~$s''$ into~$\underline{s}+1$ (note that $s'' \ge 1$).
\end{defi}

For instance, if $s = \max(i + 1, \min(i, j + 1))$ then $\overline{s}
= i$ and $\underline{s} = \max(i, \min(i, j))$. If $s = 0$ then
$\underline{s} = 0$.

\begin{lem}\label{lem_bar_underline_bound}~
  \begin{enumerate}
  \item If $s \ge 1$ then $s \ge \overline{s} + 1$.
  \item $s \le \underline{s} + 1$.
  \end{enumerate}
\end{lem}

\begin{proof}
  Follows from definitions and Lemma~\ref{lem_size_subst_monotone_2}.
\end{proof}

\begin{lem}\label{lem_bar_underline_le}
  Assume $s_1 \le s_2$.
  \begin{enumerate}
  \item If $s_1 \ge 1$ then $\overline{s_1} \le \overline{s_2}$.
  \item $\underline{s_1} \le \underline{s_2}$.
  \end{enumerate}
\end{lem}

\begin{proof}
  Follows from definitions and Lemma~\ref{lem_size_subst_monotone_2}.
\end{proof}

\begin{lem}\label{lem_bar}
  If $s \ge s_0 + 1$ then $\overline{s} \ge s_0$.
\end{lem}

\begin{proof}
  By the identities
  \[
    \begin{array}{rcl}
      \max(\min(s_1,s_2),s_3) &=& \min(\max(s_1,s_3),\max(s_2,s_3)) \\
      \min(\max(s_1,s_2),s_3) &=& \max(\min(s_1,s_3),\min(s_2,s_3))
    \end{array}
  \]
  we may assume that
  e.g.~$s = \min(\max(i+1,j,\ldots),\max(\ldots),\ldots)$ and
  $s_0 = \max(\min(k+2,i+2,\ldots),\ldots)$. Then by the equivalences
  \[
    \begin{array}{rcl}
      \min(s_1,s_2) \ge s_3 &\leftrightarrow& s_1 \ge s_3 \land s_2 \ge s_3 \\
      \max(s_1,s_2) \le s_3 &\leftrightarrow& s_1 \le s_3 \land s_2 \le s_3
    \end{array}
  \]
  it suffices to consider the case e.g.~$s = \max(s_1,s_2,s_3)$ and
  $s_0 = \min(s_1',s_2')$ with each $s_i,s_i'$ of the form $j + c$ or
  $c$, with~$j$ a size variable and $c \in \Nbb$. Note that the
  operations performed to obtain $s,s_0$ of this form do not affect
  whether the occurrences of size variables are superfluous or not,
  i.e., when transforming $s$ to~$s'$ of the required form analogous
  operations may be simultaneously performed on $\overline{s}$ to
  obtain $\overline{s'}$. So it suffices to show
  $\overline{s} \ge s_0$ for $s,s_0$ of the form as above. Define
  $s_0'$ by replacing in~$s_0$ each
  $i \in \SV(s) \setminus \SV(\overline{s})$ (i.e.~each size variable
  which occurs only superfluously in~$s$) with~$\infty$, and
  simplifying using obvious identities. First note that we may assume
  that some $i \in \SV(s) \setminus \SV(\overline{s})$ occurs
  in~$s_0$, because otherwise $\overline{s} \ge s_0$ follows from
  $s \ge s_0 + 1$ by setting each
  $i \in \SV(s) \setminus \SV(\overline{s})$ to~$0$ and
  simplifying. We have $s_0' \ge s_0$. Thus it suffices to show
  $\overline{s} \ge s_0'$. Assume otherwise, i.e., there is a size
  variable valuation~$v$ such that $v(\overline{s}) < v(s_0')$. Note
  that the values of~$v(\overline{s})$ and~$v(s_0')$ do not depend on
  $v(i)$ for $i \in \SV(s) \setminus \SV(\overline{s})$. Hence, we may
  assume $v(i) = v(\overline{s}) + 1$ for
  $i \in \SV(s) \setminus \SV(\overline{s})$. Then
  $v(s) = \max(v(i),v(\overline{s}+1)) =
  \max(v(\overline{s})+1,v(\overline{s}+1)) = v(\overline{s}) + 1$
  where $i \in \SV(s) \setminus \SV(\overline{s})$ (by
  how~$\overline{s}$ is obtained from~$s$). Also
  $v(s_0) \ge \min(v(i), v(s_0'))$ where
  $i \in \SV(s) \setminus \SV(\overline{s})$. Hence
  $v(s_0) \ge v(\overline{s}) + 1$, because
  $v(\overline{s}) + 1 \le v(s_0')$. Thus
  $v(\overline{s}) + 2 = v(s_0) + 1 \le v(s) = v(\overline{s}) +
  1$. Contradiction.
\end{proof}

\begin{lem}\label{lem_underline}
  If $s \le s_0 + 1$ then $\underline{s} \le s_0$.
\end{lem}

\begin{proof}
  Analogously to the proof of Lemma~\ref{lem_bar}, it suffices to
  consider the case e.g.~$s = \min(s_1,s_2,s_3)$ and
  $s_0 = \max(s_1',s_2')$ with each $s_i,s_i'$ of the form $j + c$ or
  $c$, with~$j$ a size variable and $c \in \Nbb$. Then it suffices to
  show: if $\min(i,s_1,s_2,\ldots) \le s_0 + 1$ then
  $\min(i+1,s_1,s_2,\ldots) \le s_0 + 1$ with~$s_0$ of the form
  $\max(\ldots)$ as above. We may assume
  $i \notin \SV(s_1,s_2,\ldots)$. There are two cases.
  \begin{itemize}
  \item $i\notin \SV(s_0)$. Suppose
    $v(s_0) + 1 < \min(v(i)+1,v(s_1),v(s_2),\ldots)$. Then
    $v'(s_0)+1 = v(s_0)+1 < \min(v(i)+1,v(s_1),v(s_2),\ldots) =
    v'(\min(i,s_1,s_2,\ldots))$ for $v' = v[v(i)+1/i]$. Contradiction.
  \item $s_0 = \max(i+c,s_1',s_2',\ldots)$. Then
    $v(s_0)+1 = \max(v(i)+c,v(s_1'),v(s_2'),\ldots)+1 \ge v(i) + c + 1
    \ge v(i) + 1 \ge v(\min(i+1,s_1,s_2,\ldots))$.\qedhere
  \end{itemize}
\end{proof}

To save on notation we introduce a dummy~$\bot$ type and set
$\bot \sqcup \tau = \tau \sqcup \bot = \tau$. The dummy type~$\bot$ is
not a valid type, it is used only to simplify the presentation of type
sums below. We assume that for every parameter type variable~$B$ of a
(co)inductive definition~$d$ there exists a
constructor~$c \in \Constr(d)$ such that $B \in \FV(\sigma_i)$ for
some~$i$, where $\ArgTypes(c) = (\sigma_1,\ldots,\sigma_n)$. In other
words, we do not allow parameter type variables which do not occur in
any constructor argument types.

\begin{defi}
  For a context~$\Gamma$ and a term~$t$ we inductively define a
  minimal type~$\Tc(\Gamma;t)$ of~$t$ in~$\Gamma$.
  \begin{itemize}
  \item $\Tc(\Gamma,x:\tau;x) = \tau$.
  \item $\Tc(\Gamma;c t_1\ldots t_n) = \mu^{\max(s_1,\ldots,s_n)+1}$
    if $\ArgTypes(c) = (\sigma_1,\ldots,\sigma_n)$ and
    $\mu = d_\mu(\tau_1,\ldots,\tau_m)$ and $\Def(c) = d_\mu$ and
    $\Tc(\Gamma;t_i) =
    \sigma_i'[d_\mu^{s_i}(\alpha_1^i,\ldots,\alpha_m^i)/A][\beta_1^i/B_1,\ldots,\beta_m^i/B_m]$
    (we take $s_i=0$ and $\alpha_j^i = \bot$ if
    $A \notin \FV(\sigma_i)$, and $\beta_j^i = \bot$ if
    $B_j \notin \FV(\sigma_i)$) and $\sigma_i' \sqsubseteq \sigma_i$
    and $\tau_j = \bigsqcup_{i=1}^n(\alpha_j^i\sqcup\beta_j^i)$. Note
    that $\tau_j \ne \bot$ because of our assumption on the
    occurrences of~$B_j$.
  \item $\Tc(\Gamma;c t_1\ldots t_n) = \nu^{\min(s_1,\ldots,s_n)+1}$
    if $\ArgTypes(c) = (\sigma_1,\ldots,\sigma_n)$ and
    $\nu = d_\nu(\tau_1,\ldots,\tau_m)$ and $\Def(c) = d_\nu$ and
    $\Tc(\Gamma;t_i) =
    \sigma_i'[d_\nu^{s_i}(\alpha_1^i,\ldots,\alpha_m^i)/A][\beta_1^i/B_1,\ldots,\beta_m^i/B_m]$
    (we take $s_i=\infty$ and $\alpha_j^i = \bot$ if
    $A \notin \FV(\sigma_i)$, and $\beta_j^i = \bot$ if
    $B_j \notin \FV(\sigma_i)$) and $\sigma_i' \sqsubseteq \sigma_i$
    and $\tau_j = \bigsqcup_{i=1}^n(\alpha_j^i\sqcup\beta_j^i)$.
  \item $\Tc(\Gamma;\lambda x : \alpha . t) = \alpha \to \beta$ if $\Tc(\Gamma,x:\alpha;t) = \beta$.
  \item $\Tc(\Gamma;tt') = \beta$ if
    $\Tc(\Gamma;t) = \alpha \to \beta$ and
    $\Tc(\Gamma;t') \sqsubseteq \alpha$.
  \item $\Tc(\Gamma;ts) = \tau[s/i]$ if
    $\Tc(\Gamma;t) = \forall i . \tau$.
  \item $\Tc(\Gamma;\Lambda i . t) = \forall i . \tau$ if
    $\Tc(\Gamma;t) = \tau$ and $i \notin \FSV(\Gamma)$.
  \item
    $\Tc(\Gamma;\case(t;\{c_k\vec{x_k} \To t_k \mid k=1,\ldots,n\})) =
    \tau$ if $\Tc(\Gamma;t) = \mu^{s}$ and $\mu = d(\vec{\beta})$ and
    $\ArgTypes(c_k) = (\sigma_k^1,\ldots,\sigma_k^{n_k})$ and
    $\delta_k^l =
    \sigma_k^l[\mu^{\underline{s}}/A][\vec{\beta}/\vec{B}]$ and
    $\Tc(\Gamma,x_k^1:\delta_k^1,\ldots,x_k^{n_k}:\delta_k^{n_k}; t_k)
    = \tau_k$ and $\tau = \bigsqcup_{k=1}^n\tau_k$.
  \item
    $\Tc(\Gamma;\case(t;\{c_k\vec{x_k} \To t_k \mid k=1,\ldots,n\})) =
    \tau$ if $\Tc(\Gamma;t) = \nu^{s}$ and $s \ge 1$ and
    $\nu = d(\vec{\beta})$ and
    $\ArgTypes(c_k) = (\sigma_k^1,\ldots,\sigma_k^{n_k})$ and
    $\delta_k^l =
    \sigma_k^l[\nu^{\overline{s}}/A][\vec{\beta}/\vec{B}]$ and
    $\Tc(\Gamma,x_k^1:\delta_k^1,\ldots,x_k^{n_k}:\delta_k^{n_k}; t_k)
    = \tau_k$ and $\tau = \bigsqcup_{k=1}^n\tau_k$.
  \item
    $\Tc(\Gamma;\fix f : \forall j_1 \ldots j_m . \mu \to \tau . t) =
    \forall j_1 \ldots j_m. \mu \to \tau$ if
    \[
      \Tc(\Gamma,f : \forall j_1 \ldots j_m . \mu^i \to \tau; t)
      \sqsubseteq \forall j_1 \ldots j_m . \mu^{i+1} \to \tau
    \]
    and $i \notin \FSV(\Gamma,\mu,\tau,j_1,\ldots,j_n)$.
  \item $\Tc(\Gamma;\cofix f : \tau . t) = \tau$ if
    \[
    \Tc(\Gamma,f : \chgtgt(\tau,\nu^{\min(s,j)}); t) \sqsubseteq
    \chgtgt(\tau,\nu^{\min(s,j+1)})
    \]
    and $\tgt(\tau) = \nu^s$ and $j \notin \FSV(\Gamma)$ and
    $j \notin \SV(\tau)$.
  \end{itemize}
  In other cases not accounted for by the above points~$\Tc(\Gamma;
  t)$ is undefined. In particular, if the result of the
  operation~$\sqcup$ is not defined then~$\Tc(\Gamma;t)$ is
  undefined. Note that if~$\Tc(\Gamma;t)$ is defined then it is
  uniquely determined.
\end{defi}

\begin{lem}\label{lem_tc_types}
  If $\Tc(\Gamma; t)$ is defined then
  $\Gamma \proves t : \Tc(\Gamma; t)$.
\end{lem}

\begin{proof}
  Induction on the definition of~$\Tc(\Gamma; t)$, using
  Lemma~\ref{lem_meet_bound}, Lemma~\ref{lem_subtype_subst} and
  Lemma~\ref{lem_bar_underline_bound}. We show a few representative
  cases in detail.
  \begin{itemize}
  \item $\Tc(\Gamma,x:\tau; x) = \tau$. Then $\Gamma,x:\tau \proves x
    : \tau$.
  \item $\Tc(\Gamma;c t_1\ldots t_n) = \mu^{\max(s_1,\ldots,s_n)+1}$
    where $\ArgTypes(c) = (\sigma_1,\ldots,\sigma_n)$ and $\mu =
    d_\mu(\tau_1,\ldots,\tau_m)$ and $\Def(c) = d_\mu$ and
    $\Tc(\Gamma;t_i) =
    \sigma_i'[d_\mu^{s_i}(\alpha_1^i,\ldots,\alpha_m^i)/A][\beta_1^i/B_1,\ldots,\beta_m^i/B_m]$
    and $\sigma_i' \sqsubseteq \sigma_i$ and $\tau_j =
    \bigsqcup_{i=1}^n(\alpha_j^i\sqcup\beta_j^i)$. We have $\Gamma
    \proves t_i :
    \sigma_i'[d_\mu^{s_i}(\alpha_1^i,\ldots,\alpha_m^i)/A][\beta_1^i/B_1,\ldots,\beta_m^i/B_m]$
    by the inductive hypothesis. By Lemma~\ref{lem_meet_bound} we have
    $\alpha_j^i \sqsubseteq \tau_j$. Hence, by
    Lemma~\ref{lem_subtype_subst} and the (sub) typing rule, $\Gamma
    \proves t_i :
    \sigma_i[d_\mu^{\max(s_1,\ldots,s_n)}(\tau_1,\ldots,\tau_m)/A][\tau_1/B_1,\ldots,\tau_m/B_m]$. Thus
    $\Gamma \proves c t_1 \ldots t_n : \mu^{\max(s_1,\ldots,s_n)+1}$
    by the (con) typing rule.
  \item
    $\Tc(\Gamma;\case(t;\{c_k\vec{x_k} \To t_k \mid k=1,\ldots,n\})) =
    \tau$ where $\Tc(\Gamma;t) = \mu^{s}$ and $\mu = d(\vec{\beta})$
    and $\ArgTypes(c_k) = (\sigma_k^1,\ldots,\sigma_k^{n_k})$ and
    $\delta_k^l =
    \sigma_k^l[\mu^{\underline{s}}/A][\vec{\beta}/\vec{B}]$ and
    $\Tc(\Gamma,x_k^1:\delta_k^1,\ldots,x_k^{n_k}:\delta_k^{n_k}; t_k)
    = \tau_k$ and $\tau = \bigsqcup_{k=1}^n\tau_k$. By the inductive
    hypothesis $\Gamma \proves t : \mu^s$ and
    $\Gamma,x_k^1:\delta_k^1,\ldots,x_k^{n_k}:\delta_k^{n_k} \proves
    t_k : \tau_k$. Since $s \le \underline{s} + 1$ by
    Lemma~\ref{lem_bar_underline_bound}, we have $\Gamma \proves t :
    \mu^{\underline{s}+1}$ by (sub). By Lemma~\ref{lem_meet_bound} and
    (sub) we have
    $\Gamma,x_k^1:\delta_k^1,\ldots,x_k^{n_k}:\delta_k^{n_k} \proves
    t_k : \tau$. Thus $\Gamma \proves \case(t;\{c_k\vec{x_k} \To t_k
    \mid k=1,\ldots,n\}) : \tau$ by (case).\qedhere
  \end{itemize}
\end{proof}

\begin{lem}\label{lem_subtype_prop}~
  \begin{enumerate}
  \item For any type~$\tau$ we have $\tau \sqsubseteq \tau$.
  \item If $\tau_1 \sqsubseteq \tau_2 \sqsubseteq \tau_3$ then $\tau_1
    \sqsubseteq \tau_3$.
  \end{enumerate}
\end{lem}

\begin{proof}
  By induction. Point (1) is straightforward, using the definition
  of~$\sqsubseteq$. We show a few representative cases for point (2).

  If $\tau_2 = A$ then we must have $\tau_1 = \tau_3 = A$, so $\tau_1
  \sqsubseteq \tau_3$. If $\tau_2 = d_\mu^{s_2}(\vec{\beta})$ then
  $\tau_1 = d_\mu^{s_1}(\vec{\alpha})$ and $\tau_3 =
  d_\mu^{s_3}(\vec{\gamma})$ and $s_1 \le s_2 \le s_3$ and $\alpha_k
  \sqsubseteq \beta_k \sqsubseteq \gamma_k$. By the inductive
  hypothesis $\alpha_k \sqsubseteq \gamma_k$. Hence $\tau_1
  \sqsubseteq \tau_3$. If $\tau_2 = \alpha_2 \to \beta_2$ then $\tau_1
  = \alpha_1 \to \beta_1$ and $\tau_3 = \alpha_3 \to \beta_3$ and
  $\alpha_3 \sqsubseteq \alpha_2 \sqsubseteq \alpha_1$ and $\beta_1
  \sqsubseteq \beta_2 \sqsubseteq \beta_3$. By the inductive
  hypothesis $\alpha_3 \sqsubseteq \alpha_1$ and $\beta_1 \sqsubseteq
  \beta_3$. Thus $\tau_1 \sqsubseteq \tau_3$.
\end{proof}

\begin{lem}\label{lem_meet_lub}~
  \begin{enumerate}
  \item If $\tau_1 \sqcup \tau_2 \sqsubseteq \tau$ then $\tau_1
    \sqsubseteq \tau$ and $\tau_2 \sqsubseteq \tau$.
  \item If $\tau \sqsubseteq \tau_1 \sqcap \tau_2$ then $\tau
    \sqsubseteq \tau_1$ and $\tau \sqsubseteq \tau_2$.
  \end{enumerate}
\end{lem}

\begin{proof}
  We show both points simultaneously by induction on~$\tau$.
  \begin{enumerate}
  \item Assume $\tau_1 = \alpha \to \beta$, $\tau_2 = \alpha' \to
    \beta'$ and $\tau = \gamma_1 \to \gamma_2$. Then $\tau_1 \sqcup
    \tau_2 = (\alpha \sqcap \alpha') \to (\beta \sqcup \beta')$ and
    thus $\gamma_1 \sqsubseteq \alpha \sqcap \alpha'$ and $\beta
    \sqcup \beta' \sqsubseteq \gamma_2$. Hence by the inductive
    hypothesis $\gamma_1 \sqsubseteq \alpha$, $\gamma_1 \sqsubseteq
    \alpha'$, $\beta \sqsubseteq \gamma_2$ and $\beta' \sqsubseteq
    \gamma_2$. This implies $\tau_1 \sqsubseteq \tau$ and $\tau_2
    \sqsubseteq \tau$.

    Assume $\tau_1 = \forall i . \alpha_1$,
    $\tau_2 = \forall i . \alpha_2$ and $\tau = \forall i
    . \gamma$. Then
    $\tau_1 \sqcup \tau_2 = \forall i . \alpha_1 \sqcup
    \alpha_2$. Because $\alpha_1 \sqcup \alpha_2 \sqsubseteq \gamma$,
    by the inductive hypothesis $\alpha_1 \sqsubseteq \gamma$ and
    $\alpha_2 \sqsubseteq \gamma$. Thus $\tau_1 \sqsubseteq \tau$ and
    $\tau_2 \sqsubseteq \tau$.

    Assume $\tau_1 = d_\mu^{s_1}(\vec{\alpha})$ and
    $\tau_2 = d_\mu^{s_2}(\vec{\beta})$ and
    $\tau = d_\mu^s(\vec{\gamma})$. Then
    $\tau_1 \sqcup \tau_2 = d_\mu^{\max(s_1,s_2)}(\vec{\delta})$ with
    $\delta_i = \alpha_i \sqcup \beta_i \sqsubseteq \gamma_i$ and
    $\max(s_1,s_2) \le s$. By the inductive hypothesis
    $\alpha_i \sqsubseteq \gamma_i$ and
    $\beta_i \sqsubseteq \gamma_i$. Also
    $s_1,s_2 \le \max(s_1,s_2) \le s$. Hence $\tau_1 \sqsubseteq \tau$
    and $\tau_2 \sqsubseteq \tau$.

    Assume $\tau_1 = d_\nu^{s_1}(\vec{\alpha})$ and
    $\tau_2 = d_\nu^{s_2}(\vec{\beta})$ and
    $\tau = d_\nu^s(\vec{\gamma})$. Then
    $\tau_1 \sqcup \tau_2 = d_\nu^{\min(s_1,s_2)}(\vec{\delta})$ with
    $\delta_i = \alpha_i \sqcup \beta_i \sqsubseteq \gamma_i$ and
    $\min(s_1,s_2) \ge s$. By the inductive hypothesis
    $\alpha_i \sqsubseteq \gamma_i$ and
    $\beta_i \sqsubseteq \gamma_i$. Also
    $s_1,s_2 \ge \min(s_1,s_2) \ge s$. Hence $\tau_1 \sqsubseteq \tau$
    and $\tau_2 \sqsubseteq \tau$.
  \item The proof for the second point is analogous to the first one.\qedhere
  \end{enumerate}
\end{proof}

\begin{lem}\label{lem_meet_subtype}
  Assume $\tau_1 \sqsubseteq \tau_1'$ and $\tau_2 \sqsubseteq
  \tau_2'$. Then:
  \begin{enumerate}
  \item $\tau_1 \sqcup \tau_2 \sqsubseteq \tau_1' \sqcup \tau_2'$,
  \item $\tau_1 \sqcap \tau_2 \sqsubseteq \tau_1' \sqcap \tau_2'$.
  \end{enumerate}
\end{lem}

\begin{proof}
  We show both points simultaneously by induction on~$\tau_1$.
  \begin{enumerate}
  \item If $\tau_1 = \alpha_1 \to \beta_1$ and $\tau_2 = \alpha_2 \to
    \beta_2$ then $\tau_1' = \alpha_1' \to \beta_1'$ and $\tau_2' =
    \alpha_2' \to \beta_2'$ with $\alpha_1' \sqsubseteq \alpha_1$,
    $\alpha_2' \sqsubseteq \alpha_2$, $\beta_1 \sqsubseteq \beta_1'$
    and $\beta_2 \sqsubseteq \beta_2'$. We have $\tau_1 \sqcup \tau_2
    = (\alpha_1 \sqcap \alpha_2) \to (\beta_1 \sqcup \beta_2)$ and
    $\tau_1' \sqcup \tau_2' = (\alpha_1' \sqcap \alpha_2') \to
    (\beta_1' \sqcup \beta_2')$. By the inductive hypothesis
    $\alpha_1' \sqcap \alpha_2' \sqsubseteq \alpha_1 \sqcap \alpha_2$
    and $\beta_1 \sqcup \beta_2 \sqsubseteq \beta_1' \sqcup
    \beta_2'$. Hence $\tau_1 \sqcup \tau_2 \sqsubseteq \tau_1' \sqcup
    \tau_2'$.

    If $\tau_1 = \forall i . \alpha$ and $\tau_2 = \forall i . \beta$
    then $\tau_1' = \forall i . \alpha'$ and
    $\tau_2' = \forall i . \beta'$ with $\alpha \sqsubseteq \alpha'$
    and $\beta \sqsubseteq \beta'$. By the inductive hypothesis
    $\alpha \sqcup \beta \sqsubseteq \alpha' \sqcup \beta'$. Thus
    $\tau_1 \sqcup \tau_2 \sqsubseteq \tau_1' \sqcup \tau_2'$.

    If $\tau_1 = d_\mu^{s_1}(\vec{\alpha})$ and
    $\tau_2 = d_\mu^{s_2}(\vec{\beta})$ then
    $\tau_1' = d_\mu^{s_1'}(\vec{\alpha}')$ and
    $\tau_2' = d_\mu^{s_2'}(\vec{\beta}')$ with
    $\alpha_i \sqsubseteq \alpha_i'$ and
    $\beta_i \sqsubseteq \beta_i'$ and $s_1 \le s_1'$ and
    $s_2 \le s_2'$. We have
    $\tau_1 \sqcup \tau_2 = d_\mu^{\max(s_1,s_2)}(\vec{\gamma})$ and
    $\tau_1' \sqcup \tau_2' = d_\mu^{\max(s_1',s_2')}(\vec{\gamma}')$,
    where $\gamma_i = \alpha_i \sqcup \beta_i$ and
    $\gamma_i' = \alpha_i' \sqcup \beta_i'$. By the inductive
    hypothesis $\gamma_i \sqsubseteq \gamma_i'$. Also
    $\max(s_1,s_2) \le \max(s_1',s_2')$. Hence
    $\tau_1 \sqcup \tau_2 \sqsubseteq \tau_1' \sqcup \tau_2'$.

    If $\tau_1 = d_\nu^{s_1}(\vec{\alpha})$ and
    $\tau_2 = d_\nu^{s_2}(\vec{\beta})$ then
    $\tau_1' = d_\nu^{s_1'}(\vec{\alpha}')$ and
    $\tau_2' = d_\mu^{s_2'}(\vec{\beta}')$ with
    $\alpha_i \sqsubseteq \alpha_i'$ and
    $\beta_i \sqsubseteq \beta_i'$ and $s_1 \ge s_1'$ and
    $s_2 \ge s_2'$. We have
    $\tau_1 \sqcup \tau_2 = d_\nu^{\min(s_1,s_2)}(\vec{\gamma})$ and
    $\tau_1' \sqcup \tau_2' = d_\mu^{\min(s_1',s_2')}(\vec{\gamma}')$,
    where $\gamma_i = \alpha_i \sqcup \beta_i$ and
    $\gamma_i' = \alpha_i' \sqcup \beta_i'$. By the inductive
    hypothesis $\gamma_i \sqsubseteq \gamma_i'$. Also
    $\min(s_1,s_2) \ge \min(s_1',s_2')$. Hence
    $\tau_1 \sqcup \tau_2 \sqsubseteq \tau_1' \sqcup \tau_2'$.
  \item The proof for the second point is analogous to the first
    point.\qedhere
  \end{enumerate}
\end{proof}

\begin{cor}\label{cor_meet_subtype}
  If $\tau_1 \sqsubseteq \tau$ and $\tau_2 \sqsubseteq \tau$ then
  $\tau_1\sqcup\tau_2\sqsubseteq\tau$.
\end{cor}

\begin{proof}
  One shows by induction that $\tau \sqsubseteq \tau \sqcap \tau$ and
  $\tau \sqcup \tau \sqsubseteq \tau$. Then the corollary follows from
  Lemma~\ref{lem_meet_subtype} and Lemma~\ref{lem_subtype_prop}.
\end{proof}

We write $\Gamma\sqsubseteq\Gamma'$ if
$\Gamma=x_1:\tau_1,\ldots,x_n:\tau_n$ and
$\Gamma'=x_1:\tau_1',\ldots,x_n:\tau_n'$ and $\tau_i \sqsubseteq
\tau_i'$ for $i=1,\ldots,n$.

Note that if $\alpha \sqcup \beta$ is defined and $\alpha' \sqsubseteq
\alpha$ and $\beta' \sqsubseteq \beta$ then $\alpha' \sqcup \beta'$ is
also defined. We will often use this observation implicitly.

\begin{lem}\label{lem_context_tc_subtype}
  If $\Gamma' \sqsubseteq \Gamma$ and $\Tc(\Gamma;t)$ is defined then
  $\Tc(\Gamma';t)$ is defined and $\Tc(\Gamma';t) \sqsubseteq
  \Tc(\Gamma;t)$.
\end{lem}

\begin{proof}
  We proceed by induction on the definition of~$\Tc$.

  If $\Gamma' \sqsubseteq \Gamma$ and $\tau' \sqsubseteq \tau$ then
  $\Tc(\Gamma,x:\tau;t) = \tau \sqsupseteq \tau' =
  \Tc(\Gamma',x:\tau';t)$.

  If $\Gamma' \sqsubseteq \Gamma$ and $\Tc(\Gamma;c t_1\ldots t_n)$ is
  defined, then
  \[
  \Tc(\Gamma;c t_1\ldots t_n) = \mu^{\max(s_1+1,\ldots,s_n+1)}
  \]
  where we have $\ArgTypes(c) = (\sigma_1,\ldots,\sigma_n)$ and $\mu =
  d_\mu(\tau_1,\ldots,\tau_n)$ and $\Def(c) = d_\mu$ and
  \[
    \Tc(\Gamma;t_i) =
    \sigma_i'[d_\mu^{s_i}(\alpha_1^i,\ldots,\alpha_m^i)/A][\beta_1^i/B_1,\ldots,\beta_m^i/B_m]
  \]
  (we take $s_i=0$ and $\alpha_j^i = \bot$ if
  $A \notin \FV(\sigma_i)$, and $\beta_j^i = \bot$ if
  $B_j \notin \FV(\sigma_i)$) and $\sigma_i' \sqsubseteq \sigma_i$ and
  $\tau_j = \bigsqcup_{i=1}^n(\alpha_j^i\sqcup\beta_j^i)$. By the
  inductive hypothesis
  \[
    \Tc(\Gamma';t_i) \sqsubseteq \Tc(\Gamma;t_i) =
    \sigma_i'[d_\mu^{s_i}(\alpha_1^i,\ldots,\alpha_m^i)/A][\beta_1^i/B_1,\ldots,\beta_m^i/B_m].
  \]
  By Lemma~\ref{lem_subtype_subst} we have
  $\Tc(\Gamma';t_i) =
  \rho_i[d_\mu^{s_i'}(\gamma_1^i,\ldots,\gamma_m^i)/A][\delta_1^i/B_1,\ldots,\delta_m^i/B_m]$
  with $\gamma_j^i \sqsubseteq \alpha_j^i$ and
  $\delta_j^i \sqsubseteq \beta_j^i$ and $s_i' \le s_i$. Since
  $\rho_i \sqsubseteq \sigma_i' \sqsubseteq \sigma_i$, by
  Lemma~\ref{lem_subtype_prop} we obtain
  $\rho_i \sqsubseteq \sigma_i$. Let
  $\tau_j' = \bigsqcup_{i=1}^n(\gamma_j^i\sqcup\delta_j^i)$. Thus
  $\Tc(\Gamma';c t_1\ldots t_n) = \mu_1^{\max(s_1'+1,\ldots,s_n'+1)}$
  where $\mu_1=d_\mu(\tau_1',\ldots,\tau_m')$. By
  Lemma~\ref{lem_meet_subtype} we have $\tau_j' \sqsubseteq
  \tau_j$. Also
  $\max(s_1'+1,\ldots,s_n'+1) \le \max(s_1+1,\ldots,s_n+1)$. Hence
  $\Tc(\Gamma';c t_1\ldots t_n) \sqsubseteq \Tc(\Gamma; c t_1\ldots
  t_n)$.

  If $\Gamma' \sqsubseteq \Gamma$ and
  $\Tc(\Gamma;\lambda x : \alpha . t)$ is defined then
  $\Tc(\Gamma;\lambda x : \alpha . t) = \alpha \to \beta$ and
  $\Tc(\Gamma,x:\alpha;t) = \beta$. By the inductive hypothesis
  $\beta' = \Tc(\Gamma',x:\alpha;t) \sqsubseteq \beta$. Hence
  $\Tc(\Gamma'; \lambda x : \alpha . t) = \alpha \to \beta'
  \sqsubseteq \alpha \to \beta = \Tc(\Gamma; \lambda x : \alpha . t)$.

  If $\Gamma' \sqsubseteq \Gamma$ and~$\Tc(\Gamma;tt') = \beta$ then
  $\Tc(\Gamma;t) = \alpha \to \beta$ and
  $\Tc(\Gamma;t') \sqsubseteq \alpha$. By the inductive hypothesis
  $\Tc(\Gamma';t) \sqsubseteq \alpha \to \beta$ and
  $\Tc(\Gamma';t') \sqsubseteq \Tc(\Gamma;t')$. Hence
  $\Tc(\Gamma';t) = \alpha' \to \beta'$ with
  $\alpha \sqsubseteq \alpha'$ and $\beta' \sqsubseteq \beta$. By
  Lemma~\ref{lem_subtype_prop} we have
  $\Tc(\Gamma';t') \sqsubseteq \alpha'$. Hence
  $\Tc(\Gamma';tt') = \beta' \sqsubseteq \beta = \Tc(\Gamma;tt')$.

  If $\Gamma' \sqsubseteq \Gamma$ and $\Tc(\Gamma;ts) = \tau[s/i]$
  then $\Tc(\Gamma;t) = \forall i . \tau$. By the inductive hypothesis
  $\Tc(\Gamma';t) \sqsubseteq \forall i . \tau$, so
  $\Tc(\Gamma';t) = \forall i . \tau'$ with $\tau' \sqsubseteq \tau$.
  Hence
  $\Tc(\Gamma';ts) = \tau'[s/i] \sqsubseteq \tau[s/i] =
  \Tc(\Gamma;ts)$.

  If $\Gamma' \sqsubseteq \Gamma$ and
  $\Tc(\Gamma;\Lambda i . t) = \forall i . \tau$ then
  $\Tc(\Gamma; t) = \tau$ and $i \notin \FSV(\Gamma)$. By the
  inductive hypothesis $\Tc(\Gamma';t) = \tau' \sqsubseteq
  \tau$. Without loss of generality $i \notin \FSV(\tau')$. Then
  $\Tc(\Gamma';\Lambda i . t) = \forall i . \tau' \sqsubseteq
  \Tc(\Gamma;\Lambda i . t)$.

  If $\Gamma' \sqsubseteq \Gamma$ and
  $\Tc(\Gamma;\case(t;\{c_k\vec{x_k} \To t_k \mid k=1,\ldots,n\})) =
  \tau$ and $\Tc(\Gamma;t) = \nu^{s}$ then $\nu = d(\vec{\beta})$ and
  $s \ge 1$ and $\ArgTypes(c_k) = (\sigma_k^1,\ldots,\sigma_k^{n_k})$
  and
  $\delta_k^l = \sigma_k^l[\nu^{\overline{s}}/A][\vec{\beta}/\vec{B}]$
  and
  $\Tc(\Gamma,x_k^1:\delta_k^1,\ldots,x_k^{n_k}:\delta_k^{n_k};t_k) =
  \tau_k$ and $\tau = \bigsqcup_{k=1}^n\tau_k$. By the inductive
  hypothesis $\Tc(\Gamma';t) \sqsubseteq \nu^{s}$. Hence
  $\Tc(\Gamma';t) = \nu_0^{s'} = d^{s'}(\vec{\beta'})$ with
  $\beta_i' \sqsubseteq \beta_i$ and $s' \ge s \ge 1$. Let
  $\gamma_k^l =
  \sigma_k^l[\nu_0^{\overline{s'}}/A][\vec{\beta'}/\vec{B}]$. By
  Lemma~\ref{lem_subtype_subst} and Lemma~\ref{lem_bar_underline_le}
  we have $\gamma_k^l \sqsubseteq \delta_k^l$. So by the inductive
  hypothesis
  $\Tc(\Gamma',x_k^1:\gamma_k^1,\ldots,x_k^{n_k}:\gamma_k^{n_k};t_k) =
  \tau_k' \sqsubseteq \tau_k$. Let
  $\tau'=\bigsqcup_{k=1}^n\tau_k'$. Then
  $\Tc(\Gamma';\case(t;\{c_k\vec{x_k} \To t_k\})) = \tau' \sqsubseteq
  \tau = \Tc(\Gamma;\case(t;\{c_k\vec{x_k} \To t_k\}))$ by
  Lemma~\ref{lem_meet_subtype}.

  Other cases are similar to the ones already considered or follow
  directly from the inductive hypothesis.
\end{proof}

\begin{thm}\label{thm_typing_characterisation}
  $\Gamma \proves t : \tau$ iff $\Gamma \proves t : \Tc(\Gamma;t)$ and
  $\Tc(\Gamma;t) \sqsubseteq \tau$.
\end{thm}

\begin{proof}
  The implication from right to left follows directly from
  definitions. For the other direction we proceed by induction on the
  typing derivation. By Lemma~\ref{lem_tc_types} it suffices to show
  that~$\Tc(\Gamma;t)$ is defined and $\Tc(\Gamma;t) \sqsubseteq
  \tau$.
  \begin{itemize}
  \item If $\Gamma, x : \tau \proves x : \tau$ then
    $\Tc(\Gamma,x:\tau;x) = \tau$.
  \item If $\Gamma \proves t : \tau'$ because of $\Gamma \proves t :
    \tau$ and $\tau \sqsubseteq \tau'$, then by the inductive
    hypothesis $\Gamma \proves t : \Tc(\Gamma;t)$ and $\Tc(\Gamma;t)
    \sqsubseteq \tau$. Hence also $\Tc(\Gamma;t) \sqsubseteq \tau'$ by
    Lemma~\ref{lem_subtype_prop}.
  \item Assume $\Gamma \proves c t_1 \ldots t_n : \rho^{s+1}$ because
    of $\Gamma \proves t_k : \sigma_k[\rho^s/A][\vec{\tau}/\vec{B}]$
    where $\ArgTypes(c)=(\sigma_1,\ldots,\sigma_n)$ and $\Def(c) = d$
    and $\rho = d(\vec{\tau})$. Let $\theta_k=\Tc(\Gamma;t_k)$. By the
    inductive hypothesis $\Gamma \proves t_k : \theta_k$ and
    $\theta_k \sqsubseteq
    \sigma_k[\rho^s/A][\vec{\tau}/\vec{B}]$. Hence by
    Lemma~\ref{lem_subtype_subst} we have
    $\theta_k = \sigma_k'[\rho_k^{s_k}/A][\vec{\beta_k}/\vec{B}]$ with
    $\sigma_k' \sqsubseteq \sigma_k$ and
    $\rho_k^{s_k} \sqsubseteq \rho^s$ and
    $\beta_k^j \sqsubseteq \tau_j$ and $\rho_k = d(\vec{\alpha_k})$
    and $\alpha_k^j \sqsubseteq \tau_j$. We may assume $s_k=0$ and
    $\alpha_k^j = \bot$ if $A \notin \FV(\sigma_k)$, and
    $\beta_k^j = \bot$ if $B_j \notin \FV(\sigma_k)$. Let
    $\tau_j' = \bigsqcup_{k=1}^m(\alpha_k^j\sqcup\beta_k^j)$. Then
    $\Tc(\Gamma;c t_1 \ldots t_n) =
    d^{m(s_1+1,\ldots,s_n+1)}(\vec{\tau'})$ where $m = \max$ if $d$ is
    inductive, and $m = \min$ if $d$ is coinductive. by
    Lemma~\ref{lem_meet_subtype} we have $\tau_i' \sqsubseteq
    \tau_i$. Together with properties of size expressions this implies
    $\Tc(\Gamma;c t_1 \ldots t_n) \sqsubseteq \rho^{s+1}$.
  \item Assume
    $\Gamma \proves (\lambda x : \alpha . t) : \alpha \to \beta$
    because of $\Gamma,x:\alpha \proves t : \beta$. By the inductive
    hypothesis $\beta'=\Tc(\Gamma,x:\alpha)\sqsubseteq\beta$. Then
    $\Tc(\Gamma;\lambda x : \alpha . t) = \alpha \to \beta'
    \sqsubseteq \alpha \to \beta$.
  \item Assume $\Gamma \proves tt' : \beta$ because of
    $\Gamma \proves t : \alpha \to \beta$ and
    $\Gamma \proves t' : \alpha$. By the inductive hypothesis
    $\Tc(\Gamma;t) \sqsubseteq \alpha \to \beta$ and
    $\Tc(\Gamma;t') \sqsubseteq \alpha$. Then
    $\Tc(\Gamma;t) = \alpha' \to \beta'$ with
    $\alpha \sqsubseteq \alpha'$ and $\beta' \sqsubseteq \beta$. We
    have $\Tc(\Gamma;t') \sqsubseteq \alpha'$ by
    Lemma~\ref{lem_subtype_prop}. Hence
    $\Tc(\Gamma;tt') = \beta' \sqsubseteq \beta$.
  \item Assume $\Gamma \proves t s : \tau[s/i]$ because of
    $\Gamma \proves t : \forall i . \tau$. By the inductive hypothesis
    $\Tc(\Gamma;t) \sqsubseteq \forall i . \tau$. Then
    $\Tc(\Gamma;t) = \forall i . \tau'$ with $\tau' \sqsubseteq
    \tau$. Thus $\Tc(\Gamma;ts) = \tau'[s/i] \sqsubseteq \tau[s/i]$.
  \item Assume
    $\Gamma \proves \case(t;\{c_k\vec{x_k}\To t_k\}) : \tau$ because
    of $\Gamma \proves t : \nu^{s+1}$ and
    $\Gamma,x_k^1:\delta_k^1,\ldots,x_k^{n_k}:\delta_k^{n_k}\proves
    t_k : \tau$ and
    $\ArgTypes(c_k) = (\sigma_k^1,\ldots,\sigma_k^{n_k})$ and
    $\delta_k^l=\sigma_k^l[\nu^{s}/A][\vec{\beta}/\vec{B}]$ and
    $\nu=d(\vec{\beta})$. By the inductive hypothesis
    $\Tc(\Gamma;t) \sqsubseteq \nu^{s+1}$. Hence
    $\Tc(\Gamma;t) = d^{s'}(\vec{\beta'})$ with $s' \ge s + 1$ and
    $\beta_i' \sqsubseteq \beta_i$. Let
    $\gamma_k^l=\sigma_k^l[d^{\overline{s'}}(\vec{\tau'})][\vec{\tau'}/\vec{B}]$. By
    Lemma~\ref{lem_bar} we have $\overline{s'} \ge s$. So by
    Lemma~\ref{lem_subtype_subst} we have
    $\gamma_k^l \sqsubseteq \delta_k^l$. By the inductive hypothesis
    $\Tc(\Gamma,
    x_k^1:\delta_k^1,\ldots,x_k^{n_k}:\delta_k^{n_k};t_k)\sqsubseteq\tau$. By
    Lemma~\ref{lem_context_tc_subtype} we have
    $\Tc(\Gamma,x_k^1:\gamma_k^1,\ldots,x_k^{n_k}:\gamma_k^{n_k};t_k)\sqsubseteq
    \Tc(\Gamma,
    x_k^1:\delta_k^1,\ldots,x_k^{n_k}:\delta_k^{n_k};t_k)$, so
    $\Tc(\Gamma,x_k^1:\gamma_k^1,\ldots,x_k^{n_k}:\gamma_k^{n_k};t_k)=\tau_k\sqsubseteq\tau$. Let
    $\tau' = \bigsqcup_{k=1}^n\tau_k$. Then
    $\Tc(\Gamma;\case(t;\{c_k\vec{x_k}\To t_k\})) = \tau'$. By
    Corollary~\ref{cor_meet_subtype} we have $\tau' \sqsubseteq \tau$.
  \item Other cases are analogous to the ones already considered or
    follow directly from the inductive hypothesis.\qedhere
  \end{itemize}
\end{proof}

\subsection{Type checking}

We now show that type checking in~$\lambda^\Diamond$ is
coNP-complete. For this purpose we show how to compute the minimal
type and how to check subtyping.

The size of a type or a size expression is defined in a natural way as
the length of its textual representation. Let~$U$ be a partial finite
function from the set of size variables to the set of size expression
satisfying the \emph{acyclicity condition}: for any choice of
$j_1,\ldots,j_n$ with $j_1 = i$ and $j_{k+1} \in \SV(U(j_k))$ for
$k=1,\ldots,n-1$, we have $j_n \ne i$. In other words, there are no
cycles in the directed graph constructed from~$U$ by postulating an
edge from~$i$ to each $j \in \SV(U(i))$. Let~$S$ be a set of pairs of
size expressions. The size of~$U$ (resp.~$S$) is the sum of the sizes
of all size expressions in the pairs in~$U$ (resp.~$S$). The pair
$(U,S)$ is called a \emph{size constraint}. We say that the size
constraint $(U,S)$ is \emph{valid} if for every valuation~$v$ such
that $v(i) = v(U(i))$ holds for all $i \in \dom(U)$, we have
$v(s_1) \le v(s_2)$ for all $(s_1,s_2) \in S$. We sometimes identify
the function~$U$ with the set of equalities
$\{ i = U(i) \mid i \in \dom(U) \}$.

The purpose of~$U$ is not to express any constraints, but to avoid
duplicating size expressions in the inequalities in~$S$. This is in
order to avoid exponential blow-up in the size of size contraints.

The size of a finite decorated term~$t$ is defined in a natural way,
except that for each occurence of a constant~$c$ in~$t$ we add the
size of~$\ArgTypes(c)$ to the size of~$t$.

For a size expression~$s$, by~$U(s)$ we denote the size
expression~$s'$ obtained from~$s$ by recursively (i.e.~as long as
possible) substituting each free occurence of a size variable
$i \in \dom(U)$ with~$U(i)$. For example, if
$U = \{i_1 = \min(i_2, i_2 + 1), i_2 = s\}$ then
$U(\max(i_1, i_1)) = \max(\min(s, s+1), \min(s, s+1))$. Because of the
acyclicity condition on~$U$ the result of this recursive substitution
process is well-defined. We extend this in the obvious way to types,
terms and contexts. Note that $(U, S)$ is valid iff
$U(s_1) \le U(s_2)$ for all $(s_1, s_2) \in S$.

We now show that it suffices to consider size variable valuations
$v : V_S \to \Nbb$ with the codomain restricted to $\Nbb$.

\begin{lem}\label{lem_codomain}
  If $v(s_1) \le v(s_2)$ for every $v : V_S \to \Nbb \cup \{\infty\}$,
  then $v(s_1) \le v(s_2)$ for every $v : V_S \to \Omega$.
\end{lem}

\begin{proof}
  Assume $v(s_1) > v(s_2)$ for some $v : V_S \to \Omega$. We show how
  to construct $v' : V_S \to \Nbb \cup \{\omega\}$ such that
  $v'(s_1) > v'(s_2)$. Because $\SV(s_1,s_2)$ is finite, there exist
  limit ordinals $\iota_1 < \ldots < \iota_n < \infty$ such that for
  each $i \in \SV(s_1,s_2)$ either $v(i) = \infty$ or there are
  $k,m\in\Nbb$ with $v(i) = \iota_k + m$. Let $M \in \Nbb$ be maximal
  such that $v(i) = \iota_k + M$ for some $i,k$. Let $N$ be the
  maximal nesting of $+1$ in $s_1,s_2$, e.g., for a size expression
  $\max(i+1,j)+1$ we have $N = 2$. Let $j_k = k (M + N + 1)$ for
  $k=1,\ldots,n$. Now it suffices to set $v'(i) = j_k + m$ if
  $v(i) = \iota_k + m$, and $v'(i) = \infty$ if $v(i) = \infty$.
\end{proof}

\begin{cor}\label{cor_codomain}
  A size constraint $(U, S)$ is valid iff for every~$v : V_S \to \Nbb$
  such that $v(i) = v(U(i))$ for $i \in \dom(U)$ we have
  $v(s_1) \le v(s_2)$ for all $(s_1,s_2) \in S$.
\end{cor}

\begin{proof}
  The implication from left to right follows from definitions. The
  other direction follows from Lemma~\ref{lem_codomain} and the fact
  that a non-strict inequality is preserved when taking the limit.
\end{proof}

\begin{lem}\label{lem_size_coNP}
  The problem of checking whether a size constraint~$(U,S)$ is valid
  is in~coNP.
\end{lem}

\begin{proof}
  The complement of the problem may be reduced to the problem of the
  satisfiability of a polynomially large formula in quantifier-free
  Presburger arithmetic, which is
  in~NP~\cite{BoroshTreybing1976,Haase2014}. We proceed with the
  details.

  By Corollary~\ref{cor_codomain} it suffices to consider valuations
  $v : V_s \to \Nbb$ with~$\Nbb$ as codomain.

  Using the identities
  $\infty + 1 = \max(\infty,s) = \max(s,\infty) = \infty$ and
  $\min(\infty,s) = \min(s,\infty) = s$ we may simplify each size
  expression in a linear number of steps to either~$\infty$ or a size
  expression not containing~$\infty$. If $U(i) = \infty$ for some
  $i \in \dom(U)$ then we may substitute $\infty$ for~$i$ in each size
  expression and set $U(i)$ to undefined. We perform these
  simplifications for $(U, S)$ as long as possible, obtaining after a
  polynomial number of steps an equivalent size constraint $(U',S')$
  (i.e.~such that it is valid iff $(U,S)$ is) such that $U'(i)$ does
  not contain~$\infty$ and for each $(s_1,s_2) \in S'$ one of the
  following holds:
  \begin{itemize}
  \item neither $s_1$ nor $s_2$ contain $\infty$,
  \item $s_2 = \infty$ -- then $(s_1,s_2)$ may be removed from $S$
    because $s_1 \le \infty$ always holds,
  \item $s_1 = \infty$ and $s_2$ does not contain~$\infty$ -- then
    $(U',S')$ is not valid, because then $v(s_2) < \infty$ for
    $v : V_S \to \Nbb$.
  \end{itemize}
  Hence, we may assume that none of the size expressions in $(U,S)$
  contains~$\infty$.

  Thus the answer to our decision problem is negative iff there exists
  e.g.~$(s_1,s_2) \in S$ such that
  \[
  i_1 = s_{i_1} \land \ldots \land i_k = s_{i_k} \land s_1 >= s_2 + 1
  \]
  is satisfiable in natural numbers, where the equalities
  $i_l = s_{i_l}$ come from~$U$.

  Using the identities
  \[
    \begin{array}{rcl}
      \max(a,b) + 1 &=& \max(a + 1, b + 1) \\
      \min(a,b) + 1 &=& \min(a + 1, b + 1)
    \end{array}
  \]
  we may further normalize the size expressions so that $\max$ and
  $\min$ never occur within the scope of $+1$.

  Hence, it suffices to show that the satisfiability of conjunctions
  of normalized size expression inequalities is in~NP. However, noting
  that
  \[
  \begin{array}{rcl}
    \min(a,b) \le c &\eqvl& \exists n . (n \ge a \lor n \ge b) \land n
    \le c \\
    c \le \min(a, b) &\eqvl& \exists n . c \le n \land n \le a \land n
    \le b \\
    \max(a, b) \le c &\eqvl& \exists n . a \le n \land b \le n \land n
    \le c \\
    c \le \max(a, b) &\eqvl& \exists n . (n \le a \lor n \le b) \land
    n \le c
  \end{array}
  \]
  this problem may be reduced to satisfiability of a polynomially
  large formula in quantifier-free Presburger arithmetic. The latter
  problem is in~NP~\cite{BoroshTreybing1976,Haase2014}. See also the
  remark at the end of Section~2.2 in~\cite{Haase2014}.
\end{proof}

\begin{lem}\label{lem_subtype_reduction}
  For any types~$\tau_1,\tau_2$ there exists~$S = \Sc(\tau_1,\tau_2)$
  such that for any~$U$ we have: $U(\tau_1) \sqsubseteq U(\tau_2)$ iff
  $(U,S)$ is valid. Moreover, the size of~$S$ is at most polynomial in
  the size of~$\tau_1,\tau_2$.
\end{lem}

\begin{proof}
  Follows by induction on the definition of~$\sqsubseteq$.
\end{proof}

\begin{cor}\label{cor_subtype_coNP}
  Given two types~$\tau_1,\tau_2$ and a partial finite function~$U$
  satisfying the acyclicity condition, checking whether $U(\tau_1)
  \sqsubseteq U(\tau_2)$ is in~coNP.
\end{cor}

\begin{proof}
  Follows from Lemma~\ref{lem_subtype_reduction} and
  Lemma~\ref{lem_size_coNP}.
\end{proof}

\begin{lem}\label{lem_sexp_inequality}
  Given $k \in \Nbb$, a partial finite function~$U$ satisfying the
  acyclicity condition, and a size expression~$s$, it is decidable in
  polynomial time whether $U(s) \ge k$.
\end{lem}

\begin{proof}
  Note that the smallest value of $v(U(s))$ is when $v(i) = 0$ for $i
  \notin \dom(U)$. So it suffices to evaluate $U(s)$ with all size
  variables set to~$0$ and check whether the result is at
  least~$k$. This may be done in polynomial time.
\end{proof}

\begin{lem}\label{lem_triple_computation}
  Given a finite context~$\Gamma$ and a term~$t$, one may compute in
  polynomial time a triple $(U,S,\tau)$ of polynomial size satisfying:
  \begin{itemize}
  \item $(U,S)$ is valid iff $\Tc(\Gamma;t)$ is defined,
  \item if $\Tc(\Gamma;t)$ is defined then $U(\tau) = \Tc(\Gamma;t)$.
  \end{itemize}
\end{lem}

\begin{proof}
  We semi-informally describe an algorithm to compute $(U,S,\tau)$ by
  the following definition of a recursive
  function~$\Tc'(U_0;\Gamma;t)$. To obtain the desired triple one
  takes $U_0=\emptyset$.
  \begin{itemize}
  \item $\Tc'(U_0;\Gamma,x:\tau;x) = (U_0,\emptyset,\tau)$.
  \item $\Tc'(U_0;\Gamma;c t_1\ldots t_n) = (U,S,\theta)$ if
    $\theta = \mu^{\max(s_1,\ldots,s_n)+1}$ and
    $\ArgTypes(c) = (\sigma_1,\ldots,\sigma_n)$ and
    $\mu = d_\mu(\tau_1,\ldots,\tau_m)$ and $\Def(c) = d_\mu$ and
    $\Tc'(U_0;\Gamma;t_i) = (U_i,S_i,\theta_i)$ and
    \[
      \theta_i =
      \sigma_i'[d_\mu^{s_i}(\alpha_1^i,\ldots,\alpha_m^i)/A][\beta_1^i/B_1,\ldots,\beta_m^i/B_m]
    \]
    (we take $s_i=0$ and $\alpha_j^i = \bot$ if $A \notin
    \FV(\sigma_i)$, and $\beta_j^i = \bot$ if $B_j \notin
    \FV(\sigma_i)$; if some~$\theta_i$ does not have the desired form
    then the present case does not apply) and $\tau_j =
    \bigsqcup_{i=1}^n(\alpha_j^i\sqcup\beta_j^i)$
    (if~$\bigsqcup_{i=1}^n(\alpha_j^i\sqcup\beta_j^i)$ is not defined
    then the present case does not apply) and $U = \bigcup_{i=0}^nU_i$
    and $S = \bigcup_{i=1}^nS_i\cup\Sc(\sigma_i',\sigma_i)$.
  \item $\Tc'(U_0;\Gamma;c t_1\ldots t_n) = (U,S,\theta)$ if
    $\theta = \nu^{\min(s_1,\ldots,s_n)+1}$ and
    $\ArgTypes(c) = (\sigma_1,\ldots,\sigma_n)$ and
    $\nu = d_\nu(\tau_1,\ldots,\tau_m)$ and $\Def(c) = d_\nu$ and
    $\Tc'(U_0;\Gamma;t_i) = (U_i,S_i,\theta_i)$ and
    \[
      \theta_i =
      \sigma_i'[d_\nu^{s_i}(\alpha_1^i,\ldots,\alpha_m^i)/A][\beta_1^i/B_1,\ldots,\beta_m^i/B_m]
    \]
    (we take $s_i=0$ and $\alpha_j^i = \bot$ if $A \notin
    \FV(\sigma_i)$, and $\beta_j^i = \bot$ if $B_j \notin
    \FV(\sigma_i)$; if some~$\theta_i$ does not have the desired form
    then the present case does not apply) and $\tau_j =
    \bigsqcup_{i=1}^n(\alpha_j^i\sqcup\beta_j^i)$
    (if~$\bigsqcup_{i=1}^n(\alpha_j^i\sqcup\beta_j^i)$ is not defined
    then the present case does not apply) and $U = \bigcup_{i=0}^nU_i$
    and $S = \bigcup_{i=1}^nS_i\cup\Sc(\sigma_i',\sigma_i)$.
  \item $\Tc'(U_0;\Gamma;\lambda x : \alpha . t) = (U,S,\alpha\to\beta)$
    if $\Tc'(U_0;\Gamma,x:\alpha;t) = (U,S,\beta)$.
  \item $\Tc'(U_0;\Gamma;tt') = (U,S,\beta)$ if
    $\Tc(U_0;\Gamma;t) = (U_1,S_1,\alpha \to \beta)$ and
    $\Tc'(U_0;\Gamma;t') = (U_2,S_2,\alpha')$ and $U=U_1\cup U_2$ and
    $S=S_1 \cup S_2\cup\Sc(\alpha',\alpha)$.
  \item $\Tc'(U_0;\Gamma;ts) = (U\cup\{i=s\},S,\tau)$ if
    $\Tc'(U_0;\Gamma;t) = (U, S, \forall i . \tau)$ with~$i$ fresh.
  \item $\Tc'(U_0;\Gamma;\Lambda i . t) = (U, S, \forall i . \tau)$ if
    $\Tc'(U_0;\Gamma;t) = (U,S,\tau)$ and $i \notin \FSV(\Gamma)$.
  \item
    $\Tc'(U';\Gamma;\case(t;\{c_k\vec{x_k} \To t_k \mid
    k=1,\ldots,n\})) = (U,S,\tau)$ if
    $\Tc'(U';\Gamma;t) = (U_0,S_0,\mu^s)$ and $\mu = d(\vec{\beta})$
    and $\ArgTypes(c_k) = (\sigma_k^1,\ldots,\sigma_k^{n_k})$ and
    $\delta_k^l = \sigma_k^l[\mu^i/A][\vec{\beta}/\vec{B}]$ with~$i$
    fresh and
    $\Tc'(U'\cup\{i=\underline{s}\};\Gamma,x_k^1:\delta_k^1,\ldots,x_k^{n_k}:\delta_k^{n_k};
    t_k) = (U_k,S_k,\tau_k)$ and $\tau = \bigsqcup_{k=1}^n\tau_k$ and
    $U=\bigcup_{i=0}^nU_i$ and $S=\bigcup_{i=0}^nS_i$.
  \item
    $\Tc'(U';\Gamma;\case(t;\{c_k\vec{x_k} \To t_k \mid
    k=1,\ldots,n\})) = (U,S,\tau)$ if
    $\Tc'(U';\Gamma;t) = (U_0,S_0,\nu^s)$ and $\nu = d(\vec{\beta})$
    and $s \ge 1$ and
    $\ArgTypes(c_k) = (\sigma_k^1,\ldots,\sigma_k^{n_k})$ and
    $\delta_k^l = \sigma_k^l[\nu^i/A][\vec{\beta}/\vec{B}]$ with~$i$
    fresh and
    $\Tc'(U'\cup\{i=\overline{s}\};\Gamma,x_k^1:\delta_k^1,\ldots,x_k^{n_k}:\delta_k^{n_k};
    t_k) = (U_k,S_k,\tau_k)$ and $\tau = \bigsqcup_{k=1}^n\tau_k$ and
    $U=\bigcup_{i=0}^nU_i$ and $S=\bigcup_{i=0}^nS_i$.
  \item $\Tc'(U_0;\Gamma;\fix f : (\forall j_1 \ldots j_m . \mu \to \tau) . t) =
    (U,S, \forall j_1 \ldots j_m . \mu \to \tau)$ if
    \[
    \Tc'(U_0;\Gamma,f : \forall j_1 \ldots j_m . \mu^i \to \tau;
    t) = (U,S_0,\theta)
    \]
    and $i \notin \FSV(\Gamma,\mu,\tau,j_1,\ldots,j_m)$ and
    $S = S_0 \cup \Sc(\theta,\forall j_1 \ldots j_m . \mu \to \tau)$.
  \item $\Tc'(U_0;\Gamma;\cofix f : \tau . t) =
    (U,S,\tau)$ if
    \[
    \Tc'(U_0;\Gamma,f : \chgtgt(\tau,\nu^{\min(s,j)}); t) = (U,S_0,\theta)
    \]
    and $\tgt(\tau) = \nu^s$ and $j \notin \FSV(\Gamma)$ and
    $j \notin \SV(\tau)$ and
    $S=S_0\cup\Sc(\theta,\chgtgt(\tau,\nu^{\min(s,j+1)}))$.
  \item Otherwise, if none of the above cases hold, we define
    $\Tc'(U_0;\Gamma;t) = (U_0, \{1 \le 0\}, \bot)$ with $\bot$ and
    arbitrary fixed type.
  \end{itemize}
  First note that if $U' \supseteq U$ where the new size variables
  from $\dom(U')\setminus\dom(U)$ do not occur in~$S$ or~$\tau$ then:
  (1) $(U,S)$ is valid iff $(U',S)$ is valid, and (2)
  $U(\tau) = U'(\tau)$. Note also that when forming a sum
  $U=\bigcup_{i=1}^nU_i$ in the above definition, the function (set of
  equations)~$U$ is well-defined and satisfies the acyclicity
  condition because the left-hand side variable~$i$ in each newly
  added equation $i = s$ is always chosen to be fresh. Using these
  observations one shows by induction that if
  $\Tc'(U';\Gamma;t) = (U,S,\tau)$ then:
  \begin{itemize}
  \item $(U,S)$ is valid iff $\Tc(U'(\Gamma);U'(t))$ is defined,
  \item if $\Tc(U'(\Gamma);U'(t))$ is defined then we have $U(\tau) =
    \Tc(U'(\Gamma);U'(t))$.
  \end{itemize}
  It remains to check that the algorithm implicit in the definition
  of~$\Tc'$ is polynomial. Let~$N$ be the initial size of the input
  (i.e.~the size of~$\Gamma,t$). The total number of calls to~$\Tc'$
  is proportional to~$N$, because in each immediate recursive call in
  the definition of~$\Tc'(U_0;\Gamma;t)$ different disjoint proper
  subterms of~$t$ are given as the third argument.

  In each immediate recursive call the size of the context~$\Gamma$
  grows by at most~$O(N^2)$. Indeed, there are essentially two
  possibilities of what we add to the context~$\Gamma$.
  \begin{enumerate}
  \item We add~$x : \alpha$ for the case of lambda abstraction
    $\lambda x : \alpha . t'$. Then $\alpha$ occurs in the original
    term~$t$, so the size of~$\Gamma$ grows by at most~$N$.
  \item We add e.g.~$x_k^1 : \delta_k^1,\ldots,x_k^{n_k} :
    \delta_k^{n_k}$ where $\delta_k^j =
    \sigma_k^j[\mu^i/A][\vec{\beta}/B]$ or the case of a
    $\case$-term. Then $\mu$ and~$\vec{\beta}$ occur in the original
    term~$t$, and $\sigma_k^j$ is an argument type for the
    constructor~$c_k$ which occurs in~$t$ (so the size of~$\sigma_k^j$
    counts towards the size of~$t$). Hence the total size of
    $\sigma_k^1,\ldots,\sigma_k^{n_k}$ is~$\le N$, and thus so is the
    total number of occurences of $A,\vec{B}$ in
    $\sigma_k^1,\ldots,\sigma_k^{n_k}$. The size of each of
    $\mu,\vec{\beta}$ is~$\le N$. Therefore, the total size of
    $\delta_k^1,\ldots,\delta_k^{n_k}$ is at most~$N^2 + N$.
  \end{enumerate}
  Hence, the size of the context at any given call to~$\Tc'$ (during
  the whole run of the algorithm) is at most~$O(N^3)$. Let
  $(U,S,\tau)$ denote the result of calling $\Tc'$. At the leaves of
  the computation tree (i.e.~when there are no more immediate
  recursive calls) the type~$\tau$ is taken from the context, so its
  size is at most~$O(N^3)$. At internal nodes, the size of~$\tau$ is
  equal to at most the sum of sizes of the types returned by immediate
  recursive calls (note that the size of~$\alpha \sqcup \beta$ is
  equal to at most the sum of sizes of~$\alpha$ and~$\beta$ plus a constant), plus
  possibly the size of a type occuring in~$t$ (which is~$\le N$), plus
  possibly~$O(N)$. Hence, each call to~$\Tc'$ contributes at
  most~$O(N^3)$ towards the size of the final result type. Since there
  are~$O(N)$ calls in total, for any given call the result type~$\tau$
  of this call has size at most~$O(N^4)$. Now we count the final size
  of~$S$. At the leaves of the computation tree~$S = \emptyset$, and
  at each internal node we add at most~$O(N)$ sets~$\Sc(\alpha,\beta)$
  where each of $\alpha,\beta$ is either a subtype of a type returned
  by an immediate recursive call or of the term~$t$. So the size of
  $\alpha,\beta$ is polynomial in~$N$, and thus so is the size
  of~$\Sc(\alpha,\beta)$ by Lemma~\ref{lem_subtype_reduction}. Hence,
  the total final size of~$S$ is polynomial in~$N$. To count the total
  final size of~$U$, note that we may consider it to be a mutable
  global variable which at each call is modified by adding at most one
  equation of polynomial size (because the right-hand side size
  expression has size proportional to the size of a size expression
  occuring in a type returned by one of the immediate recursive
  calls). Thus the total final size of~$U$ is polynomial in~$N$.

  We have thus shown that the computed triple $(U,S,\tau)$ has
  polynomial size. Note that in each of the calls to~$\Tc'$, the
  computation time (not counting the immediate recursive calls) is
  proportional to the size of the returned triple, and is thus
  polynomial (we need Lemma~\ref{lem_sexp_inequality} to decide in
  polynomial time if $U(s_0) \ge 1$ in the third-last point in the
  definition of~$\Tc'$). Hence, the whole running time is polynomial.
\end{proof}

{ \renewcommand{\thethm}{\ref{thm_type_checking}}
\begin{thm}
  Type checking in the system~$\lambda^\Diamond$ is
  coNP-complete. More precisely, given $\Gamma,t,\tau$ the problem of
  checking whether $\Gamma \proves t : \tau$ is coNP-complete.
\end{thm}
\addtocounter{thm}{-1}}

\begin{proof}
  It follows from Theorem~\ref{thm_typing_characterisation},
  Lemma~\ref{lem_triple_computation}, Lemma~\ref{lem_size_coNP} and
  Corollary~\ref{cor_subtype_coNP} that the problem is in~coNP.

  To show that the problem is coNP-hard we reduce the problem of
  unsatisfiability of 3-CNF boolean formulas, which is coNP-hard. We
  show how to construct in polynomial time an inequality $s_1 \le s_2$
  of size expressions which is equisatisfiable with a given 3-CNF
  boolean formula~$\varphi$. For concreteness assume~$\varphi$ is
  \[
  (x \lor \neg y \lor z) \land (x \lor \neg z \lor y).
  \]
  This formula is translated to the inequality $s_1 \le s_2$ where
  \[
  \begin{array}{l}
    s_1 = \max(\min(x,\bar{y},z)+1,\min(x,\bar{z},y)+1,1,\\
    \quad\quad\quad\quad\quad \min(x,\bar{x})+1,\min(y,\bar{y})+1,\min(z,\bar{z})+1)\\
    s_2 = \min(1, \max(x,\bar{x}), \max(y,\bar{y}), \max(z,\bar{z}))
  \end{array}
  \]
  and $\bar{x},\bar{y},\bar{z}$ are fresh variables intended to
  represent the negations of $x,y,z$ respectively.

  Let~$v$ with $\codom(v) = \{\top,\bot\}$ be a satifying valuation
  for~$\varphi$. Define $\bar{v}$ with $\codom(\bar{v}) = \{0,1\}$ by
  $\bar{v}(i) = 0$ if $v(i) = \top$, and $\bar{v}(i) = 1$ if
  $v(i) = \bot$, and $\bar{v}(\bar{i}) = 1 - v(i)$, for any
  variable~$i$. Then
  $\bar{v}(\max(i,\bar{i})) = \bar{v}(\min(i,\bar{i})+1) = 1$ for any
  variable~$i$, and
  $\bar{v}(\min(x,\bar{y},z)) = \bar{v}(\min(x,\bar{z},y)) = 0$. Hence
  $\bar{v}(s_1) \le \bar{v}(s_2)$.

  Let~$\bar{v}$ be a satisfying valuation for $s_1 \le s_2$. The
  inequality $s_1 \le s_2$ is equivalent to the following conjunction
  of inequalities:
  \[
  \begin{array}{l}
    \min(x,\bar{y},z)+1 \le 1 \land \min(x,\bar{y},z)+1 \le \max(x,\bar{x})
    \land {} \\ \quad \min(x,\bar{y},z)+1 \le \max(y, \bar{y}) \land
    \min(x,\bar{y},z)+1 \le \max(z, \bar{z}) \land {} \\
    \min(x,\bar{z},y)+1 \le 1 \land \min(x,\bar{z},y)+1 \le \max(x, \bar{x})
    \land {} \\ \quad \min(x,\bar{z},y)+1 \le \max(y, \bar{y}) \land
    \min(x,\bar{z},y)+1 \le \max(z, \bar{z}) \land {} \\
    1 \le 1 \land 1 \le \max(x,\bar{x}) \land 1 \le \max(y,\bar{y}) \land 1 \le
    \max(z,\bar{z}) \land {} \\
    \min(x,\bar{x})+1 \le 1 \land \min(x,\bar{x}) + 1 \le \max(x,\bar{x}) \land \min(x,\bar{x})+1 \le
    \max(y,\bar{y}) \land {} \\ \quad \min(x,\bar{x})+1 \le \max(z,\bar{z}) \land {} \\
    \min(y,\bar{y})+1 \le 1 \land \min(y,\bar{y})+1 \le \max(x,\bar{x}) \land \min(y,\bar{y})+1 \le
    \max(y,\bar{y}) \land {} \\ \quad \min(y,\bar{y})+1 \le \max(z,\bar{z}) \land {} \\
    \min(z,\bar{z})+1 \le 1 \land \min(z,\bar{z})+1 \le \max(x,\bar{x}) \land \min(z,\bar{z})+1 \le
    \max(y,\bar{y}) \land {} \\ \quad \min(z,\bar{z})+1 \le \max(z,\bar{z}).
  \end{array}
  \]
  For each variable~$i$, since $1 \le \max(i,\bar{i})$ and
  $\min(i,\bar{i})+1 \le 1$ occur in this conjunction, we conclude
  that exactly one of $\bar{v}(i), \bar{v}(\bar{i})$ is zero and the
  other one is nonzero. Define $v$ with $\codom(v)=\{\top,\bot\}$ by
  $v(i) = \top$ if $\bar{v}(i) = 0$, and $v(i) = \bot$ if
  $\bar{v}(i) \ne 0$. We have
  $\bar{v}(\min(x,\bar{y},z)) = \bar{v}(\min(x,\bar{z},y)) = 0$
  because of the inequalities $\min(x,\bar{y},z)+1 \le 1$ and
  $\min(x,\bar{z},y)+1 \le 1$. This implies that~$v$ is a satisfying
  valuation for~$\varphi$.

  Hence for every 3-CNF boolean formula~$\varphi$ there exists an
  equisatisfiable inequality $s_1 \le s_2$ of size expressions which
  may be computed in polynomial time. So~$\varphi$ is unsatisfiable
  iff $s_1 > s_2$ is valid. Because $v(s_2) \ne \infty$ for any
  valuation~$v$, the inequality $s_1 > s_2$ is equivalent
  to~$s_1 \ge s_2 + 1$. Therefore~$\varphi$ is unsatisfiable iff
  $x : \nu^{s_1}, f : \nu^{s_2 + 1} \to \mu^0 \proves f x : \mu^0$ for
  some fixed~$\nu,\mu$.
\end{proof}

\begin{rem}
  The use of the set of equations~$U$ is necessary to avoid an
  exponential blow-up in the size of size constraints. For instance,
  consider $\Gamma = f : \forall i . \mu^i \to \mu^i$ and define
  $t_0 = f$, $t_{n+1} = \Lambda i . t_n \max(i,i)$. Let $s_0 = i$ and
  $s_{n+1} = \max(s_n,s_n)$. We have
  $\Tc(\Gamma; t_n) = \forall i . \mu^{s_n} \to \mu^{s_n}$. The size
  of~$s_n$ is proportional to~$2^n$ while the size of~$t_n$ is
  proportional to~$n$.

  Similarly, suppose~$\mu$ has two constructors $c_1 : \mu \to \mu$
  and $c_2 : \mu \to \mu$. Let $t_0 = x$ and
  $t_{n+1} = \case(t_n; \{ c_1 y \To y, c_2 y \To y \})$. Let
  $s_n = 0 + 1 + 1 + \ldots + 1$ where~$1$ occurs~$n$ times. By
  induction on~$n$ one shows $\Tc(x : \mu^{s_n}; t_n) = \mu^{s_n'}$
  where $s_0' = 0$ and $s_{n+1}' = \max(s_n',s_n')$. The size
  of~$s_n'$ is proportional to~$2^n$, while the sizes of~$t_n,s_n$ are
  proportional to~$n$.
\end{rem}

\end{document}